\documentclass[11pt]{article}
\usepackage{appendix}
\usepackage{amsmath}
\usepackage{amssymb}
\usepackage{amsthm}
\usepackage{ifpdf}
\usepackage{algorithmic}
\usepackage{wrapfig}
\usepackage{cite}
\usepackage{colortbl}
\usepackage{tikz}
\usepackage{morefloats}
\usepackage{enumitem}
\usepackage{caption}
\usepackage{subcaption}
\usepackage{authblk}
\usepackage{url}
\usepackage{hyperref}
\usepackage[letterpaper, portrait, margin=1.0in]{geometry}

\usepackage[ruled, vlined, linesnumbered]{algorithm2e}

\SetCommentSty{mycommfont}

\newcommand{\calS}{\mathcal{S}}
\newcommand{\calUT}{\mathcal{U}_{\mathcal{T}}}
\newcommand{\sigmaT}{\sigma_{\mathcal{T}}}

\newcommand{\adderArray}{\texttt{Adder Array}}
\newcommand{\adderUnit}{\texttt{Adder Unit}}
\newcommand{\componentAdder}{\texttt{Component Adder}}
\newcommand{\partialAdder}{\texttt{Partial Adder}}
\newcommand{\bracket}{\texttt{Bracket}}
\newcommand{\externalCommunication}{\texttt{External Communication}}

\newcommand{\genome}{\texttt{Genome}}
\newcommand{\bracketBlocker}{\texttt{Bracket Blocker}}

\newcommand{\greenoff}{\texttt{green-off}}
\newcommand{\greenon}{\texttt{green-on}}
\newcommand{\yellowoff}{\texttt{yellow-off}}
\newcommand{\yellowup}{\texttt{yellow-up}}
\newcommand{\yellowdown}{\texttt{yellow-down}}

\def\planter/{{\texttt{planter}}}
\def\leftcomp/{{\texttt{left-comp}}}
\def\rightcomp/{{\texttt{right-comp}}}
\def\leftarm/{{\texttt{left-arm}}}
\def\rightarm/{{\texttt{right-arm}}}
\def\crashingcounter/{{\texttt{crashing-counter}}}

\usepackage[textsize=scriptsize]{todonotes}

\usepackage{commands-tam}
\newcommand{\fullgridgraph}{G^\mathrm{f}}\newcommand{\bindinggraph}{G^\mathrm{b}}

\ifpdf

\usepackage[pdftex]{epsfig}

\else

\usepackage[dvips]{epsfig}
\newcommand{\href}[2]{#2}

\fi

\theoremstyle{definition}
\newcounter{thm-count}
\newcounter{lem-count}
\newcounter{def-count}
\newcounter{cor-count}
\newcounter{obs-count}
\newtheorem{theorem}[thm-count]{Theorem}
\newtheorem{lemma}[lem-count]{Lemma}
\newtheorem{definition}[def-count]{Definition}
\newtheorem{observation}[obs-count]{Observation}

\title{The Impacts of Dimensionality, Diffusion, and Directedness on Intrinsic Universality in the abstract Tile Assembly Model}

\author[1]{Daniel Hader}

\author[2]{Aaron Koch}

\author[3]{Matthew J. Patitz}

\author[4]{Michael Sharp}

\affil[ ]{\small{\textit{${}^1$dhader@uark.edu, ${}^2$aekoch@uark.edu, ${}^3$patitz@uark.edu, ${}^4$mrs018@uark.edu}}}
\affil[ ]{Department of Computer Science and Computer Engineering, University of Arkansas, Fayetteville, AR, USA. This research was supported in part by National Science Foundation Grants CCF-1422152 and CAREER-1553166.}

\begin{document}

\date{}
\maketitle

\begin{abstract}
In this paper we present a series of results related to mathematical models of self-assembling systems of tiles and the impacts that three diverse properties have on their dynamics. In these self-assembling systems, initially unorganized collections of tiles undergo random motion and can bind together, if they collide and enough of their incident glues match, to form assemblies. Here we greatly expand upon a series of prior results which showed that (1) the abstract Tile Assembly Model (aTAM) is intrinsically universal (FOCS 2012), and (2) the class of directed aTAM systems is not intrinsically universal (FOCS 2016). Intrinsic universality (IU) for a model (or class of systems within a model) means that there is a universal tile set which can be used to simulate an arbitrary system within that model (or class). Furthermore, the simulation must not only produce the same resultant structures, it must also maintain the full dynamics of the systems being simulated and display the same behaviors modulo a scale factor.
While the FOCS 2012 result showed that the standard, two-dimensional (2D) aTAM is IU, here we show that this is also the case for the three-dimensional (3D) version. Conversely, the FOCS 2016 result showed that the class of aTAM systems which are directed (a.k.a. deterministic, or confluent) is not IU, meaning that there is no universal simulator which can simulate directed aTAM systems while itself always remaining directed, implying that nondeterminism is fundamentally required for such simulations. Here, however, we show that in 3D the class of directed aTAM systems is actually IU, i.e. there is a universal directed simulator for them. This implies that the constraint of tiles binding only in the plane forced the necessity of nondeterminism for the simulation of 2D directed systems. This then leads us to continue to explore the impacts of dimensionality and directedness on simulation of tile-based self-assembling systems by considering the influence of more rigid notions of dimensionality. Namely, we introduce the Planar aTAM, where tiles are not only restricted to binding in the plane, but they are also restricted to traveling within the plane, and we prove that the Planar aTAM is not IU, and prove that the class of directed systems within the Planar aTAM also is not IU. Finally, analogous to the Planar aTAM, we introduce the Spatial aTAM, its 3D counterpart, and prove that the Spatial aTAM is IU.

This paper adds to a broad set of results which have been used to classify and compare the relative powers of differing models and classes of self-assembling systems, and also helps to further the understanding of the roles of dimension and nondeterminism on the dynamics of self-assembling systems. Furthermore, to prove our positive results we have not only designed, but also implemented what we believe to be the first IU tile set ever implemented and simulated in any tile assembly model, and have made it, along with a simulator which can demonstrate it, freely available.
\end{abstract}

\thispagestyle{empty}
\clearpage
\setcounter{page}{1}

\section{Introduction}

Self-assembling systems create structure from randomness, utilizing only local interactions between components which begin in disorganized collections but randomly mix and collide with each other, possibly binding when allowed by those local interactions.  Natural self-assembling systems abound (e.g. the formation of the crystalline structure of snowflakes or the autonomous combination of the proteins composing a virus) and, inspired by the complexity they generate, researchers have sought to model them and to create novel self-assembling systems. This has led to an impressive variety of, among other things, DNA-based self-assembled creations (e.g. \cite{Douglas2009,ke2012three,rothemund2004algorithmic,RothOrigami,evans2014crystals,SignalTilesExperimental,WinfreeDNARobots2010,SeemanDNARobots2010,SeemanMaoLaBeanReif,drmaurdsa}).  It has also led to a variety of mathematical models based on components of different sizes and shapes (e.g. \cite{Winf98,GeoTiles,Polyominoes,Polygons,BeckerTimeOpt,OneTile,jDuples}) as well as a diverse set of dynamics (e.g.  \cite{Winf98,AGKS05g,FTAM,jSignals,SingleNegative,JonoskaSignals1,FlexibleCompModel,jRTAM,MLR07}).  An important trait of nearly all of these models is that they are capable of \emph{algorithmic self-assembly}, in which systems are able to create assemblies that represent computations, following embedded algorithms via the rule-based combination of their constituent components. In fact, even the first and simplest of these models, the abstract Tile Assembly Model (aTAM)\cite{Winf98}, is capable of Turing universal computation.  Given the powerful computational potential of systems in these models, and the variety between component geometries and dynamics, initially it was difficult to compare the relative powers and limitations between them, and direct comparisons were often piecemeal (e.g. \cite{Versus}).  Fortunately, with the incorporation of a tool used within the domain of cellular automata, namely \emph{intrinsic universality} \cite{Ollinger-CSP08,Goles-etal-2011}, such comparisons became possible.  In \cite{IUSA} is was shown that the 2D aTAM is intrinsically universal (IU), which means that there exists a constant-sized tile set, among the infinite collection of tile sets, which is capable of simulating any of the infinite systems within the 2D aTAM.  Furthermore, this simulation preserves the full dynamics and geometry of the original system, following the exact same assembly processes, modulo only a scaling factor such that constant-sized regions of tiles in the simulating system represent individual tiles of the original system.

In \cite{WoodsIUSurvey}, Woods describes the formation of a ``kind of computational complexity theory for self-assembly'' which utilizes the notion of intrinsic universality. The IU concept has been used to characterize the relative powers of many combinations of models and systems with differing parameters (e.g. \cite{2HAMIU,DirectedNotIU,jDuples,Signals3D,IUNeedsCoop,TempOneNotIU,OneTile,j2HAMSim,jBreakableDuples}). For instance, results have shown that allowing tiles with more complex geometries can dramatically reduce the number of unique tile types required for universal simulation and computations - to only \emph{one} if flipping and rotation of tiles are allowed \cite{OneTile}, or that the dynamics of hierarchical assembly models (in which arbitrarily large assemblies can combine in pairs) make the binding threshold (a.k.a. temperature parameter) a crucial and separating factor in the ability of one system to simulate another \cite{2HAMIU,j2HAMSim}.

While these directions provide interesting explorations of the ways in which permuting the properties of self-assembly models affects their relative powers, fundamental questions remain in even the simplest models in relation to nondeterminism and dimensionality.  In \cite{TempOneNotIU} it was shown that the set of non-cooperative 2D aTAM systems (in which a tile can attach to a growing assembly if it binds with only a single other tile in the assembly, as opposed to requiring two or more bindings), a.k.a. temperature-1 self-assembling systems, are not intrinsically universal or capable of bounded Turing machine computation, while \cite{CookFuSch11} showed that non-cooperative but ``just barely 3D'' systems (which are those that require only two planes) are, in fact, capable of deterministic Turing universal computation. In \cite{IUNeedsCoop} the authors showed that cooperative, i.e. temperature-2, self-assembly is required for intrinsic universality of both 2D and 3D classes of systems, which proved that the initial 2D aTAM IU construction of \cite{IUSA} at temperature 2 was optimal with respect to the temperature parameter. However, in the construction of that proof there was built-in nondeterminism in the form of many locations of the simulating systems which were forced to nondeterministically allow for the selection of which tile type may appear in those locations.  A self-assembling system is called \emph{directed} if, irrespective of the (valid) assembly path which it follows, the exact same final assembly results, meaning that the final assembly is always identical in shape and in the types of tiles located in each position.  The construction of \cite{IUSA} was forced, due to that nondeterminism, to result in simulating systems which were not directed even when they were simulating directed systems. The result of \cite{DirectedNotIU} proved that for the 2D aTAM, such nondeterminism was actually fundamental and unavoidable.  The question then remained as to whether this nondeterminism is a by-product of the dynamics of the aTAM itself, or instead is caused by the planarity of the 2D aTAM, i.e. the fact that assemblies must be embedded in the plane and tiles are not able to grow over other tiles.

\subsection{Our results}

In this paper, our results extend what is known about the effects of the interplay between dimensionality and directedness on intrinsic universality in self-assembly. However, to better understand the impacts of embedding systems within different dimensions, we introduce new variants of the aTAM in which tiles are not only restricted to attaching within regular lattices of the correct dimension, they are also required to travel only within those dimensions. We consider such \emph{diffusion restricted} models and call them the \emph{Linear aTAM}, \emph{Planar aTAM}, and \emph{Spatial aTAM} in 1D, 2D, and 3D, respectively. In these models, new tiles are only allowed to attach to locations on the perimeters of assemblies if they can diffuse into them along collision-free paths beginning from infinitely far away, i.e. they cannot be blocked from diffusing into those locations by tiles already attached to the assembly. As an example, in the Planar aTAM the tiles forming a 2D square which fully surrounds an empty central location prevent the diffusion of any tile into that central location.

We first make some relatively straightforward observations about intrinsic universality in the 1D aTAM, where tiles are restricted to forming linear assemblies. Section~\ref{sec:1D} contains more details about these observations, but to summarize, it is easy to show that the 1D aTAM is not IU. This is because any universal tile set $U$ must have a fixed number of tile types, say $|U| = t$. However, it is simple to define a 1D system with greater than $t$, say $t+1$, unique tile types, where that system forms a $t+1$ length line. Any system using $U$ must ``pump'' after forming a line of length $t$, meaning that a tile type must be repeated and the segment between the repeats could grow an infinite number of times. This would not be a valid simulation of the system which only made a finite line. Since the system failing to be simulated is also directed, the class of directed 1D aTAM systems is not IU. Finally, since tiles already attached to a linear assembly cannot block the ability of other tiles to bind to either end (the only possible frontier locations), diffusion can't actually be restricted so the dynamics are the same as for the regular 1D aTAM.

We next turn to 2D, noting that the standard 2D aTAM isn't entirely restricted to two dimensions, since tiles are allowed to diffuse into attachment locations through 3D space. We elucidate how that impacts the intrinsic universality of the aTAM. We prove that although the standard aTAM is IU\cite{IUSA}, the Planar aTAM is not IU (Theorem~\ref{thm:PaTAM-not-IU}), which means that the restriction of keeping tiles in the plane is too restrictive for a universal simulator to exist. To complete the results in 2D, we explore the combination of the diffusion constraint and directedness. Specifically, we prove that the class of directed systems in the Planar aTAM is not IU (Theorem~\ref{thm:DP2DaTAM-not-IU}). Thus, the combination of the two restrictions on tile assembly systems does not result in dynamics which allow for universal simulation.

We then move to 3D, proving that the 3D aTAM is IU (Theorem~\ref{thm:3DaTAMIU}), and present a universal tile set $U$ along with an algorithm to create necessary seed assemblies for $U$ to simulate arbitrary 3D aTAM systems.  We next show that, due to the careful design of $U$, $U$ is also an intrinsically universal tile set for the set of directed 3D aTAM systems (Theorem~\ref{thm:directed3DIU}), which means that when $U$ is used to simulate a directed 3D aTAM system, the simulating system itself remains directed. Thus we prove that the necessity of nondeterminism proven in \cite{DirectedNotIU} is a result of the 2D aTAM being limited to the plane. Finally, we prove that the Spatial aTAM is IU (Theorem~\ref{thm:spatial-IU}), contrasting with our result showing that the Planar aTAM is not. The one remaining combination, that of directed classes of the Spatial aTAM, we conjecture to not be IU.

\begin{table}[ht]
\centering
\begin{tabular}{l|c|c|c|c|}
\cline{2-5}
& General & Directed & \begin{tabular}[c]{@{}c@{}}Diffusion\\ restricted\end{tabular} & \begin{tabular}[c]{@{}c@{}}Directed +\\ diffusion restricted\end{tabular} \\ \hline

\multicolumn{1}{|l|}{1D} & \begin{tabular}[c]{@{}c@{}}not IU\\ (Obs. \ref{obs:1D-not-IU})\end{tabular} & \begin{tabular}[c]{@{}c@{}}not IU\\ (Obs. \ref{obs:1D-directed-not-IU})\end{tabular} & \begin{tabular}[c]{@{}c@{}}not IU\\ (Obs. \ref{obs:1D-linear-not-IU})\end{tabular} & \begin{tabular}[c]{@{}c@{}}not IU\\ (Obs. \ref{obs:1D-directed linear-not-IU})\end{tabular} \\ \hline

\multicolumn{1}{|l|}{2D} & IU\cite{IUSA} & not IU\cite{DirectedNotIU} & \begin{tabular}[c]{@{}c@{}}not IU\\ (Thm. \ref{thm:PaTAM-not-IU})\end{tabular} & \begin{tabular}[c]{@{}c@{}}not IU\\ (Thm. \ref{thm:DP2DaTAM-not-IU})\end{tabular} \\ \hline

\multicolumn{1}{|l|}{3D} & \begin{tabular}[c]{@{}c@{}}IU\\ (Thm. \ref{thm:3DaTAMIU})\end{tabular} & \begin{tabular}[c]{@{}c@{}}IU\\ (Thm. \ref{thm:directed3DIU})\end{tabular} & \begin{tabular}[c]{@{}c@{}}IU\\ (Thm. \ref{thm:spatial-IU})\end{tabular} & Conj. not IU \\ \hline
\end{tabular}
\end{table}

\subsection{Implemented IU tile set} \label{sec:implementation}

Due to the complexity of IU tile sets, which are capable of universally simulating entire classes of systems, as far as we know, no IU tile set has ever been explicitly defined down to the individual tile level.  Rather, they have been logically described at high levels of abstraction.  We believe that our IU tile set is the first, in any model of self-assembly, to be explicitly generated and tested.  We developed a set of Python scripts to design, generate, and test each component of our construction.  We combined the tile sets for each component into our universal tile set with approximately 152,000 unique tile types, also making this what we believe to be the most complex aTAM system which has ever been fully developed.  We explicitly defined the algorithm and wrote the code required to take as input an arbitrary 3D aTAM system and generate the seed assembly which uses our IU tile set to generate the initial (seed) assembly required for the tile set to simulate the input system.\footnote{Our implementation omits two relatively trivial components which do not impact the correctness of the 3D aTAM IU simulations, but which are fully designed and described in the following text.} The tile set and scripts used to test it, along with images and videos of examples, are freely available online at \url{http://self-assembly.net/wiki/index.php?title=Intrinsic_Universality_of_the_aTAM#3D}. Additionally, we developed the 2D/3D aTAM simulator PyTAS that is specifically optimized to efficiently handle loading, simulation, and rendering of 3D systems consisting of several millions of tiles, and was used extensively to test this construction.  PyTAS is freely available online at \url{http://self-assembly.net/wiki/index.php?title=PyTAS}.

\section{Preliminaries}

In this section, we present definitions for the models and concepts used throughout the paper.

\subsection{Informal description of the abstract Tile Assembly Model}
\label{sec-tam-informal}

This section gives a brief informal sketch of the abstract Tile Assembly Model (aTAM).
See Section~\ref{sec-tam-formal} for formal definitions. 
Here, we define the 2D aTAM, whereas in Section~\ref{sec-tam-formal} we formulate the $d$-dimensional aTAM. For 
notational convenience, throughout this paper  the term ``aTAM'' refers to the 2D aTAM.

A \emph{tile type} is a unit square with four sides, each consisting of a \emph{glue label}, often represented as a finite string, and a nonnegative integer \emph{strength}. A glue~$g$ that appears on multiple tiles (or sides) always has the same strength~$s_g \in \{0,1,2,\ldots \}$. 
There are a finite set $T$ of tile types, but an infinite number of copies of each tile type, with each copy being referred to as a \emph{tile}. An \emph{assembly}
is a positioning of tiles on the integer lattice $\Z^2$, described  formally as a partial function $\alpha:\Z^2 \dashrightarrow T$. 
Let $\mathcal{A}^T$ denote the set of all assemblies of tiles from $T$, and let $\mathcal{A}^T_{< \infty}$ denote the set of finite assemblies of tiles from $T$.
We write $\alpha \sqsubseteq \beta$ to denote that $\alpha$ is a \emph{subassembly} of $\beta$, which means that $\dom\alpha \subseteq \dom\beta$ and $\alpha(p)=\beta(p)$ for all points $p\in\dom\alpha$.
Two adjacent tiles in an assembly \emph{interact}, or are \emph{attached}, if the glue labels on their abutting sides are equal and have positive strength. 
Each assembly induces a \emph{binding graph}, a grid graph whose vertices are tiles, with an edge between two tiles if they interact.
The assembly is \emph{$\tau$-stable} if every cut of its binding graph has strength at least~$\tau$, where the strength   of a cut is the sum of all of the individual glue strengths in~the~cut.

A \emph{tile assembly system} (TAS) is a triple $\calT = (T,\sigma,\tau)$, where $T$ is a finite set of tile types, $\sigma:\Z^2 \dashrightarrow T$ is a finite, $\tau$-stable \emph{seed assembly},
and $\tau$ is the \emph{temperature}.
An assembly $\alpha$ is \emph{producible} if either $\alpha = \sigma$ or if $\beta$ is a producible assembly and $\alpha$ can be obtained from $\beta$ by the stable binding of a single tile.
In this case we write $\beta\to_1^\calT \alpha$ (to mean~$\alpha$ is producible from $\beta$ by the attachment of one tile), and we write $\beta\to^\calT \alpha$ if $\beta \to_1^{\calT*} \alpha$ (to mean $\alpha$ is producible from $\beta$ by the attachment of zero or more tiles).
When~$\calT$ is clear from context, we may write $\to_1$ and $\to$ instead.
We let $\prodasm{\calT}$ denote the set of producible assemblies of~$\calT$.
An assembly is \emph{terminal} if no tile can be $\tau$-stably attached to it.
We let~$\termasm{\calT} \subseteq \prodasm{\calT}$ denote  the set of producible, terminal assemblies of $\calT$.
A TAS~$\calT$ is \emph{directed} if $|\termasm{\calT}| = 1$. Hence, although a directed system may be nondeterministic in terms of the order of tile placements,  it is deterministic in the sense that exactly one terminal assembly is producible (this is analogous to the notion of {\em confluence} in rewriting systems).

\subsection{Diffusion restrictions: Planar and Spatial aTAM definitions}

In addition to the standard constraints of temperature, dimension, and directedness which serve to differentiate various classes of aTAM systems, in this paper we will also investigate a constraint based on the ability of tiles to diffuse, from arbitrarily far away from an assembly, into frontier locations while always remaining within two dimensions for the 2D version, or three dimensions for the 3D version. This constraint serves to model the fact that when systems are constrained to restricted dimensions, once a region of space is completely blocked by surrounding tiles, there will be no way for tiles to attach within that space. We call the 1D version of the model with this constraint the Linear aTAM, the 2D version the Planar aTAM, and the 3D version the Spatial aTAM.

More formally, a Planar (Spatial) aTAM system is one where, in addition to all of the normal requirements for tile attachment, a tile can only attach to an assembly if there is exists a contiguous path from the node representing the attachment location to a node outside of the minimal bounding box of the assembly in the graph corresponding to the lattice $\mathbb{Z}^2$ ($\mathbb{Z}^3$), such that none of the points along the path are in the domain of the assembly. We call such a path a \emph{diffusion path}.

Notice that, since tiles never detach in the aTAM, once a given location has had all diffusion paths blocked, i.e. it is surrounded by the assembly, no tile will ever be able to attach in that location. We say that the planar (spatial) constraint has been invoked on such a tile location. We call a connected set of locations in which tiles cannot attach due to the planar (spatial) constraint a \emph{constrained subspace}. The set of all tiles that are adjacent to a constrained subspace is called the \emph{constraining subassembly}. Notice that a constraining subassembly is not actually a connected assembly, as it will always contain disconnected sets of tiles (due to the diffusion path only including $\pm x$, $\pm y$, and $\pm z$ movements). In other words, the constraining subassembly is the set of all tiles such that, if any single tile were removed, the constrained subspace would either no longer be constrained or would now contain the location of the removed tile.

Finally, we note that a restriction based on the ability of tiles to be blocked from diffusing into frontier locations by tiles already existing in an assembly does not have any impact on the 1D aTAM, where assemblies are linear (i.e. $1 \times n$ lines). This is because any assembly can only have $0$, $1$, or $2$ frontier locations, and none can be blocked by a tile already attached to the assembly. Thus, the dynamics of the Linear aTAM do not differ from the standard 1D aTAM.

\subsection{Simulation Overview}

In this section, we provide a high-level, intuitive definition of what it means for one tile assembly system to simulate another, and the definition of intrinsic universality. See Section~\ref{sec:simulation_def} for full technical definitions.

Consider the simulation of one system, $\calT$, by another system $\calS$. The simulation by $\calS$ will be done at some scale factor, say $m$, such that in $\calS$, $m \times m$ squares of tiles in 2D, or $m \times m \times m$ cubes of tiles in 3D, represent individual tiles of $\calT$. We call such squares or cubes of tiles in the simulator \emph{macrotiles}, and a macrotile representation function, $R$, must be given to map each macrotile in $\calS$ to a tile in $\calT$. The application of $R$ to all of the macrotiles of an assembly is referred to as $R^*$. For the simulation of $\calT$ by $\calS$ to be \emph{valid}, we say that an assembly $\alpha'$ in $\calS$ which maps (under $R$) to an assembly $\alpha$ in $\calT$ must be able to grow into representations of exactly the same next assemblies that $\calT$ can, and vice versa. An additional constraint that is placed upon the simulator $\calS$ is that it can perform partial growth into empty macrotile regions immediately adjacent to an assembly, allowing it to compute which type of tile may need to be represented there, but it cannot perform such growth, called \emph{fuzz}, further than one macrotile distance into empty space.

We say that a model (or class of systems) is IU if there exists some tile set, say $U$, such that for any system $\calT$ in that model (or class), tiles of $U$ can be arranged into a seed assembly so that subsequent growth of the system using $U$ will correctly simulate $\calT$.

\begin{wrapfigure}{r}{0.6in}
\vspace{-55pt}
\centering
\includegraphics[width=0.55in]{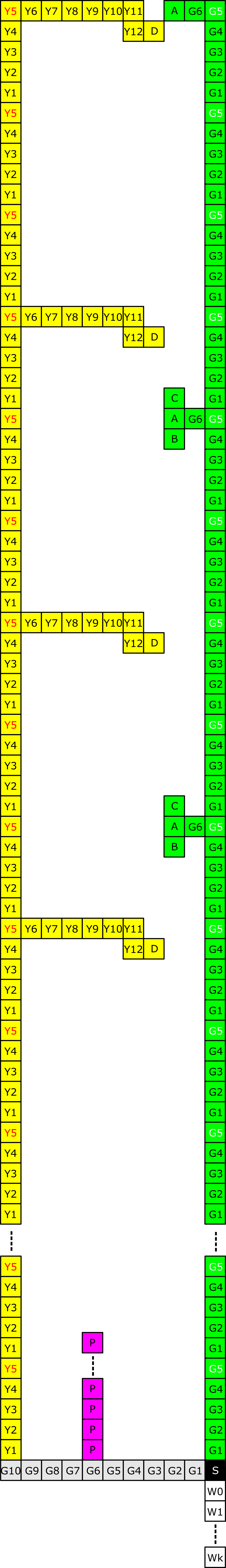}
\caption{$\calT$ of proof of Thm~\ref{thm:PaTAM-not-IU}}
\label{fig:pTAM-notIU-high-level}
\vspace{-85pt}
\end{wrapfigure}

\section{The Planar aTAM is not IU} \label{sec:PaTAM-short}

Here we provide a sketch for the proof of Theorem~\ref{thm:PaTAM-not-IU}. The full proof can be found in Section~\ref{sec:PaTAM-appendix}.

\begin{theorem}\label{thm:PaTAM-not-IU}
The Planar aTAM is not intrinsically universal.
\end{theorem}

We prove Theorem~\ref{thm:PaTAM-not-IU} by contradiction. Therefore, assume that the Planar aTAM is IU, and that the tile set $U$ is the tile set that is IU for it. We give a high-level description of a Planar aTAM system $\calT$ and show that any system using tile set $U$, say $\calUT$, cannot simulate it.

The idea behind $\calT$ is to grow two parallel, arbitrarily tall columns. These columns, at carefully defined periodic intervals, can grow arms inwards to meet each other and seal off space between them and below the meeting point. Figure \ref{fig:pTAM-notIU-high-level} illustrates what these columns look like. There are two types of arms which can grow from the left column and a single type of arm which can grow from the right column. One of the left arms grows a tile upwards, and the other grows a tile downwards, before possibly meeting the corner of a tile of a right arm. Which arm grows from the left column is non-deterministic. Because of this non-determinism and the fact that the arms can grow infinitely tall, we use the Window Movie Lemma (a result shown in \cite{IUNeedsCoop} which is similar to the pumping lemma for regular languages) to show that there is an assembly sequence in $\calT$ such that the macrotile in $\calUT$ at the end of the right arm, which should form a constrained subspace between the arms, will not be able to ``know'' which of the left arm types grew. We then use a case analysis to prove that this macrotile will either have to (1) resolve (i.e. represent a tile in $\calT$ under representation function $R$) before it can sufficiently close off the space below it (which would allow growth to continue in a space representing a constrained subspace), or (2) that tiles must be able to grow outside of the allowed fuzz regions. Either of these result in $\calUT$ not properly simulating $\calT$. (Brief overviews of these cases can be seen in Figures \ref{fig:pTAM-notIU-pinch-points-intro} and \ref{fig:pTAM-notIU-bad-D-growth-intro}.)

\begin{figure}[ht]
\centering
\includegraphics[width=4.0in]{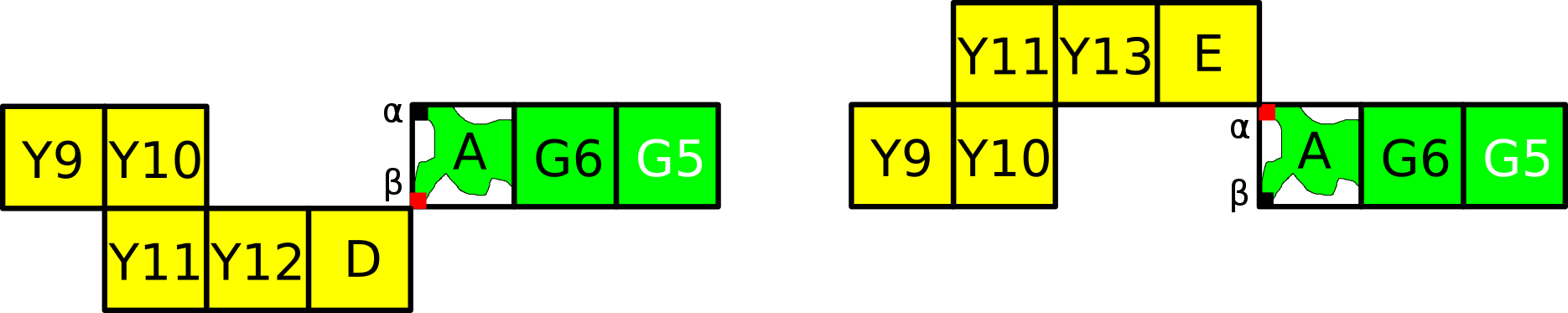}
\caption{Depiction of how untiled locations $\alpha$ or $\beta$ after macrotile $A$ resolves would leave a gap for diffusion of tiles.}
\label{fig:pTAM-notIU-pinch-points-intro}
\end{figure}

\begin{figure}[ht]
\centering
\includegraphics[width=4.0in]{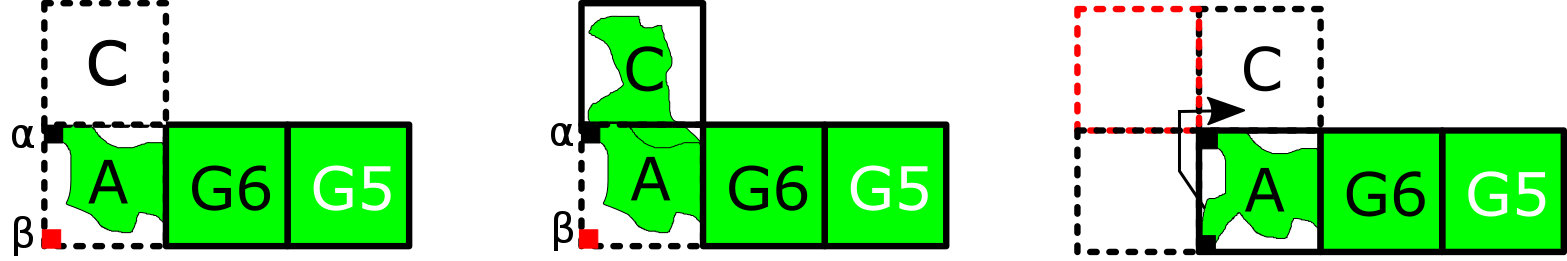}
\caption{(left) Depiction of an $A$ macrotile which has not yet resolved but has a tile in location $\alpha$ and not yet in $\beta$, (middle) If growth necessary to resolve the $C$ macrotile is possible without growing to the left of the $\alpha$ location, then it must be possible to grow and resolve that macrotile before $A$ resolves, (right) If growth necessary to resolve $C$ must grow to the left of the $\alpha$ location, then growth must violate the restriction on fuzz.}
\label{fig:pTAM-notIU-bad-D-growth-intro}
\end{figure}

\vspace{-15pt}
\section{The Directed Planar aTAM is not IU} \label{sec:DPaTAM-short}

Here we provide a sketch for Theorem~\ref{thm:DP2DaTAM-not-IU}. The full proof can be found in Section~\ref{sec:DPaTAM-appendix}.

\begin{theorem} \label{thm:DP2DaTAM-not-IU}
The directed Planar aTAM is not intrinsically universal.
\end{theorem}

We prove Theorem~\ref{thm:DP2DaTAM-not-IU} by contradiction. Therefore, assume that the class of directed systems in the PaTAM is IU, and that $U$ is a universal tile set capable of simulating the entire class. We describe a temperature-1 directed PaTAM system $\calT$ that is impossible for $U$ to simulate. For a visual reference of the terminal assembly of $\calT$, refer to Figure~\ref{fig:DPaTAM-not-IU-system-intro}.

\begin{figure}[h]
\centering
\includegraphics[width=2in]{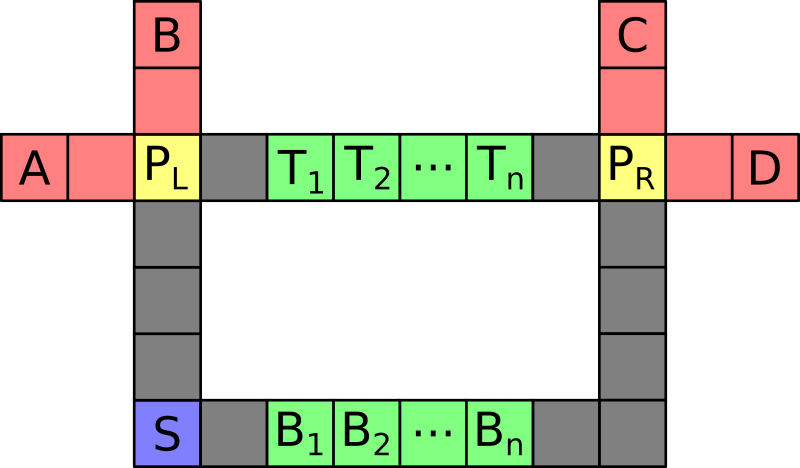}
\hspace{1cm}
\includegraphics[width=2in]{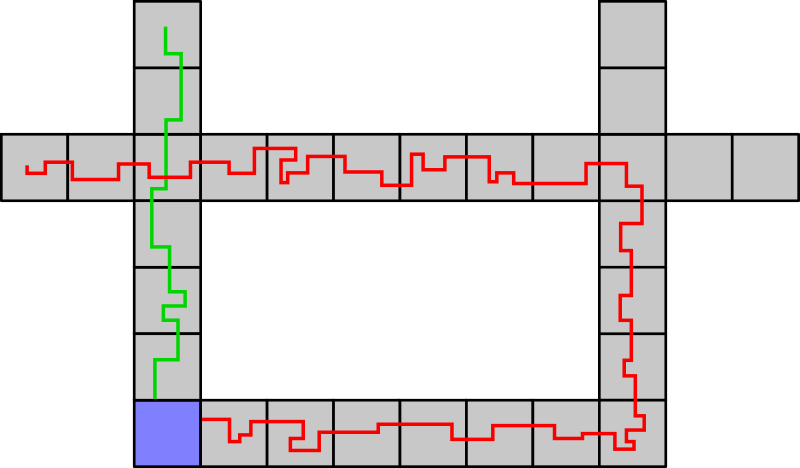}
\caption{An illustration of the directed system $\calT$ which cannot be simulated in a directed manner by a universal tileset in the PaTAM and an illustration of what two growth paths in the simulation would look like were it possible.}
\label{fig:DPaTAM-not-IU-system-intro}
\vspace{-10pt}
\end{figure}

$\calT$ starts with a single seed tile (labelled $S$ in the illustration) and we pick $n$ to be equal to $|U|$. This means that the tiles labelled $T_1, \ldots, T_n$ and $B_1, \ldots, B_n$ grow into rows with as many tiles as there are in the IU tileset $U$. All glues among tiles in $\calT$ are $\tau$-strength and thus there are many possible assembly sequences, all yielding the same terminal assembly. To show that this system cannot be simulated by any system using $U$, we let $\calUT$ be any arbitrary such system and consider two assembly sequences in $\calUT$ which grow the terminal assembly: a clockwise one in which the macrotiles $T_1, \ldots, T_n$ resolve from left to right and a counter-clockwise one where they resolve from right to left. The red tiles protruding from the assembly ending in the tiles $A$, $B$, $C$, and $D$ ensure that contiguous paths of tiles growing in opposite directions ending in these protrusions must constrain the subspace within the assembly (illustrated in Figure \ref{fig:DPaTAM-not-IU-system-intro}).

Then we consider the set $P$ of tiles in the terminal assembly of $\calUT$ which consists of exactly the tiles with the smallest $y$ coordinate in each column within the macrotiles $T_1, \ldots, T_n$ (i.e. the south-most tile for each given $x$ coordinate in these macrotiles). We then show that it must be the case that, if these macrotiles resolve from left to right, the order of attachment of tiles in $P$ must also be from left to right, otherwise there must necessarily exist a tile in $P$ who's attachment can occur inside a potentially constrained subspace. The same can be shown, using a symmetric argument, for an assembly sequence where $T_1, \ldots, T_n$ resolve from right to left. Using this, we then perform a case analysis to consider assembly sequences in which tiles are attached from both directions, meeting at a single tile in $P$. This tile in $P$ will necessarily be inside of a constrained subspace (Figure \ref{fig:DPaTAM-not-IU-contra-intro} illustrates the idea used to prove both of these claims). If a tile can grow inside of a potentially constrained subspace, under certain assembly sequences, it would be impossible for that tile to attach in all. This would result in multiple terminal assemblies which contradicts the assumption that $\calUT$ is a directed simulator.

\begin{figure}[h]
\centering
\includegraphics[height=1in]{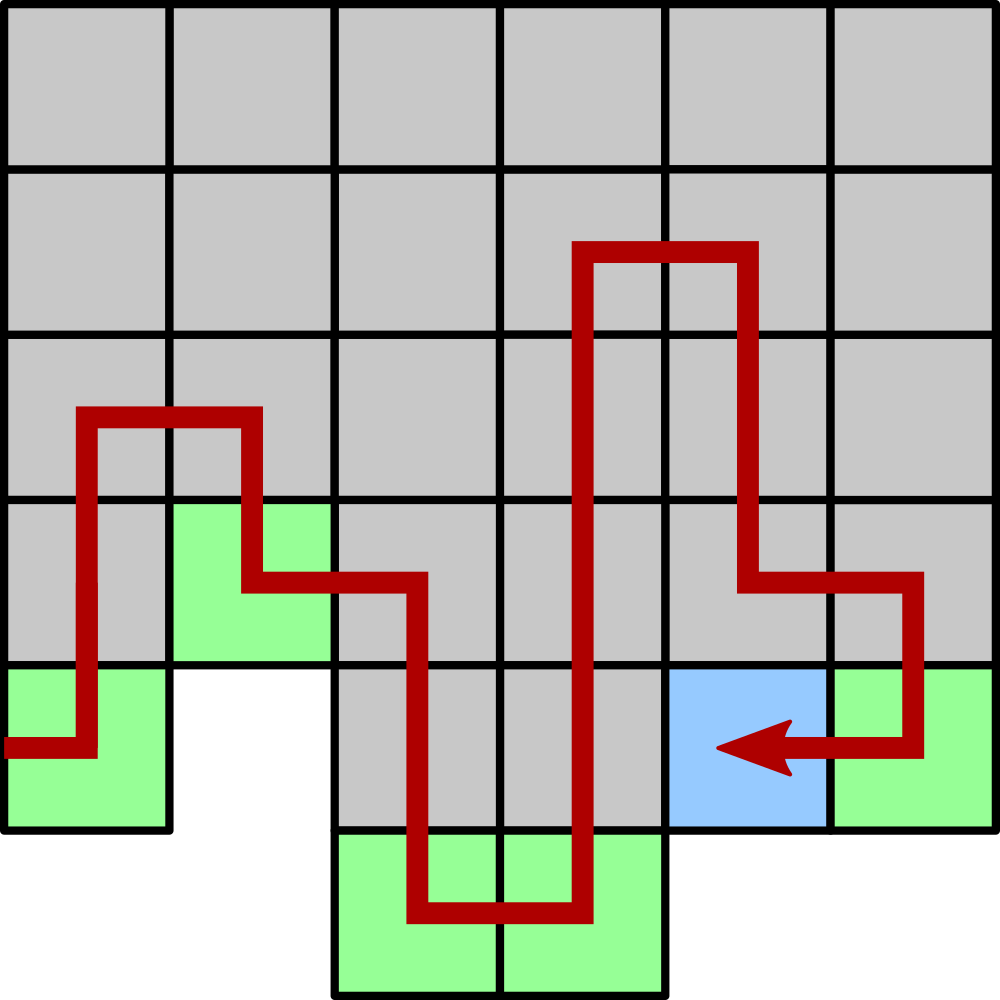}
\hspace{1cm}
\includegraphics[height=1in]{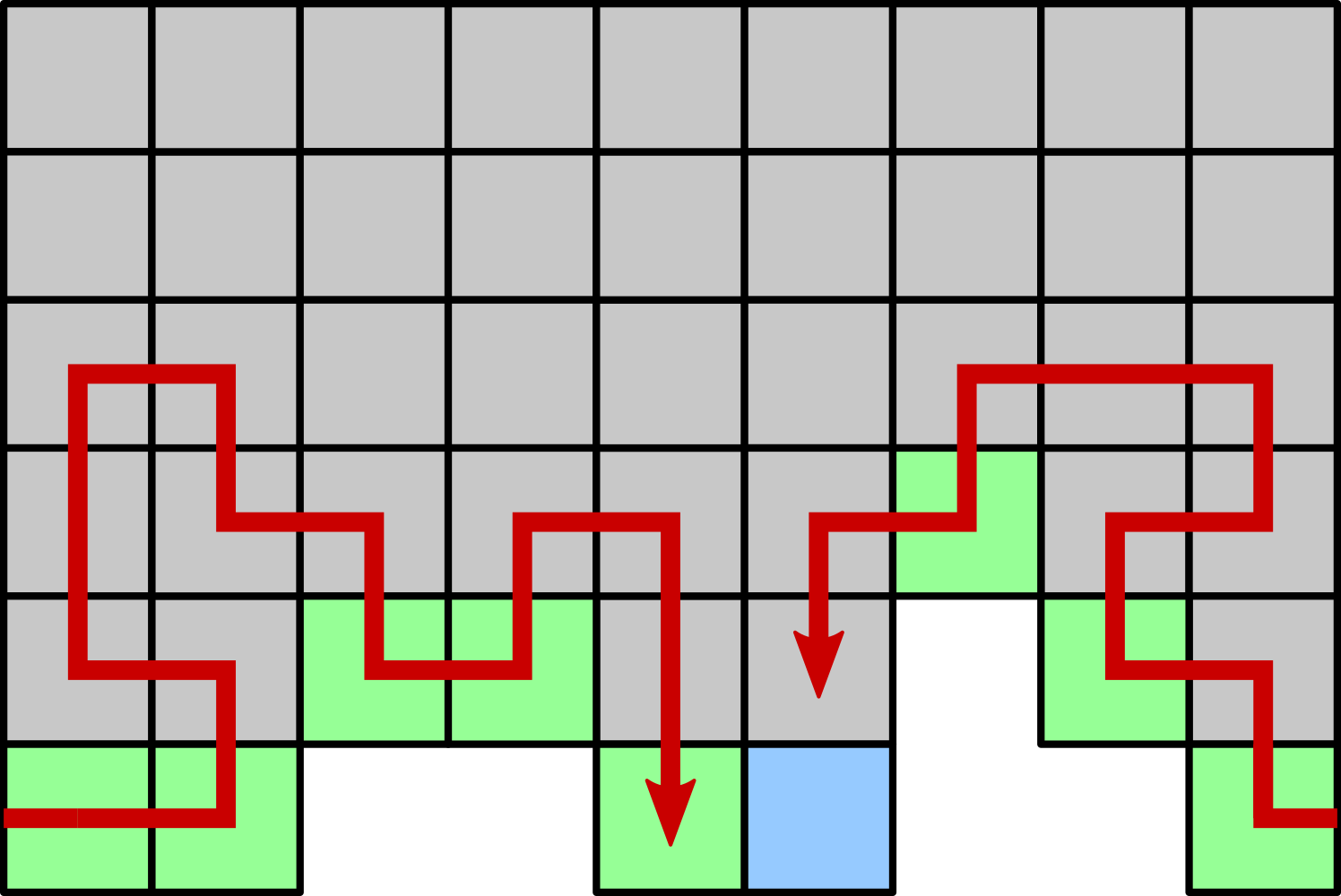}
\caption{(left) If one of the tiles in $P$ can be attached after both of its neighbors, that tile (illustrated in blue) must be able to attach within a constrained subspace. (right) It's possible to follow an assembly sequence which places tiles in $P$ such that one of them is trapped in a constrained subspace.}
\label{fig:DPaTAM-not-IU-contra-intro}
\vspace{-10pt}
\end{figure}

\section{The 3D aTAM is IU} \label{sec:construction-short}

\begin{theorem} \label{thm:3DaTAMIU}
The 3D aTAM is intrinsically universal.
\end{theorem}

To prove Theorem~\ref{thm:3DaTAMIU}, we show that there exist functions  $\mathcal{R}$ and $S$ (to generate representation functions and seed assemblies) and some tile set $U$ such that, for each $\mathcal{T} = (T,\sigma,\tau)$ which is a TAS in the 3D aTAM, there is a constant $m \in \mathbb{N}$ such that, letting $R = \mathcal{R}(\mathcal{T})$, $\sigma_\mathcal{T} = S(\mathcal{T})$, and $\calUT = (U,\sigma_\mathcal{T},\tau')$, $\calUT$ simulates $\mathcal{T}$ at scale $m$ using macrotile representation function $R$.  To do so, we will set $\tau' = 2$ (i.e. the simulations by $U$ will all be at temperature $=2$), and we will explicitly define $U$ and give the algorithms which implement $\mathcal{R}$ and $S$. The scale factor $m$ of the simulation will be $O(|T|^2\log(|T|\tau))$.

In this section, we provide a high-level overview of the components in the construction and how they are combined to create $\calUT$ which simulates arbitrary 3D aTAM system $\mathcal{T}$. More thorough descriptions, proofs of correctness, and low-level details are in Sections~\ref{sec:construction}, \ref{sec:thm1-proof}, and \ref{sec:low-level}, respectively.

The main concepts of our construction can be broken down into four \emph{modules} or functional subassemblies: the $\genome$, the $\adderArray$, the $\bracket$, and the $\externalCommunication$. We provide brief descriptions of the functions of each here.

\begin{itemize}
\itemsep0em
\item $\genome$ : This module contains an encoding of the system to be simulated, $\calT$, in the form of a look-up table that takes as input the tile type / direction pair of a neighboring macrotile and outputs every potential tile type that could form a bond with that neighbor and the strength of that bond. The $\genome$ also contains instructions to build the other modules listed here.
\item $\adderArray$ : This module is responsible for determining, for each tile type $t \in T$, if there are enough glues incident on the current macrotile for it to begin to represent a tile of type $t$ under $R$. It does this by adding up the bond strengths (from the $\genome$) with which tile $t$ could attach in the current location and making sure the total is sufficient for attachment.
\item $\bracket$ : Once the $\adderArray$ determines the tile types into which the macrotile could resolve, this module picks one tile type non-deterministically (if there is a choice).
\item $\externalCommunication$ : This module carries an encoding of the decided upon tile type (as output from the $\bracket$) from the current macrotile to all neighboring macrotiles.
\end{itemize}

Now, we will describe the process by which one macrotile block $L$ goes from empty space to fully grown. We call the transition of a macrotile from mapping to empty space in $\calT$ to mapping to a tile in $\calT$ \emph{differentiation}.

\begin{enumerate}
\itemsep0em
\item Once a neighboring macrotile location has differentiated, it exports a copy of the $\genome$ and its $\externalCommunication$ to macrotile $L$.
\item The $\genome$ propagates around $L$ and initiates growth of the other three modules.
\item The incoming $\externalCommunication$ modules from differentiated neighbors grow into the $\genome$ to query for whether or not their glues could contribute to the differentiation of $L$.
\item The information from the previous step is sent to the $\adderArray$ to determine if enough glues are present to allow simulation of the attachment of specific tile types.
\item Encodings of potential tile types enter the $\bracket$ where one is non-deterministically chosen.
\item The winning tile type leaves the $\bracket$ and grows into the $\externalCommunication$.
\item The $\genome$ and $\externalCommunication$ modules are propagated to neighboring macrotiles.
\end{enumerate}

Essentially, a macrotile block $L$ in the terminal assembly of $\calUT$ (representing a tile location $l$ in the terminal assembly of $\calT$) can be in one of three states. If $l$ is not adjacent to any tile in the terminal assembly, $L$ will be completely empty. If $l$ is adjacent to a tile but does not have enough incident glues to be a frontier location, $L$ will have all four modules set up but no tile types will be output from the $\adderArray$ to the $\bracket$. Finally, if $l$ represents a tile, $L$ will have all four modules set up and an encoding of that tile type will have left the $\adderArray$, made it through the $\bracket$, and be outputted to neighboring macrotile locations by the $\externalCommunication$.

In addition to this growth paradigm, our construction needs a  representation function $R$ and seed function $S$. $R$ works by using the scale factor of the simulation to determine where the output of the $\bracket$ will be in each macrotile block. Once this set of relative tile positions within each block is filled, $R$ reads the encoding within the individual tiles and outputs the corresponding tile type from $T$. To obtain the seed, the simulated system $\calT$ is input and a corresponding $\genome$ is created. This $\genome$ is placed in the macrotile locations that map to tile locations filled by the seed in $\calT$. Additionally, a hard-coded $\bracket$ output is included in each seed macrotile to ensure that the seed of $\calUT$ maps under $R^*$ to the seed of $\calT$. Once the simulation begins, this seed is able to start propagating the $\genome$ and $\externalCommunication$ modules to neighboring macrotile locations of the seed to start the process of differentiation for those locations and all further growth.

Recall that this construction is implemented on the individual tile level. Section~\ref{sec:implementation} contains the URL's for the universal tile set, seed generation scripts, and optimized PyTAS simulator.

\section{The Directed 3D aTAM is IU} \label{sec:directed3d-short}

In this section, we show that the directed subset of 3D aTAM systems is itself IU, since the tile set and simulation we constructed for the proof of Theorem~\ref{thm:3DaTAMIU} was carefully designed so that, whenever a directed system is simulated, the simulating system is also directed. Recall that a directed system has only a single terminal assembly. This means that if location $l \in \mathbb{Z}^3$ is mapped to a tile of type $t$ in one assembly sequence, in every other valid assembly sequence, location $l$ is also mapped (eventually) to a tile of type $t$.

\begin{theorem} \label{thm:directed3DIU}
The directed 3D aTAM is intrinsically universal.
\end{theorem}

Due to space constraints, we simply provide an overview of the scenarios which needed to be analyzed to show our construction remains directed when simulating directed systems. The full details of the proof can be found in Section~\ref{sec:thm2-proof}. For our analysis, we consider the types of nondeterminism which can arise in the 3D aTAM and show that none of them will cause nondeterminism in the form of undirectedness when $\calUT$ is simulating $\calT$. There are three essential types of nondeterminism in the 3D aTAM: (1) the random selection of one frontier location, out of possibly many, for a tile attachment in each step of the assembly process, (2) locations where one or more incident glues match those of multiple tile types with enough strength to allow any of them to bind, and (3) locations which can receive tiles of different types depending on which adjacent positions are tiled first (i.e. nondeterminism caused by the relative timing of growth of different portions of the assembly). The first type of nondeterminism does not cause a system to be undirected, as long as the ordering of tile additions does not lead to one of the other types of nondeterminism. In addition, the second type of nondeterminism is avoided in our construction by careful design of the tile types, such that (other than a few specific special cases described in the proof) they are all designed with distinct input and output sides (where input sides are used for the initial attachment of a tile and output sides are used to allow other tiles to bind afterward) and no two tile types have the same sets of input glues. Furthermore, backward growth, which could allow tiles to attach using their output glues, is avoided using ``key and latch'' techniques (see Section~\ref{sec:growth-patterns} for technical details). That leaves the final type of nondeterminism, which is based off of ``race conditions'' between the growth of different portions of the assembly, as the only type of nondeterminism to be analyzed.

Consider the case where $\calT = (T,\sigma,\tau)$ and $T$ has just a single tile type where all $6$ sides of that tile type have the same glue which is of strength $\tau$. Let $\sigma$ consist of just a single instance of that tile type placed at $(0,0,0)$. $\calT$ is directed, with a single terminal assembly which is the infinite, complete tiling of $\mathbb{Z}^3$ with tiles of that single type. However, this system has an uncountably infinite number of valid assembly sequences. If directedness in the simulator $\calUT$ is to be maintained, it is necessarily the case that the terminal assembly of this system must appear identical in the situation where every macrotile had its growth initiated by each of its neighbors but also initiated the growth of each neighbor, since for each scenario there exists a valid assembly sequence in $\alpha$ which matches that ordering, and $\calUT$, by condition of being directed, is only allowed a single terminal assembly. It is for this reason that the bands of the $\genome$ were designed so that they merge seamlessly at input and output intersections and form a connected structure that makes it impossible to determine the ordering of their growth into macrotile locations. Additionally, every macrotile which differentiates sends its output $\externalCommunication$ datapaths to all $6$ neighboring macrotiles, which will all accept them and grow through $\genome$ queries and the $\adderArray$ as though they were the first inputs and will also seamlessly merge outputs of multiple pieces within the $\adderArray$ such that it impossible to tell which subset of glues caused a specific tile type to be output. (Since we are only concerned with the simulation of directed systems, other than a specific case discussed in the proof, only one tile type will ever be able to output to the $\bracket$.) The full proof gives a description of how the modules of our construction are designed to maintain directedness in spite of nondeterministic rates of growth of the components, and thus why $U$ is IU for the class of directed 3D aTAM systems, which is therefore IU itself.

\section{The Spatial aTAM is IU} \label{sec:spatial-atam}

Here we describe our construction that proves the Spatial aTAM is IU. The full proof is in Section~\ref{sec:spatial_tech_details}.

\begin{theorem} \label{thm:spatial-IU}
The Spatial aTAM is intrinsically universal.
\end{theorem}

This construction is an augmentation of the construction used to prove Theorem~\ref{thm:3DaTAMIU}. The problem with using the original construction is that it is able to grow and differentiate new macrotiles within locations that map to constrained subspaces (i.e. subspaces which are completely sealed~off by the tiles of the assembly). To prevent this, we supplement the original construction to use a \emph{blocking protocol} that will force tiles to attach around the boundary of a macrotile which hasn't yet differentiated, but still allow diffusion through a series of one-tile-wide pipes until differentiation happens.

The centerpiece of this augmented construction is a structure that we'll subsequently refer to as the \emph{pipe intersection}. Shown in Figure~\ref{fig:meeting_point}, this structure helps scale the spatial constraint by tying the six paths that connect through it to the six faces of the macrotile. Therefore, by adding the central tile location as an input to the representation function $R$, we can have the macrotile differentiate in the exact same assembly step that the diffusion paths between all six neighbors are cut off. In other words, placing a tile in the middle of the pipe intersection (1) blocks any diffusion between the neighboring macrotiles and (2) causes the current macrotile to differentiate. To implement this new blocking protocol, instead of performing step 7 in the growth sequence of the original construction, we now perform the following sequence of steps after step 6:

\begin{enumerate}
\itemsep0em
\item $\externalCommunication$ and $\genome$ modules grow to boundaries of macrotile and pause.
\item Pipes are seeded at the pipe intersection and grow out to all boundaries (except the top).
\item Boundaries are tiled, starting from the bottom boundary and growing in a spiral around side boundaries, growing around I/O datapaths and the ends of the pipes.
\item The pipe intersection is filled from above and the macrotile officially differentiates.
\item Tiles diffuse into the pipes from the outside, grow through the pipes and activate the $\genome$ and $\externalCommunication$ to continue growing into neighboring macrotiles.
\end{enumerate}

\begin{wrapfigure}{r}{2.0in}
\centering
\includegraphics[width=2.0in]{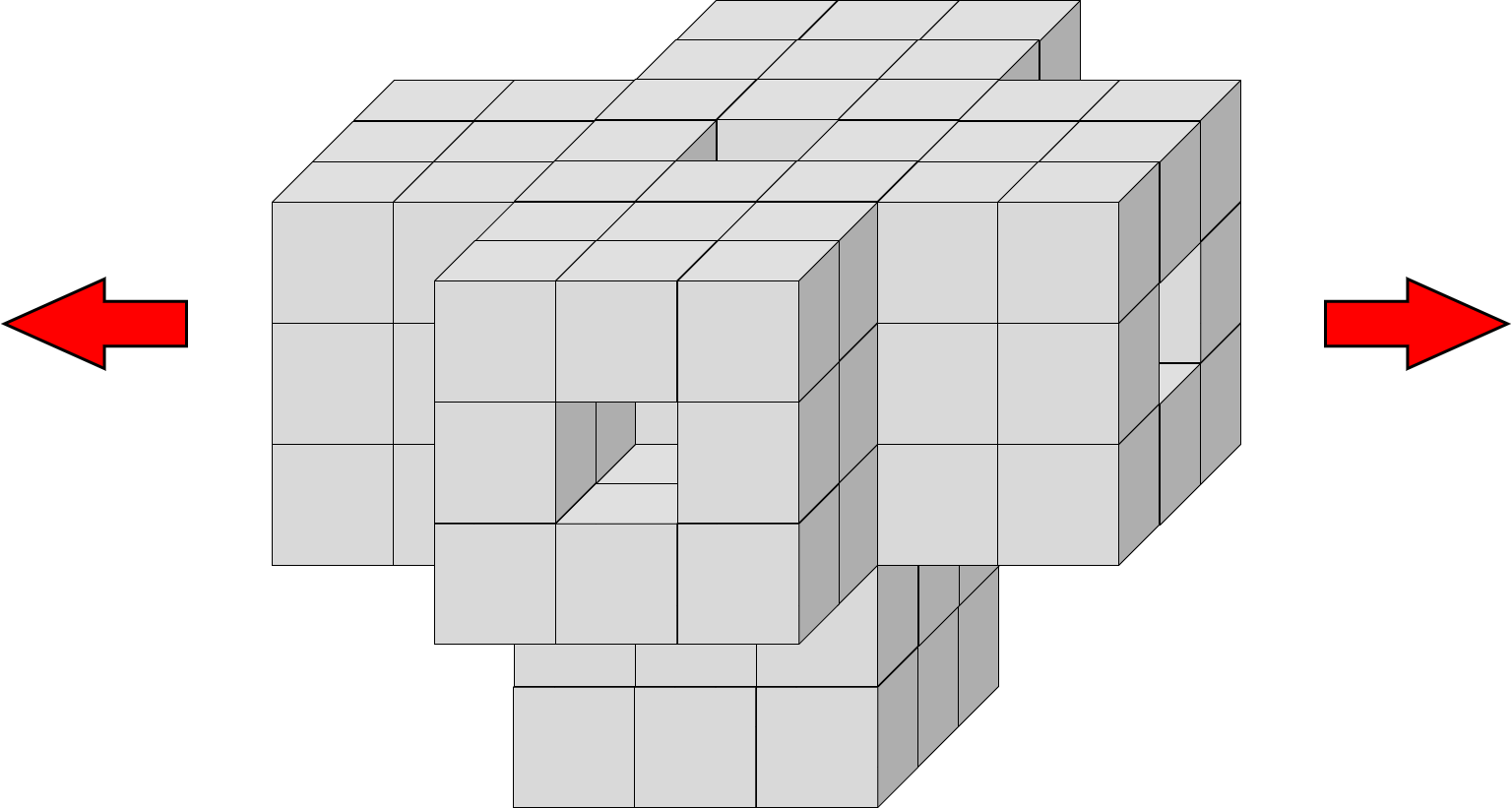}
\caption{The pipe intersection. This structure is important in the simulation of the spatial constraint because it allows multiple paths to be cut off simultaneously when the macrotile differentiates.}
\label{fig:meeting_point}
\end{wrapfigure}

From here, we can prove that diffusion paths through a macrotile (from one side to another) exist only when the macrotile maps to empty space under $R$. Using this, we can then prove that paths through non-differentiated macrotiles can be strung together to make a diffusion path in $\calUT$ to macrotile locations that map to unconstrained space in $\calT$. The full proof for both of these lemmas can be found in Section~\ref{sec:spatial_tech_details}.
Utilizing the dynamics of the original construction with the addition of the blocking protocol to simulate the spatial constraint, this system is capable of simulating any Spatial aTAM system. With the addition of a slightly augmented seed generation function $S$ and representation function $\mathcal{R}$, this construction provides as an intrinsically universal tile set for the Spatial aTAM, thereby proving Theorem~\ref{thm:spatial-IU}.

\section{Formal description of the abstract Tile Assembly Model} \label{sec-tam-formal}

This section gives a formal definition of the abstract Tile Assembly Model (aTAM)~\cite{Winf98}. For readers unfamiliar with the aTAM,~\cite{RotWin00} gives an excellent introduction to the model.
For purposes of notational convenience, throughout this paper we will use the term ``aTAM'' will refer to the 2D aTAM.

Fix an alphabet $\Sigma$.
$\Sigma^*$ is the set of finite strings over $\Sigma$.
$\Z$, $\Z^+$, and $\N$ denote the set of integers, positive integers, and nonnegative integers, respectively.
Let $d \in \{2,3\}$.
Given $V \subseteq \Z^d$, the \emph{full grid graph} of $V$ is the undirected graph $\fullgridgraph_V=(V,E)$, 
and for all $\vec{x} = \left(x_0, \ldots, x_{d-1}\right), \vec{y} = \left(y_0, \ldots, y_{d-1}\right)\in V$, $\left\{\vec{x},\vec{y}\right\} \in E \iff \| \vec{x} - \vec{y}\| = 1$; i.e., if and only if $\vec{x}$ and $\vec{y}$ are adjacent on the $d$-dimensional integer Cartesian space.

A $d$-dimensional \emph{tile type} is a tuple $t \in (\Sigma^* \times \N)^{2d}$; e.g., a unit square (or cube) with four (or six) sides listed in some standardized order, each side having a \emph{glue} $g \in \Sigma^* \times \N$ consisting of a finite string \emph{label} and nonnegative integer \emph{strength}.
We assume a finite set of tile types, but an infinite number of copies of each tile type, each copy referred to as a \emph{tile} (either a 2D square or 3D cube tile type). A $d$-dimensional tile set is a set of $d$-dimensional tile types and is written as $d$-$T$. A tile set~$T$ is a set of $d$-dimensional tile types for some $d \in \{2,3\}$.

A $d$-{\em configuration} is a (possibly empty) arrangement of tiles on the integer lattice $\Z^d$, i.e., a partial function $\alpha:\Z^d \dashrightarrow T$. A configuration $\alpha$ is a $d$-configuration for some $d \in \{2,3\}$.
Two adjacent tiles in a configuration \emph{interact}, or are \emph{attached}, if the glues on their abutting sides are equal (in both label and strength) and have positive strength.
Each configuration $\alpha$ induces a \emph{binding graph} $\bindinggraph_\alpha$, a grid graph whose vertices are positions occupied by tiles, according to $\alpha$, with an edge between two vertices if the tiles at those vertices interact. A $d$-\emph{assembly} is a connected non-empty configuration, i.e., a partial function $\alpha:\Z^d \dashrightarrow T$ such that $\fullgridgraph_{\dom \alpha}$ is connected and $\dom \alpha \neq \emptyset$. An assembly is a $d$-assembly for some $d \in \{2,3\}$. The \emph{shape} $S_\alpha \subseteq \Z^d$ of $\alpha$ is $\dom \alpha$.

Given $\tau\in\Z^+$, $\alpha$ is \emph{$\tau$-stable} if every cut of~$\bindinggraph_\alpha$ has weight at least $\tau$, where the weight of an edge is the strength of the glue it represents. 
When $\tau$ is clear from context, we say $\alpha$ is \emph{stable}.
Given two assemblies $\alpha,\beta$, we say $\alpha$ is a \emph{subassembly} of $\beta$, and we write $\alpha \sqsubseteq \beta$, if $S_\alpha \subseteq S_\beta$ and, for all points $p \in S_\alpha$, $\alpha(p) = \beta(p)$.

A $d$-dimensional \emph{tile assembly system} ($d$-TAS) is a triple $d$-$\mathcal{T} = (d\textrm{-}T,\sigma,\tau)$, where $d$-$T$ is a finite set of $d$-dimensional tile types, $\sigma:\Z^d \dashrightarrow T$ is the finite, $\tau$-stable, $d$-dimensional \emph{seed assembly}, and $\tau\in\Z^+$ is the \emph{temperature}. The triple $\mathcal{T} = (T,\sigma,\tau)$ is a TAS if it is is a $d$-TAS for some $d \in \{2,3\}$.
Given two $\tau$-stable assemblies $\alpha,\beta$, we write $\alpha \to_1^{\mathcal{T}} \beta$ if $\alpha \sqsubseteq \beta$ and $|S_\beta \setminus S_\alpha| = 1$. In this case we say $\alpha$ \emph{$\mathcal{T}$-produces $\beta$ in one step}. 
If $\alpha \to_1^{\mathcal{T}} \beta$, $ S_\beta \setminus S_\alpha=\{p\}$, and $t=\beta(p)$, we write $\beta = \alpha + (p \mapsto t)$.
The \emph{$\mathcal{T}$-frontier} of $\alpha$ is the set $\partial^\mathcal{T} \alpha = \bigcup_{\alpha \to_1^\mathcal{T} \beta} S_\beta \setminus S_\alpha$, the set of empty locations at which a tile could stably attach to $\alpha$. The \emph{$t$-frontier} $\partial_t \alpha \subseteq \partial \alpha$ of $\alpha$ is the set $\setr{p\in\partial \alpha}{\alpha \to_1^\mathcal{T} \beta \text{ and } \beta(p)=t}.$ 

Let $\mathcal{A}^T$ denote the set of all assemblies of tiles from $T$, and let $\mathcal{A}^T_{< \infty}$ denote the set of finite assemblies of tiles from $T$.
A sequence of $k\in\Z^+ \cup \{\infty\}$ assemblies $\alpha_0,\alpha_1,\ldots$ over $\mathcal{A}^T$ is a \emph{$\mathcal{T}$-assembly sequence} if, for all $1 \leq i < k$, $\alpha_{i-1} \to_1^\mathcal{T} \alpha_{i}$.
The {\em result} of an assembly sequence is the unique limiting assembly (for a finite sequence, this is the final assembly in the sequence).

We write $\alpha \to^\mathcal{T} \beta$, and we say $\alpha$ \emph{$\mathcal{T}$-produces} $\beta$ (in 0 or more steps) if there is a $\mathcal{T}$-assembly sequence $\alpha_0,\alpha_1,\ldots$ of length $k = |S_\beta \setminus S_\alpha| + 1$ such that
(1) $\alpha = \alpha_0$,
(2) $S_\beta = \bigcup_{0 \leq i < k} S_{\alpha_i}$, and
(3) for all $0 \leq i < k$, $\alpha_{i} \sqsubseteq \beta$.
If $k$ is finite then it is routine to verify that $\beta = \alpha_{k-1}$.
We say $\alpha$ is \emph{$\mathcal{T}$-producible} if $\sigma \to^\mathcal{T} \alpha$, and we write $\prodasm{\mathcal{T}}$ to denote the set of $\mathcal{T}$-producible assemblies. The relation $\to^\mathcal{T}$ is a partial order on $\prodasm{\mathcal{T}}$ 

An assembly $\alpha$ is \emph{$\mathcal{T}$-terminal} if $\alpha$ is $\tau$-stable and $\partial^\mathcal{T} \alpha=\emptyset$.
We write $\termasm{\mathcal{T}} \subseteq \prodasm{\mathcal{T}}$ to denote the set of $\mathcal{T}$-producible, $\mathcal{T}$-terminal assemblies. If $|\termasm{\mathcal{T}}| = 1$ then  $\mathcal{T}$ is said to be {\em directed}.

When $\mathcal{T}$ is clear from context, we may omit $\mathcal{T}$ from the notation above and instead write
$\to_1$,
$\to$,
$\partial \alpha$, 
\emph{assembly sequence},
\emph{produces},
\emph{producible}, and
\emph{terminal}.

\subsection{Formal Definitions of Simulation} \label{sec:simulation_def}

To state our main result, we must formally define what it means for one TAS to ``simulate'' another.  Our definitions come from \cite{IUNeedsCoop} with the natural modifications to extend from 2D to 3D.  Intuitively, simulation of a system $\mathcal{T}$ by another system $\mathcal{S}$ is done by utilizing some scale factor $m \in \mathbb{Z}^+$ such that $m \times m \times m$ cubes of tiles in $\mathcal{S}$ represent individual tiles in $\mathcal{T}$, and there is a ``representation function'' which is able to interpret the assemblies of $\mathcal{S}$ as assemblies in $\mathcal{T}$.

From this point on, let $T$ be a tile set and let $m\in\Z^+$.
An \emph{$m$-block macrotile} over $T$ is a partial function $\alpha : \Z_m^3 \dashrightarrow T$, where $\Z_m = \{0,1,\ldots,m-1\}$.
Let $B^T_m$ be the set of all $m$-block macrotiles over $T$.
The $m$-block with no domain is said to be $\emph{empty}$.
For a general assembly $\alpha:\Z^3 \dashrightarrow T$ and $(x',y',z')\in\Z^3$, define $\alpha^m_{(x',y',z')}$ to be the $m$-block macrotile defined by $\alpha^m_{(x',y',z')}(i_x,i_y,i_z) = \alpha(mx'+i_x,my'+i_y,mz'+i_z)$ for $0 \leq i_x,i_y,i_z< m$.
For some tile set $S$, a partial function $R: B^{S}_m \dashrightarrow T$ is said to be a \emph{valid $m$-block macrotile representation} from $S$ to $T$ if for any $\alpha,\beta \in B^{S}_m$ such that $\alpha \sqsubseteq \beta$ and $\alpha \in \dom R$, then $R(\alpha) = R(\beta)$.

For a given valid $m$-block macrotile representation function $R$ from tile set~$S$ to tile set $T$, define the \emph{assembly representation function}\footnote{Note that $R^*$ is a total function since every assembly of $S$ represents \emph{some} assembly of~$T$; the functions $R$ and $\alpha$ are partial to allow undefined points to represent empty space.}  $R^*: \mathcal{A}^{S} \rightarrow \mathcal{A}^T$ such that $R^*(\alpha') = \alpha$ if and only if $\alpha(x,y,z) = R\left(\alpha'^m_{(x,y,z)}\right)$ for all $(x,y,z) \in \Z^3$.
For an assembly $\alpha' \in \mathcal{A}^{S}$ such that $R^*(\alpha') = \alpha$, $\alpha'$ is said to map \emph{cleanly} to $\alpha \in \mathcal{A}^T$ under $R^*$ if for all non empty blocks $\alpha'^m_{(x,y,z)}$, $(x,y,z)+(u_x,u_y,u_z) \in \dom(\alpha)$ for some $(u_x,u_y,u_z) \in U_3$ such that $u^2_x + u^2_y + u^2_z \le 1$, or if $\alpha'$ has at most one non-empty $m$-block $\alpha^m_{0,0}$.  In other words, $\alpha'$ may have tiles on macrotile blocks representing empty space in $\alpha$, but only if that position is adjacent to a tile in $\alpha$.  We call such growth ``around the edges'' of $\alpha'$ \emph{fuzz} and thus restrict it to be adjacent to only valid macrotiles, but not diagonally adjacent (i.e.\ we do not permit \emph{diagonal fuzz}).

In the following definitions, let $\mathcal{T} = \left(T,\sigma_T,\tau_T\right)$ be a TAS, let $\mathcal{S} = \left(S,\sigma_S,\tau_S\right)$ be a TAS, and let $R$ be an $m$-block representation function $R:B^S_m \rightarrow T$.

\begin{definition}
\label{def-equiv-prod} We say that $\mathcal{S}$ and $\mathcal{T}$ have \emph{equivalent productions} (under $R$), and we write $\mathcal{S} \Leftrightarrow \mathcal{T}$ if the following conditions hold:
\begin{enumerate}[topsep=0pt,itemsep=-1ex,partopsep=1ex,parsep=1ex]
\item $\left\{R^*(\alpha') | \alpha' \in \prodasm{\mathcal{S}}\right\} = \prodasm{\mathcal{T}}$.
\item $\left\{R^*(\alpha') | \alpha' \in \termasm{\mathcal{S}}\right\} = \termasm{\mathcal{T}}$.
\item For all $\alpha'\in \prodasm{\mathcal{S}}$, $\alpha'$ maps cleanly to $R^*(\alpha')$.
\end{enumerate}
\end{definition}

\begin{definition}
\label{def-t-follows-s} We say that $\mathcal{T}$ \emph{follows} $\mathcal{S}$ (under $R$), and we write $\mathcal{T} \dashv_R \mathcal{S}$ if $\alpha' \rightarrow^\mathcal{S} \beta'$, for some $\alpha',\beta' \in \prodasm{\mathcal{S}}$, implies that $R^*(\alpha') \to^\mathcal{T} R^*(\beta')$.
\end{definition}

The next definition essentially specifies that every time $\mathcal{S}$ simulates an assembly $\alpha \in \prodasm{\mathcal{T}}$, there must be at least one valid growth path in $\mathcal{S}$ for each of the possible next steps that $\mathcal{T}$ could make from $\alpha$ which results in an assembly in $\mathcal{S}$ that maps to that next step.

\begin{definition}
\label{def-s-models-t} We say that $\mathcal{S}$ \emph{models} $\mathcal{T}$ (under $R$), and we write $\mathcal{S} \models_R \mathcal{T}$, if for every $\alpha \in \prodasm{\mathcal{T}}$, there exists $\Pi \subset \prodasm{\mathcal{S}}$ where $R^*(\alpha') = \alpha$ for all $\alpha' \in \Pi$, such that, for every $\beta \in \prodasm{\mathcal{T}}$ where $\alpha \rightarrow^\mathcal{T} \beta$, (1) for every $\alpha' \in \Pi$ there exists $\beta' \in \prodasm{\mathcal{S}}$ where $R^*(\beta') = \beta$ and $\alpha' \rightarrow^\mathcal{S} \beta'$, and (2) for every $\alpha'' \in \prodasm{\mathcal{S}}$ where $\alpha'' \rightarrow^\mathcal{S} \beta'$, $\beta' \in \prodasm{\mathcal{S}}$, $R^*(\alpha'') = \alpha$, and $R^*(\beta') = \beta$, there exists $\alpha' \in \Pi$ such that $\alpha' \rightarrow^\mathcal{S} \alpha''$.
\end{definition}

\begin{definition}\label{def:s-simulates-t}
\label{def-s-simulates-t} We say that $\mathcal{S}$ \emph{simulates} $\mathcal{T}$ (under $R$) if $\mathcal{S} \Leftrightarrow_R \mathcal{T}$ (equivalent productions), $\mathcal{T} \dashv_R \mathcal{S}$ and $\mathcal{S} \models_R \mathcal{T}$ (equivalent dynamics).
\end{definition}

\newcommand{\REPL}{\mathsf{REPR}}
\newcommand{\frakC}{\mathfrak{C}}

\subsection{Intrinsic universality}
\label{sec:iu_def}
Now that we have a formal definition of what it means for one tile system to simulate another, we can proceed to formally define the concept of intrinsic universality, i.e., when there is one general-purpose tile set that can be appropriately programmed to simulate any other tile system from a specified class of tile systems.

Let $\REPL$ denote the set of all macrotile representation functions (i.e., $m$-block macrotile representation functions for some $m\in\Z^+$).  Define $\frakC$ to be a class of tile assembly systems, and let $U$ be a tile set.  Note that each element of $\frakC$, $\REPL$, and $\mathcal{A}^U_{< \infty}$ is a finite object, hence encoding and decoding of simulated and simulator assemblies can be defined to be computable via standard models such as Turing machines and Boolean circuits.

\begin{definition}\label{def:iu-specific-temp}
We say $U$ is \emph{intrinsically universal} for $\frakC$ \emph{at temperature} $\tau' \in \Z^+$ 
if there are computable functions $\mathcal{R}:\frakC \to \REPL$ and $S:\frakC \to \mathcal{A}^U_{< \infty}$ such that, for each $\mathcal{T} = (T,\sigma,\tau) \in \frakC$, there is a constant $m\in\N$ such that, letting $R = \mathcal{R}(\mathcal{T})$, $\sigma_\mathcal{T}=S(\mathcal{T})$, and $\mathcal{U}_\mathcal{T} = (U,\sigma_\mathcal{T},\tau')$, $\mathcal{U}_\mathcal{T}$ simulates $\mathcal{T}$ at scale $m$ and using macrotile representation function~$R$.
\end{definition}
That is, $\mathcal{R}(\mathcal{T})$ outputs a representation function that interprets assemblies of $\mathcal{U}_\mathcal{T}$ as assemblies of $\mathcal{T}$, and $S(\mathcal{T})$ outputs the seed assembly used to program tiles from $U$ to represent the seed assembly of $\mathcal{T}$.

\begin{definition}
\label{def:iu-general}
We say that~$U$ is \emph{intrinsically universal} for $\frakC$ if it is intrinsically universal for $\frakC$ at some temperature $\tau'\in Z^+$.
\end{definition} 

\begin{definition}
We say that $\frakC$ is intrinsically universal if there exists some $U$ such that $U$ is instrinsically universal for $\frakC$.
\end{definition}

\section{Details of Observations Regarding the 1D aTAM}\label{sec:1D}

In this section we make observations about the lack of intrinsic universality in the 1-dimensional (1D) aTAM, which can be thought of similarly to the 2D aTAM, but where tiles are only allowed to bind via east and west edges (i.e. only forming 1D line assemblies).

\begin{observation}\label{obs:1D-not-IU}
The 1D aTAM is not intrinsically universal.
\end{observation}

\begin{proof}
Observation~\ref{obs:1D-not-IU} can be easily proven by contradiction. Therefore, assume the 1D aTAM is IU, and that tile set $U$ is a tile set that is IU for it. Let $|U| = t$, that is, $t$ is the number of unique tile types in $U$. We can simply define 1D aTAM system $\calT = (T,\sigma,1)$ such that $|T| = t+1$ and its seed $\sigma$ consists of a single tile at the origin. The tiles of $\calT$ are designed so that for each $t_i \in T$ for $0 < i < t$, the west glue of $t_1$ is of type $g_i$ and its east glue is of type $g_{i+1}$. The seed tile only has an east glue, and it is of type $g_1$, and tile type $t_t$ only has a west glue, and it is of type $g_{t-1}$. All glues have strength $=1$. Clearly, $\calT$ forms a terminal assembly which is a line of length $t+1$, which extends to the east from the seed tile. Since $U$ contains only $t$ unique tile types, any line which it forms of length $>t$ must contain a duplicated tile type, and a simple ``pumping'' argument shows that whatever assembly occurs between the two occurrence must be able to appear again to the east of the second occurrence, and this can be repeated infinitely often. Therefore, any simulating system which makes use of $U$ must be able to make infinite assemblies if it can make any which are longer than length $t$. Thus, it cannot simulate $\calT$, which only makes a single, finite terminal assembly.
\end{proof}

\begin{observation}\label{obs:1D-directed-not-IU}
The class of directed 1D aTAM systems is not intrinsically universal.
\end{observation}

\begin{proof}
The proof of Observation~\ref{obs:1D-directed-not-IU} follows immediately from the fact that $\calT$ of the proof of Observation~\ref{obs:1D-not-IU} is directed, and since it can be constructed to be larger than any 1D tile set which is claimed to be IU for directed 1D aTAM systems, it can't be simulated (in a directed manner or otherwise) using such a tile set.
\end{proof}

To be analogous to the Planar aTAM and Spatial aTAM, in which tiles are constrained by the required ability to diffuse within 2 and 3 dimensions, respectively, the \emph{Linear aTAM} for the 1D case actually turns out not to be different from the general 1D aTAM, since in 1D no open frontier location can be blocked off from possible incoming tiles by tiles already attached to an assembly. Therefore, the following observation follows immediately from Observation~\ref{obs:1D-not-IU}.

\begin{observation}\label{obs:1D-linear-not-IU}
The Linear aTAM is not intrinsically universal.
\end{observation}

As with the previous observation, the addition of the requirement for diffusion doesn't change the dynamics of the model, so the following observation follows immediately from Observation~\ref{obs:1D-directed-not-IU}.

\begin{observation}\label{obs:1D-directed linear-not-IU}
The class of directed Linear aTAM systems is not intrinsically universal.
\end{observation}

\section{Technical Details for the Planar aTAM is not IU} \label{sec:PaTAM-appendix}

In this section, we provide the technical details for Section~\ref{sec:PaTAM-short} and Theorem~\ref{thm:PaTAM-not-IU}.

\begin{proof}
We prove Theorem~\ref{thm:PaTAM-not-IU} by contradiction. Therefore, assume that the Planar aTAM is IU, and that the tile set $U$ is the tile set that is IU for it. We will now define Planar aTAM system $\calT = (T,\sigma,1)$ and assume that $\calUT = (U,\sigma_\calT,\tau)$ is the system which simulates it with $U$, where $m$ is the scale factor and $R$ is the $m$-block macrotile representation function. The general procedure will be to select valid assembly sequences in $\calT$ which arrive at specified target shapes. We'll then have the simulation by $\calUT$ proceed to the point of matching that assembly (which must be possible if $\calUT$ simulates $\calT$), and we'll inspect and record the assembly sequences followed. This will eventually allow us to prove that $\calUT$ has to violate the definition of simulation. For our notation, we will refer to assemblies in $\calT$ as $\alpha$, $\beta$, etc., and those in $\calUT$ which map, under $R$, to them as $\alpha'$, $\beta'$, etc. (i.e. they will named as primed versions). Similarly, an assembly sequence in $\calT$ will be referred to as $\Vec{\alpha}$, while one in $\calUT$ will be $\Vec{\alpha'}$.

\begin{figure}
\centering
\includegraphics[width=3.0in]{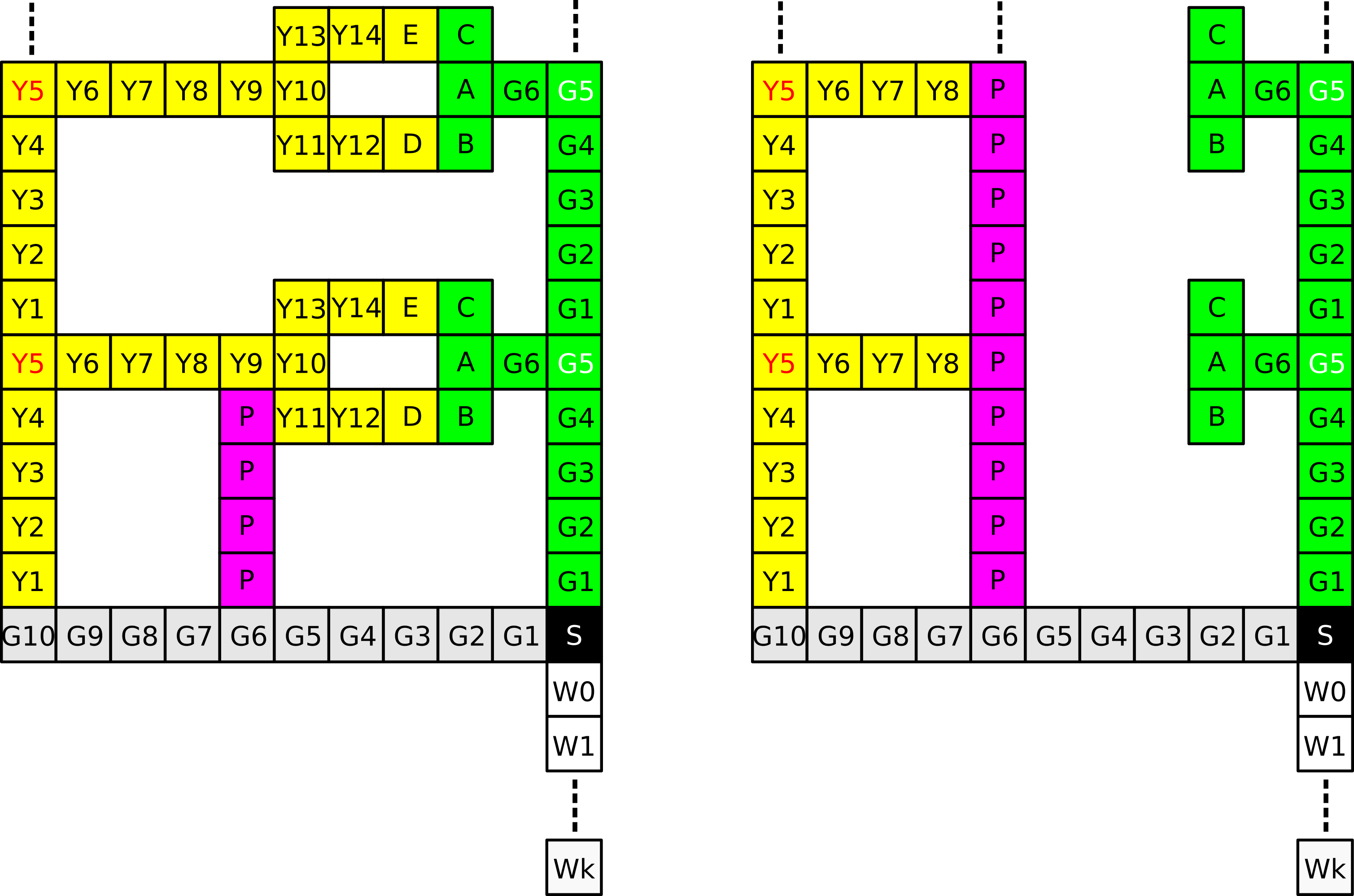}
\caption{Partial overview of Planar aTAM system $\calT$, showing partial portions of two possible infinite assemblies. Growth begins from the seed (black). Growth below the seed is deterministic and grows a length $k$ column of tiles. The green column grows infinitely tall, and from tiles of type $G5$ grows am arm leftward which either terminates in a tile of type $A$, with a $B$ tile below and $C$ tile above. Note however that restrictions based on planarity may, depending on timing, prevent growth of the arms. The yellow column grows infinitely tall, and from every tile of type $Y5$ it is possible (if not blocked by the pink column or prevented by planarity restrictions) to grow an arm rightward which terminates in $E$ and $D$ type tiles. The pink column can grow infinitely upward, unless blocked by a yellow arm or prevented by planarity restrictions.}
\label{fig:pTAM-notIU-overview}
\end{figure}

\begin{figure}
\centering
\includegraphics[width=1.7in]{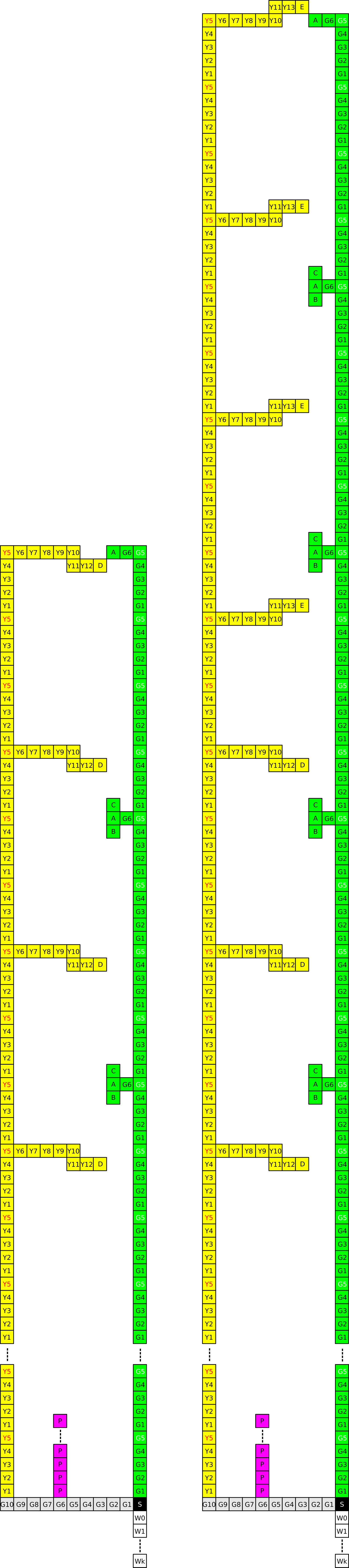}
\caption{Portion of the pumped growth of $\calT$. In this example, the height of a yellow iteration is $3$ blocks and the height of a green iteration is $4$ blocks, so on the left, tiles of the yellow and green columns become diagonally adjacent after $4$ yellow iterations and $3$ green iterations (neither of which count the bottom portion before arms are grown). The same relationship would occur for iteration heights of $n-1$ and $n$ for any $n$.}
\label{fig:pTAM-notIU-pumped-growth}
\end{figure}

Figure~\ref{fig:pTAM-notIU-overview} shows an overview of some possible subassemblies of two infinite terminal assemblies of $\calT$, which is an undirected system capable of growing an infinite number of infinite assemblies. Let $|U| = k$ be the size of tile set $U$. We define $\calT$ so that it begins from a single seed tile at the origin. We now describe a valid assembly sequence, $\Vec{\alpha_{\texttt{pre}}}$, which we will select for $\calT$.

First, a column grows downward from the seed, a distance of $k+1$ with each location being of a unique tile type. (This ensures that the scale factor $m$ of $\calUT$ must be greater than $1$, since $U$ only has $k$ tile types and therefore must use more than one to uniquely map to each of the $k$ different tile types.) It then grows a row of $10$ grey tiles to the west.

\begin{figure}
\centering
\begin{subfigure}[b]{1.8in}
\centering
\includegraphics[width=1.0in]{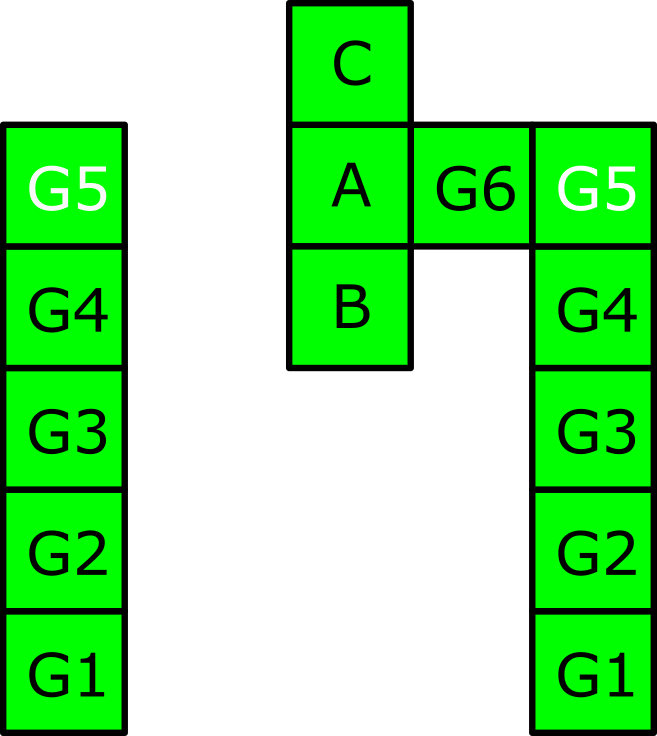}
\caption{Green blocks $\greenoff$ (left), and $\greenon$ (right)}
\label{fig:pTAM-not-IU-green-blocks}
\end{subfigure}
\hspace{0.05in}
\begin{subfigure}[b]{3.8in}
\includegraphics[width=3.8in]{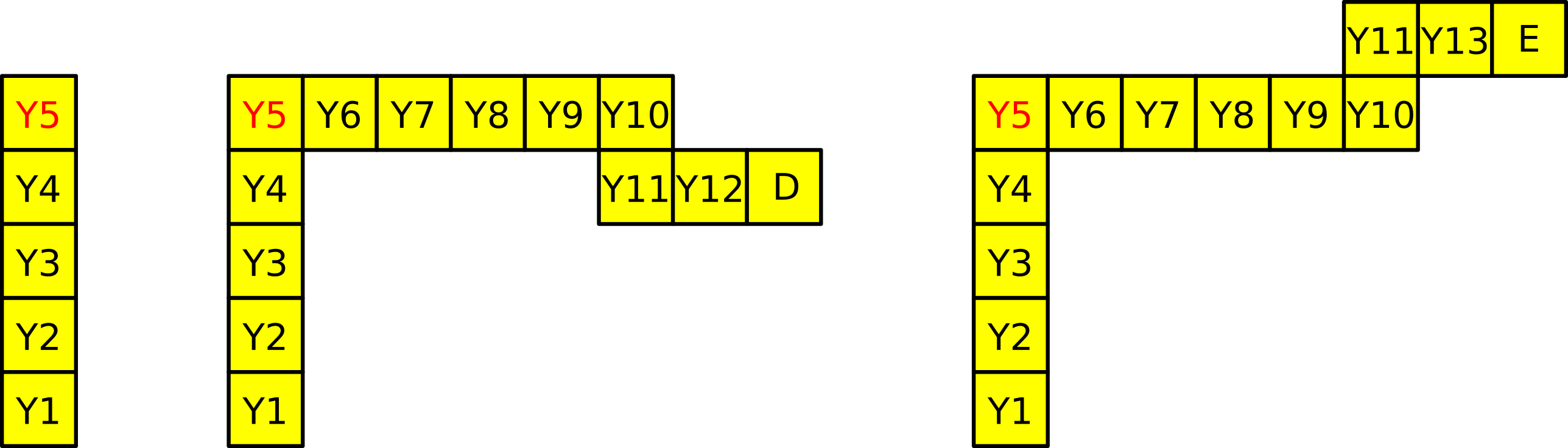}
\caption{Yellow blocks $\yellowoff$ (left), $\yellowdown$ (middle), and $\yellowup$ (right)}
\label{fig:pTAM-not-IU-yellow-blocks}
\end{subfigure}
\caption{Small subassemblies formed from green and yellow tiles in $\calT$ which will form the logical building blocks of larger assemblies. Note that in any block it is possible for all tiles of that color to attach, but we only consider assembly sequences in which particular subsets have attached.}
\label{fig:pTAM-not-IU-blocks}
\end{figure}

Note that since $m$ is the scale factor that $\calUT$ uses to simulate $\calT$, if there is a single-tile-wide column of tiles which is being represented in $\calUT$, then a cut which separates that scaled-up column (and thus the entire assembly) into two halves does not have to be longer than $3m$, since it could cut the macrotile representing a tile of $\calT$ as well as the maximum allowed fuzz of width equal to one macrotile on each side of it. If we let $g$ equal the number of unique glues in $U$ and note that $g+1$ can account for each of those glues plus the null glue, we can see that the value $p = ((g+1)^{6m}\cdot(6m)!+1)$ represents one greater than the maximum number of possible ways that such a cut could receive glues along its two sides (i.e. all possible sets of glues that are incident upon its sides and all possible orderings of arrival for the glues of each set). Following \cite{IUNeedsCoop}, we call the cuts \emph{windows} and each set of glues and ordering of their arrivals a \emph{window movie}. Similar to the use of window movies in \cite{IUNeedsCoop}, we note that if a column contains $p$ cuts, then there must be at least two which are duplicates of each other. The Window Movie Lemma of \cite{IUNeedsCoop} uses this fact to show that the subassembly between two identical cuts can be ``pumped'' either up or down, meaning that an arbitrary number of additional copies can be added between the two identical cuts, or the current copy can be removed, and the resulting assembly must be producible by a valid assembly sequence. We will use the two facts that (1) such a valid assembly sequence of $\calUT$ exists, and (2) $\calUT$ is assumed to simulate $\calT$ (so all valid assembly sequences of $\calUT$ must correspond to valid simulations of $\calT$) to note that whenever there are two identical window movies cutting an assembly in $\calUT$, there is an assembly sequence which we can run forward (or in reverse) one step at a time which will grow (or shrink) the macrotiles in a valid ordering (i.e. which maps to a valid assembly sequence in $\calT$) and without breaking the allowed boundaries of surrounding fuzz.

Next, a preliminary set of green, yellow, and pink tiles north of the seed attaches as follows. (See Figure~\ref{fig:pTAM-not-IU-blocks} for depictions of the blocks referenced, and note that whenever we talk about the attachment of a block, we mean the attachment of tiles one at a time to form that block.) A series of $6mp$ $\greenoff$ blocks form. Then, a series of $6mp$ $\yellowoff$ blocks form, and then a series of $p$ pink tiles. We will call the current assembly $\alpha_{\texttt{pre}}$. A schematic depiction of it can be seen as the lower portion of Figure~\ref{fig:pTAM-notIU-pumped-growth}.

At this point, the green and yellow columns are much taller than the pink column, but are the same height as each other (since they grew the same number of blocks and the blocks are all of height $5$). We have also guaranteed enough room for the pink column to grow upward if needed, and as much as may be needed, without being blocked by any yellow arms during the following growth sequence.

We now run the simulation of $\calUT$ until it places the first tile in its assembly $\alpha_{\texttt{pre}}'$ such that $R(\alpha_{\texttt{pre}}') = \alpha_{\texttt{pre}}$, and we record the entire assembly sequence $\Vec{\alpha_{\texttt{pre}}'}$.

We define \emph{iterations} as subassemblies composed of specific numbers of blocks. Let a $\texttt{green}$ $\texttt{iteration}$ consist of $6m+2$ $\greenoff$ blocks followed by a single $\greenon$ block. Let a $\texttt{yellow iteration}$ consist of $6m+1$ $\yellowoff$ blocks followed by a single $\yellowdown$ block. Let a $\texttt{pink iteration}$ consist of a single pink tile. We use the term \emph{pumping a column} when we do the following. In $\calT$, have the column grow $p$ iterations, then in $\calUT$ have the simulator follow that growth sequence and record the assembly sequence. By the time the $p$th iteration completes, by the definition of $p$ it must be the case that at least two iterations have the same window movie separating their first and second macrotiles. Because of this and our previously described ability to pump such an assembly, we rewind the assembly until we return to the first tile placement against the first of the identical window movies, and use our ability to construct an assembly sequence which creates identical copies of the assembly between those cuts to create a new assembly sequence which builds an assembly of exactly the same height as before we rewound (which may mean that the last full copy of the pumped portion of the assembly is not completed, but this is still a valid assembly sequence and assembly). Note that this will result in an assembly in $\calUT$ which maps to the same assembly of $\calT$ as before we pumped the column, but it may be a different assembly. 

The goal of the next portion of this process is to pump each of the columns until each can be pumped while being guaranteed not to influence, or be influenced by, any of the others. This is possible because (1) during this period, no tiles of any column will be close enough to each other to directly interact via adjacent tiles, even through fuzz, and (2) interaction via paths of tiles which travel through the grey macrotiles at the bottom of the columns is bounded because each such path adds to the width of a cut between those macrotiles, and the maximum width of such a cut, even including allowable fuzz, is limited to $3m$.

Pump the green column, and record whether or not during the final, resulting assembly sequence, any tile is placed along either the cut between the macrotiles representing $G8$ and $G9$
, or the cut between those representing $G3$ and $G4$ (which includes the boundaries between those macrotiles and the fuzz regions above and below them). We will call such growth \emph{collusion}. Do the same for the pink column, then for the yellow column.

Repeat the following loop $6m+1$ times:
\begin{enumerate}
\item If the pumping of the yellow or the pink columns caused collusion since the last time the green column was pumped, which may have resulted in new tile placements in the green column, pump the green column again (i.e. let it grow $p$ iterations, find matching window movies, rewind, and regrow using the repeated subassembly). Otherwise, if no collusion occurred, continue the same pumping which was done to complete the last $p$ iterations until the green column has grown another $p$ iterations.

\item If the pumping of the green or yellow columns caused collusion since the last time the pink column was pumped, pump the pink column again. Otherwise, continue the previous pumping of the pink column until it has grown another $p$ iterations.

\item If the pumping of the green or pink columns caused collusion since the last time the yellow column was pumped, pump the yellow column again. Otherwise, continue the previous pumping of the yellow column until it has grown another $p$ iterations.
\end{enumerate}

By the assumption that $\calUT$ simulates $\calT$, it must be the case that the previous loop can complete, with no portion of the assembly in $\calUT$ violating the constraints of fuzz. During the execution of the loop, once there is a loop iteration where no collusion occurs during the pumping of any of the three columns, then for the remainder of the loop, each column simply completes by pumping the same assembly sequence. Since each of the two cuts which may be transmitting a path of collusion is a maximum of $3m$ in width, collusion could occur a maximum of $6m$ times. Thus, by pumping each column once before the loop, then iterating the loop $6m+1$ times, it must be the case that no collusion occurred during iteration $6m+1$ of the loop, and so the last iteration in the loop consists of each column simply continuing the pumping which is used to finish the previous iteration.

Recall that the height of a green iteration is $6m+2$ $\greenoff$ blocks plus one $\greenon$ block $= 6m+3$ blocks. The height of a yellow iteration is $6m+1$ $\yellowoff$ blocks plus one $\yellowdown$ block $= 6m+2$ blocks. By the dimensions of the iterations and geometry of the blocks, the first time that a yellow tile (of type $D$) will be diagonally adjacent to a green tile (of type $A$) is after exactly $6m+3$ yellow iterations and $6m+2$ green iterations. (An example where the heights of yellow and green iterations are $3$ and $4$, respectively, can be seen on the left side of Figure~\ref{fig:pTAM-notIU-pumped-growth}.)

Finally, after this loop completes, we then grow another $6m+3$ yellow iterations; however, instead of using a $\yellowdown$ block at the end of each iteration, we use a $\yellowup$ block instead. We then pump the green column up to the height of an additional $6m+2$ green iterations. This results in a $\yellowup$ block against a $\greenon$ block which is a pumped copy of a previous $\greenon$ block. Note that since we are pumping from the previous green iterations in this final iteration, we know that there is no collusion occurring between this last $\greenon$ block and the yellow column. Furthermore, recall that we previously pumped the green and yellow columns (without any arms which could block the pink column) to a height of $6mp$ blocks (i.e. $5\cdot6mp = 30mp$ tiles), and the pink column is now only $p + 6m + 2$ tiles, so there is no chance for collision with the yellow arms.

We will now inspect the way in which the macrotiles representing green $A$ tiles are grown. Refer to Figure~\ref{fig:pTAM-notIU-pinch-points} for a depiction of the following argument. Let $\alpha$ and $\beta$ represent the extreme northwest and southwest corners, respectively, of a macrotile representing an $A$ tile in any of the first $6m+2$ iterations. We will first prove that, at the first tile placement during which such a macrotile represents an $A$ tile rather than empty space (which we will refer to as the macrotile \emph{resolving}), the locations $\alpha$ and $\beta$ must contain tiles (perhaps one of them receiving that first tile which causes the macrotile to resolve). We will rewind the assembly sequence of the final green iteration so that its last block has just placed the first tile that causes its $A$-representing macrotile to resolve. We will now grow the yellow iteration by one more iteration, pausing it immediately after it places the first tile that causes its last block's macrotile to resolve to a $D$. Recall that this $D$ macrotile will be diagonally adjacent to the final $A$ macrotile of the green column. Also recall that the additional growth of the yellow column is guaranteed to be unable to collude with the other columns, and since no additional tiles have been placed by either column since they resolved their $A$ and $D$ tiles, there cannot be any tiles in their surrounding fuzz. Therefore, if the $\beta$ position does not have a tile at this point, there must be a free path in the plane for tiles to diffuse from infinitely far away, through the gap of that location, through all of the gaps between the green and yellow columns, and down to the pink column.  Furthermore, we know that the pink column is capable of being pumped for further growth. Therefore, there is a valid assembly sequence which grows additional macrotiles which resolve to pink tiles. However, this violates the definition of simulation, specifically because $\calT$ does not follow $\calUT$ because the corresponding assembly in $\calT$ is prevented from adding additional pink tiles due to the planar constraint because the final $A$ and $D$ tiles of its green and yellow columns close off that portion of the plane. Therefore, the $\beta$ location must have a tile at the time the $A$-representing macrotile resolves.

To see why the $\alpha$ location must also have a tile, consider the $A$ macrotile in the final iteration whose yellow column contains a $\yellowup$ block. Notice that in this iteration, because of the $\yellowup$ block, it must be the case that the macrotile representing $A$ must have a tile at location $\alpha$ because otherwise, for the exact same reason as with location $\beta$ in the previous iterations, the pink tiles would be able to continue growth in a constrained subspace. Since, in this macrotile, a tile must have been placed in location $\alpha$ before resolving, and since the green column in this final iteration is simply the result of pumping from the previous iterations which must all contain tiles at location $\beta$, it must be the case that during some iteration, a tile must have attached in both the $\beta$ and $\alpha$ locations before the corresponding macrotile resolved to $A$.

We now know that there exists a macrotile representing an $A$ in which both the $\alpha$ and $\beta$ locations must have tiles before resolving. By the dynamics of the aTAM, and the fact that the scale factor $m$ at which $\calUT$ simulates $\calT$ must be greater than $1$, those must be two distinct locations and must therefore receive tiles in different steps. Therefore, without loss of generality, we will assume that it is possible for the $\alpha$ position to receive a tile first. (In the case where it's the $\beta$ position, an identical but symmetric argument will hold.) Please refer to Figure~\ref{fig:pTAM-notIU-bad-D-growth} for a depiction of the following argument. We rewind the growth of the yellow column by an iteration so that it no longer has a macrotile diagonally adjacent to the $A$ macrotile. There are two possibilities for the growth of the $C$ macrotile using the assembly sequence which we have been pumping for the green column. (1) All of the growth necessary to resolve the $C$ macrotile is possible by placing tiles strictly to the right of the $\alpha$ location. In this case, there is a valid assembly sequence which halts the growth of $A$ before it has resolved, but which grows and resolves the macrotile representing $C$. This is because, in order for macrotile $A$ to resolve, tile location $\beta$ must grow after $\alpha$ but is not strictly to the right of $\alpha$. This breaks the simulation of $\calT$ by $\calUT$ because the $A$ macrotile in the assembly in $\calUT$ maps to empty space, so $\calT$ cannot follow $\calUT$ by attaching a corresponding $C$ tile. This leaves one final possible case: (2) The growth of $C$ can only be completed if one or more tiles grow to the left of the $\alpha$ location in $A$. However, such growth would be required to grow through the diagonally adjacent macrotile location in order to grow into, or cooperate with tiles within, the $C$ macrotile. Even if the $A$ macrotile first resolves, this still results in fuzz growing diagonally out of bounds, which also breaks the simulation by $\calUT$.

\begin{figure}
\centering
\includegraphics[width=5.5in]{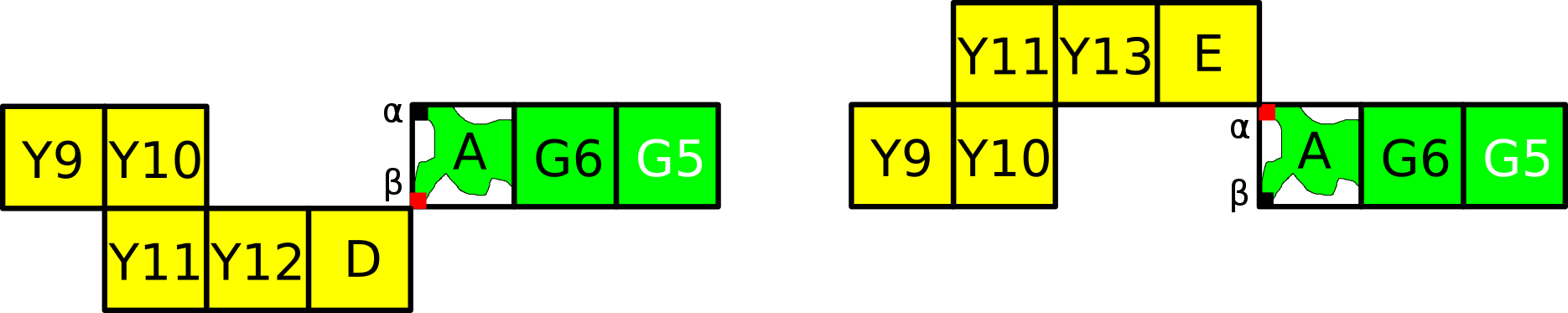}
\caption{Depiction of how untiled locations $\alpha$ or $\beta$ after the $A$ macrotile resolves would leave a gap for the diffusion of tiles.}
\label{fig:pTAM-notIU-pinch-points}
\end{figure}

\begin{figure}
\centering
\includegraphics[width=5.5in]{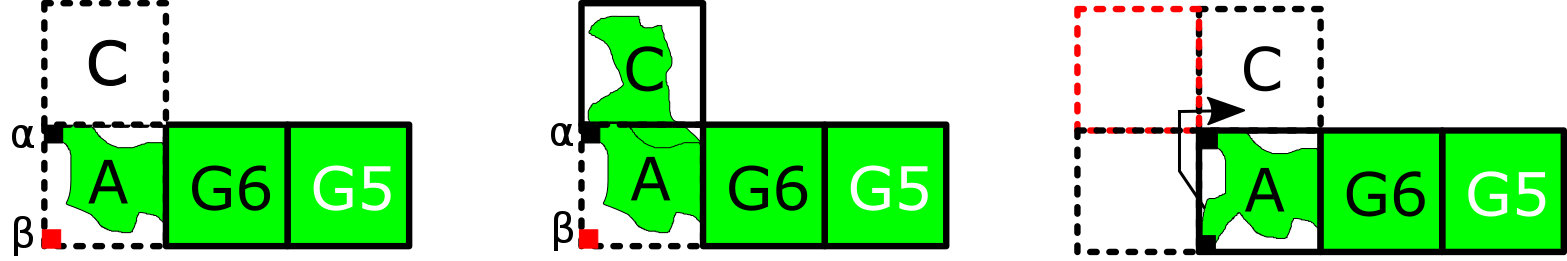}
\caption{(left) Depiction of an $A$ macrotile which has not yet resolved but has a tile in location $\alpha$ and not yet in $\beta$, (middle) If growth necessary to resolve the $C$ macrotile is possible without growing to the left of the $\alpha$ location, then it must be possible to grow and resolve that macrotile before $A$ resolves, (right) If growth necessary to resolve $C$ must grow to the left of the $\alpha$ location, then growth must violate the restriction on fuzz.}
\label{fig:pTAM-notIU-bad-D-growth}
\end{figure}

Therefore, $\calUT$ does not simulate $\calT$, and therefore our assumption that $U$ is IU for the Planar aTAM is false, and thus no tile set is IU for the Planar aTAM, and Theorem~\ref{thm:PaTAM-not-IU} is proven.

\end{proof}

\section{Technical Details for the Directed Planar aTAM is not IU} \label{sec:DPaTAM-appendix}

In this section, we provide the technical details for Section~\ref{sec:DPaTAM-short} and Theorem~\ref{thm:DP2DaTAM-not-IU}.

\begin{figure}
\centering
\includegraphics[width=2in]{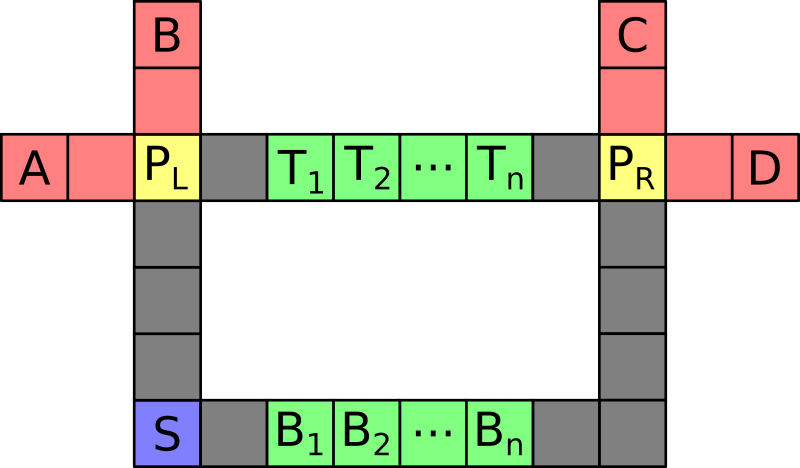}
\hspace{1cm}
\includegraphics[width=2in]{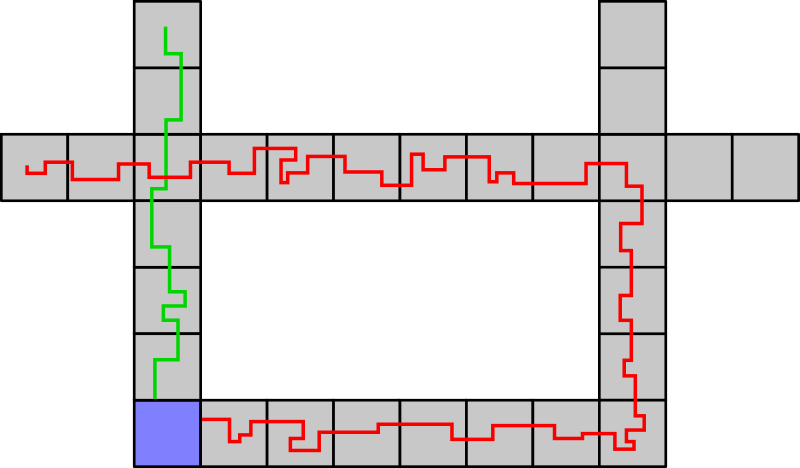}
\caption{An illustration of the directed system $\calT$ which cannot be simulated in a directed manner by a universal tileset in the PaTAM and an illustration of what two paths of the simulation would look like were it possible.}
\label{fig:DPaTAM-not-IU-system}
\vspace{-10pt}
\end{figure}

\begin{proof}

We prove Theorem~\ref{thm:DP2DaTAM-not-IU} by contradiction. Therefore, assume that the class of directed systems in the PaTAM is IU, and that $U$ is a tile set which is IU for it. Let $n = |U|$ be the number of tile types in $U$. We will now show a directed PaTAM $\calT$ which cannot be simulated by any directed PaTAM system using $U$. Let $\calT$ be the directed temperature-1 PaTAM system illustrated in Figure~\ref{fig:DPaTAM-not-IU-system}, and for the sake of contradiction assume that $\calUT = (U,\sigma_\calT,\tau)$ is the system which simulates $\calT$, and that the scale factor of the simulation is $m$. In $\calT$, the seed is the single blue tile in the bottom-leftmost corner. In this system, since it is temperature-1, there is no cooperation. Also notice that there are $2n$ green tiles, $n$ at the top and $n$ at the bottom. This will insure that the scale factor $m$ must be greater than 1, because the simulating system only has $n$ distinct tiles and needs to map macrotiles to tiles in $\calT$.

Consider a few special assemblies in $\calT$. Let $\Vec{\alpha}$ be the assembly sequence starting with $S$ which grows right through the $B_i$ tiles (for $1\leq i\leq n$), then up to the $P_R$ tiles, then left through the $T$ and $P_L$ tiles and into $A$. Let $\Vec{\beta}$ be the assembly sequence starting with $S$ which grows up through $P_L$ and into $B$. Let $\Vec{\gamma}$ be the assembly sequence starting with $S$ which grows right through the $B_i$ tiles (for $1\leq i\leq n$), then up through $P_R$ and into $C$. Finally, let $\Vec{\delta}$ be the assembly sequence starting with $S$ which grows up to $P_L$, then right through the $T$ and $P_R$ tiles and into $D$. Since $\calT$ is temperature-1, each assembly in these assembly sequences are all $\tau$-stable and thus the sequences are valid. Since $U$ is a universal tileset, there are assembly sequences of tiles in $U$ which simulate these assemblies with a scale of $m$. For each of these assembly sequences in $\calT$, let $\Vec{\alpha}'$, $\Vec{\beta}'$, $\Vec{\gamma}'$, and $\Vec{\delta}'$ be the corresponding assembly sequences in $U$. For convenience, let $\alpha'$, $\beta'$, $\gamma'$, and $\delta'$ denote the earliest assemblies in the corresponding assembly sequences in $U$ which have a tile placed in the macrotile corresponding to $A$, $B$, $C$, and $D$ respectively. Furthermore, since there are only finitely many such assembly sequences and assemblies, we can suppose that $\Vec{\alpha}'$ is the assembly sequence which simulates $\Vec{\alpha}$ such that the number of tiles in $\alpha'$ is minimized. We say that this assembly is minimal and we can assume likewise for the other three assembly sequences.

Notice that $\alpha'$ and $\beta'$ must share at least one tile somewhere within a 1 macrotile distance of the macrotiles corresponding to $\beta$ since $\alpha'$ cannot grow around macrotile $B$ nor $\beta'$ around $A$ (because such growth would be outside of the allowable \emph{fuzz}). This is likewise true for $\gamma'$ and $\delta'$. Keep in mind that we are in the Planar aTAM model and this intersection between $\alpha$ and $\beta$ means that the space in the center of the assembly would be cut-off from the plane and no tiles would be able to grow there in the future. Now, notice that $\beta'$ and $\delta'$ must share at least one tile as well. If this were not the case, then either all of the tiles of $\delta'$ would be entirely to the right of the tiles in $\beta'$ or entirely to the left. Notice that this latter situation is impossible since $\delta'$ has to place tiles in macrotiles to the right of the tiles of $\beta'$. Furthermore, if all of the tiles of $\delta'$ were to the right of the tiles in $\beta'$, then there would be some tiles of $\delta'$ inside the region of space which is encircled by $\alpha'$ and $\beta'$. Since $\alpha'$ and $\beta'$ can grow independently from these tiles, it would be possible for the growth of $\alpha'$ and $\beta'$ to finish before those tiles of $\delta'$ could grow. This would lead to multiple terminal assemblies which is impossible since we assume that our simulator is directed. Thus there must be some tile shared by both $\beta'$ and $\delta'$. This is likewise true for $\alpha'$ and $\gamma'$. For convenience call, let $X_{\alpha \beta}$ be the tile shared by both $\alpha'$ and $\beta'$ and define $X_{\alpha \gamma}$, $X_{\delta \beta}$, and $X_{\delta \gamma}$ likewise.

$\alpha$ and $\delta$ both place $T_1,\ldots,T_n$, albeit from different directions, so it must be the case that $\alpha'$ and $\delta'$ contain tiles which span across the corresponding macrotiles. For convenience, we call the macrotile blocks corresponding to $T_1,\ldots,T_n$, along with the macrotiles immediately north and south of them, $\Omega$. To reach our contradiction, that $\calT$ cannot be simulated in a directed fashion, we will consider a sequence of tiles within this $\Omega$ region belonging to both $\alpha'$ and $\delta'$ and we will show that this sequence of tiles grows as part of $\alpha$ in the opposite direction as in $\delta$. We will then show that this must lead to tiles which must be attached within an enclosed area in order to preserve directedness, which is impossible because of the planar constraint.

First consider the set of tiles belonging to $\delta'$ with the smallest $y$ coordinate for each column in $\Omega$. Let $P^\delta$ be the sequence of these tiles organized from smallest to largest $x$ coordinate. Since there are $n$ macrotiles across the width of $\Omega$ and a macrotile is an $m \times m$ square of tiles, $P^\delta$ will contain $m n$ tiles. Let $P^\delta_i$ be the $i$th tile in $P^\delta$ where $1\le i\le m n$. Notice that no tile could ever grow at a smaller $y$ coordinate in $\Omega$ than any tile in $P^\delta$ because otherwise it would be within the region bounded by $\delta'$ and $\gamma'$ and would thus lead to undirected growth. Keep in mind that $P^\delta$ may not be a contiguous sequence of tiles as can be seen in Figure~\ref{fig:DPaTAM-bottom-seq}. Notice, however, that the order in which the tiles appear in $P_\delta$, namely from smallest $x$ coordinate to largest, corresponds to the order in which the tiles are placed in $\Vec{\delta}'$. We will show this using induction. First suppose, for contradiction, that $P^\delta_1$ was placed after $P^\delta_j$ for some $1<j\le m n$. If this were true, there would necessarily be a contiguous path of tiles from $X_{\delta\beta}$ to $P^\delta_j$ which goes above $P^\delta_1$. $P^\delta_1$ would then be in a closed off region where it would not be able to influence or inhibit the growth of tiles in the $D$ macrotile at the end of $\delta'$ since it cannot grow underneath $P^\delta_j$. Thus, if during the assembly sequence $\Vec{\delta}'$, tile $P^\delta_1$ was never placed, it would still be possible to reach the $D$ macrotile and thus $\Vec{\delta}'$ does not properly minimize the size of $\delta'$. Therefore $P^\delta_1$ must be placed before all tiles in $P^\delta$ during the assembly sequence $\Vec{\delta}'$. Proving the inductive case is similar. Suppose that the order is correct for tiles in $P^\delta$ up to $P^\delta_k$ and suppose that $P^\delta_j$ is placed before $P^\delta_{k+1}$ with $k+1<j\le m n$. There must exist a contiguous path of tiles connecting $P^\delta_k$ to $P^\delta_j$ which encloses $P^\delta_{k+1}$ since it cannot grow underneath either $P^\delta_k$ to $P^\delta_j$. Thus $P^\delta_{k+1}$ would not be able to influence the growth of macrotile corresponding to $D$ and thus $P^\delta_{k+1}$ is not part of a minimal $\delta'$. This shows that $P^\delta$ is grown in order and we can define $P^\alpha$ likewise except that the order of $P^\alpha$ is from largest $x$ coordinate to smallest.

Now suppose that $P^\delta$ contains a tile which is not in $P^\alpha$. We quickly show that this is impossible and that any tile in $P^\delta$ must be identical to the tile in $P^\alpha$ with the same $x$ coordinate. If a different tile existed and it were below a tile in $P^\alpha$, as stated previously, it might be that $\alpha'$ and $\beta'$ finished growing before that tile could attach and thus the tile would attach in a constrained subspace. This would lead to undirectedness. If it were above, then the opposite could be true with a tile of $P^\alpha$ being below a tile in $P^\delta$. This, along with the symmetrical argument for all tiles in $P^\alpha$ being in $P^\delta$ show that $P^\delta$ and $P^\alpha$ share all of their tiles. Moreover, because of the previous argument this implies that the order in which those tiles are placed during $\Vec{\delta}'$ and $\Vec{\alpha}'$ are opposite.

\begin{figure}
\centering
\includegraphics[width=3in]{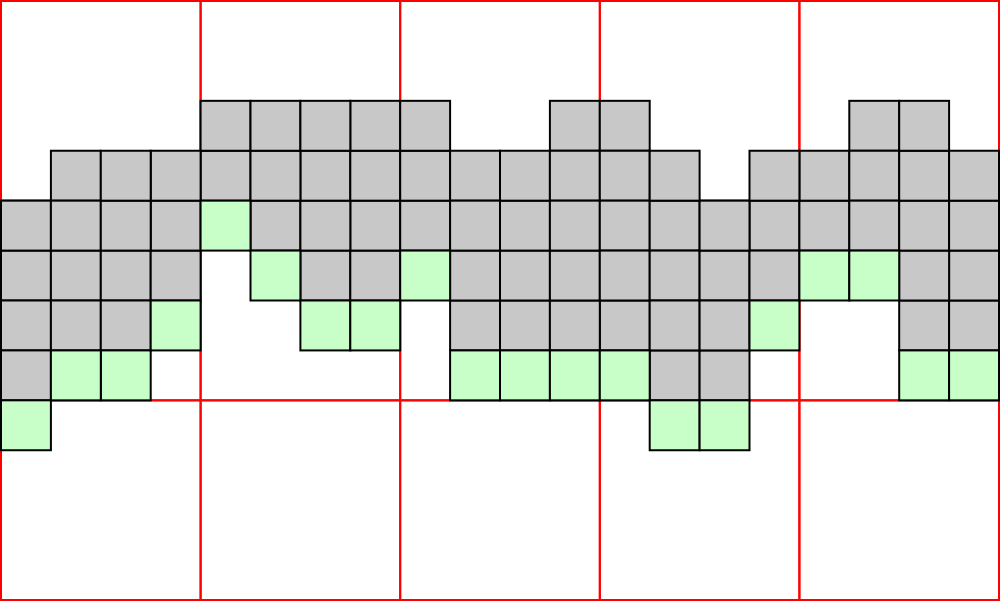}
\caption{Given the tiles of $\delta'$ in $\Omega$, the sequence $P^\delta$ is made up of the green colored tiles at the bottom of each column.}
\label{fig:DPaTAM-bottom-seq}
\end{figure}

Now suppose that $P^\delta$ consists of a straight horizontal line of tiles with no change in $y$ coordinate. In this case, there are two possibilities: either there is some tile in $P^\delta$, say $P^\delta_i$, such that the tile immediately north of $P^\delta_i$ is placed before $P^\delta_i$ during $\Vec{\delta}'$; or no such tile exists. If such a tile does exist, then it is possible for the tile north of $P^\delta_i$ to grow before $P^\delta_i$ and then for $P^\delta_{i+1}$ to grow from the other side following $\Vec{\alpha}'$. Since we assumed that all tiles in $P^\delta$ have the same $y$ coordinate, this would mean that $P^\delta_i$ is inside a constrained subspace and thus cannot attach leading to multiple finite assemblies. If no such $P^\delta_i$ existed, then the only way for each tile in $P^\delta$ to attach is using a $\tau$-strength glue from the previous tile in $P^\delta$. Since the scale factor must be at least $2$ and $\Omega$ is at least $n$ macrotiles across, this would lead to a pumpable line of tiles since there are only $n$ tiles in $U$. This could lead to arbitrary horizontal growth and means that the system would not be directed.

Thus, there must be at least one tile in $P^\delta$ with a different $y$ coordinate than the other tiles. This implies that there is at least one tile which is at a smaller $y$ coordinate than the tile before or after it in $P^\delta$. Suppose, without loss of generality, that tile $P^\delta_i$ is at a smaller $y$ coordinate than tile $P^\delta_{i+1}$. Notice that $P^\delta_i$ must have at least two adjacent tiles. If this were not the case and it only had a single adjacent tile, then it would only be bound by a $\tau$-strength glue to that one adjacent tile. Such a tile, however, could not be part of a minimal $\delta'$ since its removal would not affect any other tiles. Further notice that $P^\delta_i$ cannot have a tile to its south since it is the tile with the smallest $y$ coordinate of that $x$ coordinate in $\Omega$. Moreover $P^\delta_i$ cannot have a tile to its east since $P^\delta_{i+1}$ has the smallest $y$ coordinate with that $x$ coordinate and is at a higher $y$ coordinate than $P^\delta_i$. So $P^\delta_i$ must have a tile to its north and west. Let $t_n$ and $t_w$ be these tiles respectively. Suppose now, for contradiction, that during $\Vec{\delta}'$, the placement of $t_n$ happened before the placement of $t_w$. This would imply that there is a contiguous path of tiles connecting $P^\delta_{i-1}$ to $t_n$ before the growth of $P^\delta_i$. This would mean that the attachment of $P^\delta_i$ is unnecessary for further growth of $\delta'$ since no tiles growing from it could grow around $t_n$ or $P^\delta_{i-1}$. This means that $\delta'$ was not minimal and thus during $\Vec{\delta}'$, the placement of $t_w$ must happen before the placement of $t_n$. A similar argument shows that during $\Vec{\alpha}'$, the placement of $t_n$ must happen before the placement of $t_w$. Thus if we follow $\Vec{\delta}'$ up to the point where $t_w$ has attached and then follow $\Vec{\alpha}'$ to the point where $t_n$ has attached, $P^\delta_i$ will be in an enclosed region such that by the planarity constraint, it cannot grow. This means that the simulating system is not directed which is a contradiction.
Thus no such tileset $U$ can exist.

\end{proof}

\section{Design and Implementation of 3D aTAM Construction} \label{sec:construction}

In this section, we give more thorough description of the modules and growth process of our construction from Section~\ref{sec:construction-short} and prove lemmas regarding the functionality of these pieces. The lemmas will be put together into an overall proof of Theorem~\ref{thm:3DaTAMIU} in Section~\ref{sec:thm1-proof} and low-level technical details of the construction's implementation will be provided in Section~\ref{sec:low-level}.

The number of 3D aTAM systems is infinite. In order for a single, constant-sized tile set $U$ to simulate these systems, it is necessary for each simulating system to contain within its seed an encoding of the full tile set being simulated, and each tile $t$ of the simulated system must be represented by a macrotile containing an arrangement of tiles from $U$ which specify the type of $t$ under the representation function $R$. We will explain our construction by describing how that information is propagated into growing macrotile locations, as well as the logical \emph{modules}, or functional sub-assemblies, which form within each location which may grow into a new macrotile. These modules perform the necessary transfers and combinations of information as well as computations that determine which tiles should be represented and which information should be further propagated to neighboring locations. Along with the encoding of the simulated tile set $T$, 
encoded within the seed of the simulator is a variety of information (e.g. dimensions, relative locations, etc.) which describes how the modules specific to the simulated system $\mathcal{T}$ are constructed. All of this information together is called the $\genome$, as it specifies the full set of information needed for the macrotiles grown by the generic tiles of $U$ to form, or differentiate into, macrotiles which represent specific tiles of $\mathcal{T}$.

Let $\alpha \in \prodasm{\mathcal{T}}$ be an arbitrary producible assembly (possibly the seed) of $\mathcal{T}$, $\beta \in \prodasm{\calUT}$ such that $R^*(\beta) = \alpha$ (i.e. $\beta$ is a producible assembly in $\calUT$ which maps to assembly $\alpha$ in $\calT$ under representation function $R^*$) and let $l$ be a location in $\alpha$'s frontier. We will discuss the growth of tiles in $\calUT$ into $\beta^m_l$, which is the location of the $m$-block macrotile in $\calUT$ representing the frontier space $l$, and which we will refer to simply as $L$.  Without loss of generality, assume that given $\beta$, no further tile attachments can occur in any of the occupied $m$-block macrotile locations adjacent to $L$ (i.e. the adjacent macrotile locations are currently, but perhaps temporarily, ``complete'') and that no tiles have been placed inside of $L$.  We will call the neighboring macrotile locations $s_d$ for $d \in \{N,S,E,W,U,D\}$. We say that a macrotile $L$ \emph{differentiates} into a tile type $t$ at the first point in time in which $L$ no longer maps to empty space under representation function $R$ and instead maps to a tile of type $t$.

Assume that for some $s_d$, the macrotile in that location is complete and represents a tile in $t$.  An important concept in understanding the growth of a macrotile is the use of \emph{datapaths} and \emph{guide rails}.  These gadgets grow from one point in the simulation of $\calUT$ to another while encoding binary information about a tile type, a glue strength, etc.  While explained in more detail in Section~\ref{sec:growth-patterns}, the difference between the two is essentially that a datapath encodes a set of ``instructions'' that navigate it to specific coordinates in the macrotile while guide rails use blocking and cooperation with preset structures to navigate.  Using these gadgets, the growth of tiles within $L$ can be thought of as occurring in three stages: setup, computation, and differentiation.

The first stage of growth, setup, consists of a copy of the $\genome$ which will wrap around the exterior of the macrotile space in three concentric bands (see Figure~\ref{fig:intercube_genome}). The $\genome$ grows by following a set of instructions embedded within a specific portion of the $\genome$ which guide it to a specific, predetermined location within $L$.  This location is fixed within $L$ (and is in the same relative location in all macrotiles) and the instructions used to get here are specific to which direction the neighbor $s_d$ is from $L$.  Once the $\genome$ has grown up to a certain point, another section of instructions within the $\genome$ are activated (i.e. those instructions begin to control new growth). These instructions spur growth that is responsible for building a series of modules which will process input glue information (from $s_d$ and any other neighbors which may provide it in the future) to determine if and when the location of $L$ has received enough input glues from neighboring locations to select and resolve into a tile from $T$.  Once the growth of these modules is complete, the input glue information from neighboring macrotiles can be fed into them. These modules are called the $\adderArray$, $\bracket$, and $\externalCommunication$.

\begin{figure}[htb]
\centering
\includegraphics[width=4.5in]{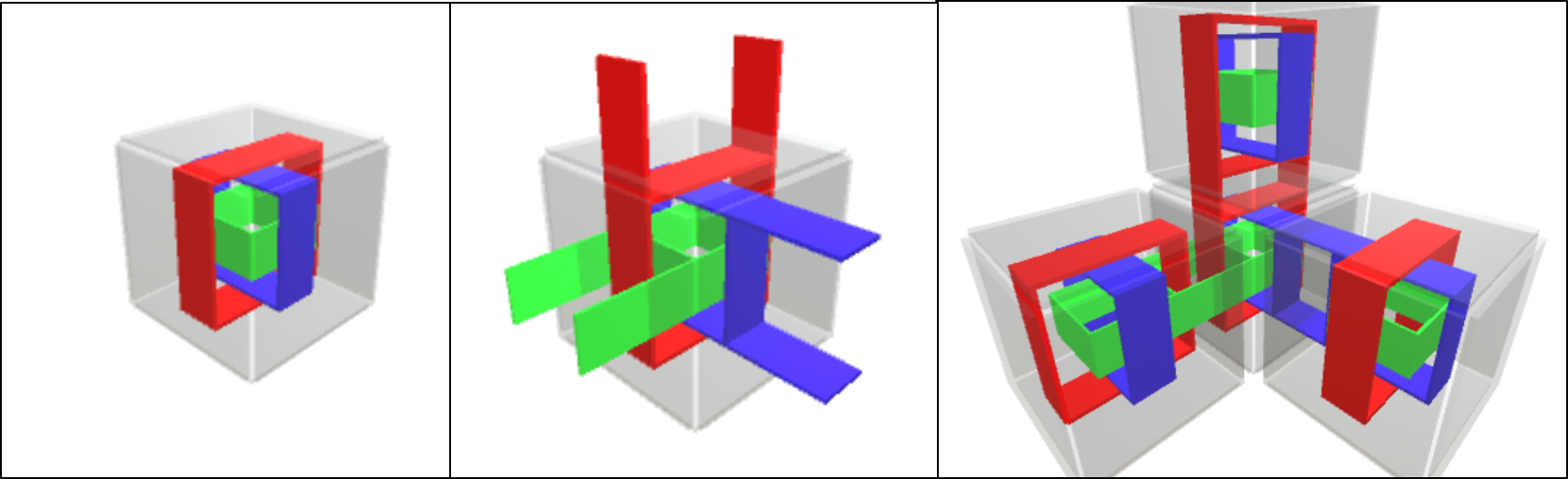}
\caption{(left) Bands of the $\genome$ before the macrotile begins to output the $\genome$ to neighboring macrotiles, (middle) depiction of growth outward from the $\genome$ bands, (right) full connections between the $\genome$ bands of 3 neighboring macrotiles.}\label{fig:intercube_genome}
\end{figure}

The second stage of growth is the computation stage. This begins with neighboring macrotiles ($s_d$ and any other neighbors that have already resolved to tiles of $T$) sending datapaths that encode their respective tile types to $L$. Each such datapath interacts with the $\genome$ in $L$ to find, or \emph{query}, which tile types in $T$ could potentially form a bond with the neighboring macrotile in the direction of that neighbor. If such a tile type exists, this spurs a new datapath that encodes the strength of the potential bond and which grows to the first module, called the $\adderArray$. This module sums up the strength values from these datapaths in order to determine which neighboring glues have enough collective strength to allow for the attachment of specific tile types from $T$.  When the $\adderArray$ determines that a tile type has enough input glues for its placement to be simulated (i.e. a tile of that type could have attached in the corresponding location in $\calT$), it outputs a guide rail to the next module, called the $\bracket$.  This module is necessary in the case that multiple such tile types in $T$ exist, and therefore $L$ could potentially resolve into multiple unique tile types. It works by moving the paths that encode these tile types through a competition in which only one ``wins''. This winning path encodes the tile type that $L$ resolves into.  Because the representation function $R$ simply looks at this encoding, the end of the $\bracket$ is where $L$ differentiates into this winning tile type. This initiates the final stage of growth.

The differentiation stage begins when the output path from the $\bracket$ grows to the final module, called the $\externalCommunication$. This module first sends a signal back to the $\genome$, prompting it to grow into all neighboring macrotiles. Additionally, the $\externalCommunication$ module sends a datapath encoding the tile type that $L$ represents to all neighboring macrotiles. This initiates the growth process in any empty neighboring macrotiles (and merges seamlessly into those which are not empty) and helps to continue the growth process in neighboring macrotiles that need more cooperation in order to differentiate. 

\subsection{Growth process overview}\label{sec:full-growth-process}

The full growth process of a macrotile may be thought of as occurring sequentially in the order specified below. Note that steps \ref{step:genome-growth}-\ref{step:query} can overlap to some degree, however, whenever an unfinished module isn't in place at an early enough time, growth of an interacting piece will simply stall until the unfinished piece grows to a point in which they cooperate to allow further growth. Steps \ref{step:genome-prop} and \ref{step:ext-comm} are similar. While the growth isn't strictly sequential, it can be thought of as such (see Lemma~\ref{lem:timing}).

\begin{enumerate}[topsep=0pt,itemsep=-1ex,partopsep=1ex,parsep=1ex]
\item $\genome$ growth: Initiated by a neighbor, the $\genome$ propagates around the macrotile to form a full, three concentric band structure (see Figure~\ref{fig:intercube_genome} for an image and Section~\ref{sec:genome} for details).\label{step:genome-growth}

\item Initialization: Once the $\genome$ has grown up to a certain point, it ``seeds'' the other major modules, i.e. initiates the growth necessary for them to perform their functions in later steps.\label{step:init}
\begin{enumerate}[topsep=0pt,itemsep=-1ex,partopsep=1ex,parsep=1ex]
\item The $\externalCommunication$ is seeded (see Section~\ref{sec:external_communication} for details).
\item The $\bracket$ is seeded (see Section~\ref{sec:bracket} for details of this module).
\item The $\adderArray$ is seeded (see Section~\ref{sec:adderArray} for details of this module).
\end{enumerate}

\item Query: Whenever a neighboring macrotile differentiates, it sends an \emph{external communication} datapath containing a binary encoding of that neighboring macrotile's tile type into the currently growing macrotile. This datapath grows to a specified portion of the $\genome$ where it performs a \emph{query} operation. This works by finding (via lookup tables encoded in the $\genome$) the tile types that could bind to that neighbor. Any tile types that meet this condition spawn new datapaths encoding the strength of the bond that could form, and grow from the $\genome$ to the $\adderArray$.\label{step:query}

\item $\adderArray$ success:  The $\adderArray$ contains an adder specific to each tile type of $T$. If one of those adders determines that enough adjacent glues are present in neighboring macrotiles for the attachment of the tile type which it uniquely represents to be simulated, it grows input to the unique location in the $\bracket$ for that tile type.\label{step:adder}

\item $\bracket$ competition: One or more entries to the $\bracket$, each containing the encoding of a unique tile type in $T$, grow through toward the end, potentially competing in a series of pairs to be the first to grow into \emph{points of competition}, with only one ultimately reaching the end of the $\bracket$.  Once this occurs, the macrotile officially represents, via the representation function $R$ (see Section~\ref{sec:rep-function} for details), the corresponding tile type in $T$.\label{step:bracket}

\item $\genome$ propagation: A datapath is grown to the $\genome$ bands which propagates around and causes them to grow branching copies of the $\genome$ into each of the neighboring macrotiles.\label{step:genome-prop}

\item $\externalCommunication$: The identity of the tile type for this macrotile is propagated to each neighboring macrotile by a datapath which guides it to the location of the critical orientation of the neighbor's $\genome$ (i.e. a special portion of the $\genome$ that handles initialization and queries) corresponding to the correct query location for a glue from the relative direction of this macrotile. This effectively outputs the type of this tile to its neighbors.\label{step:ext-comm}
\end{enumerate}

For a given macrotile location $L$, if $L$ maps to a tile in $\calT$ and is contained in a terminal assembly of $\calUT$, then the macrotile at $L$ will have grown through all \ref{step:ext-comm} steps. If $L$ doesn't map to a tile in the terminal assembly but a neighboring macrotile location does, then the macrotile at $L$ will grow through the first \ref{step:query} steps. If neither $L$ nor any of its neighbors map to a tile in the terminal assembly, then the macrotile at $L$ will be empty.

\begin{figure}[htp]
\centering
\begin{subfigure}[b]{3.2in}
\label{fig:full-macrotile1}
\includegraphics[width=3.12in]{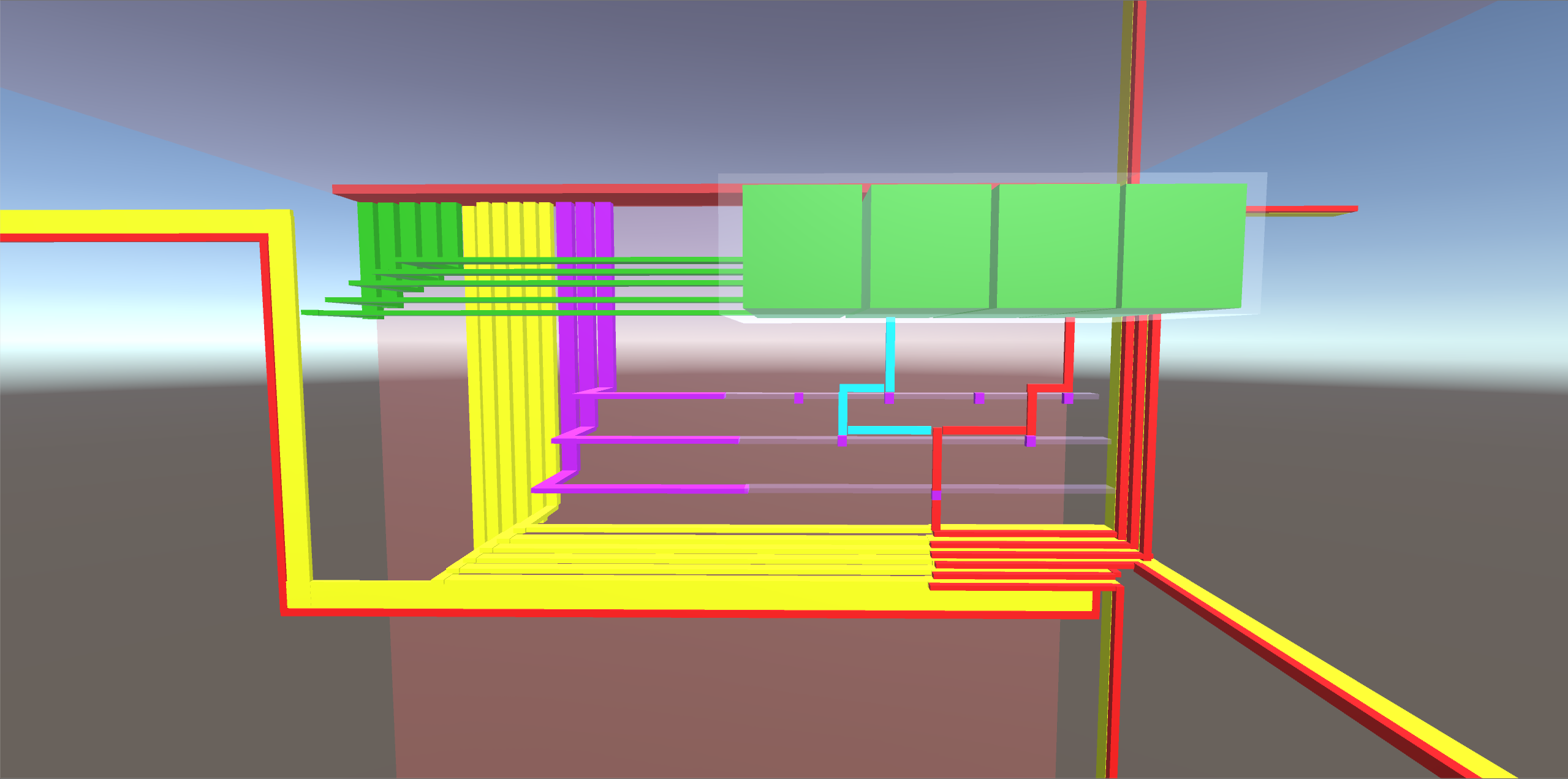}
\caption{Rotation 1}
\end{subfigure}
\begin{subfigure}[b]{3.2in}
\label{fig:full-macrotile2}
\includegraphics[width=3.12in]{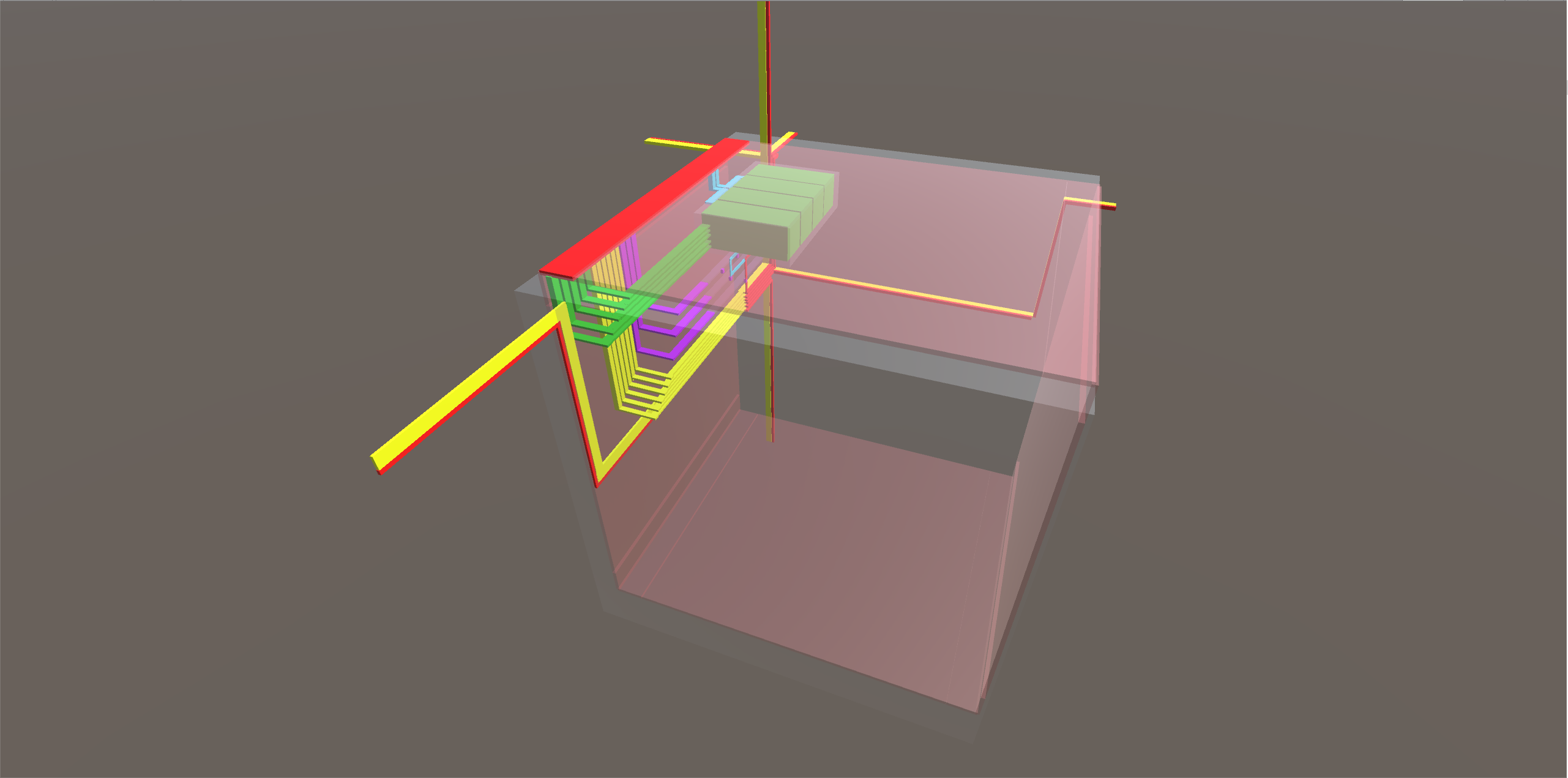}
\caption{Rotation 2}
\end{subfigure}
\caption{Schematic view of the major components of a macrotile excluding the bands of the $\genome$ other than the portion in the critical orientation (red, top). The green boxes represent the $\adderArray$, the purple section below it represents the $\bracket$, and the blue and red paths from the $\adderArray$ through the $\bracket$ represent two tile type representations competing with the red winning and initiating the $\externalCommunication$ paths which are yellow and red.  The paths from the critical orientation of the $\genome$ to components of the same color represent the datapaths which seeded the growth of those components.}
\label{fig:full-macrotile}
\end{figure} \hfill

\subsection{Representation function and seed generation}\label{sec:scale}

Given an arbitrary 3D aTAM system $\calT = (T,\sigma,\tau)$, in order to generate $\calUT = (U,\sigma_{\calT},\tau')$, we must compute (1) the scale factor $m$, (2) the structure of the $\genome$, (3) the seed structure $\sigma_{\calT}$, and (4) the representation function $R$ which maps macrotiles of $\calUT$ to tiles in $\calT$.
Discussion and details of how $m$ and the $\genome$ are computed can be found in Sections~\ref{sec:genome} and \ref{sec:scale-details}.  In this section, we briefly discuss how $R$ is created and our algorithm which implements the generation of the seed.  For more details on the calculations and algorithms, please see Section~\ref{sec:scale-details}.

\subsubsection{Macrotile representation function $R$}\label{sec:rep-function}

The tiles of $T$ are given an arbitrary but fixed ordering.  Based on that ordering, each tile type is assigned a number $0 \le i < |T|$.  The binary encoding of the tile type's number $i$ padded to length $\lceil \log_2(|T|) \rceil$ is then used as the \emph{binary representation} of that tile type throughout the construction.  To determine the representation of a macrotile $L$, the macrotile representation function $R$ simply looks at the coordinates of the $\lceil \log_2(|T|) \rceil$ locations for the end of the $\bracket$'s output (which is a fixed set of coordinates for all macrotiles based on the value of $m$).  If any of those spaces are empty, the entire macrotile maps to empty space.  However, if all of those locations are occupied by tiles, then they will encode the binary representation of a tile type in $T$ (see Lemma~\ref{lem:bracket}), and this is what $R$ maps the macrotile to.

\subsubsection{Structure of macrotiles representing seed tiles}\label{sec:seed}

For each location $l$ and tile $s$ in $\sigma$, a macrotile $S$ is created in $\sigma_{\calT}$. Each $S$ contains the following:

\begin{enumerate}[topsep=0pt,itemsep=-1ex,partopsep=1ex,parsep=1ex]
\item A single but complete row of the $\genome$ at a specific intersection point of the bands.
\item Tiles encoding the binary representation of the tile type in $T$ which this macrotile represents in the positions at the end of the output of the $\bracket$.  Note that a special tile type is used for the bottom-most such tile so that it will grow in the forward direction out of the $\bracket$ but not backward into it.
\item A single-tile-wide path of tiles which connect the row of the $\genome$ to the $\bracket$ output tiles.  This path of tiles contains strength-1 glues between all adjacent tiles so that each is attached with exactly strength 2, and each contains no additional glues on any other sides.  This connecting path is simply to ensure that $\sigma_{\calT}$ is a stable assembly.
\item A set of special \emph{blocking tiles} which are located in the $6$ positions in which $\genome$ queries would otherwise be able to begin, and $2$ more which are located in the positions from which the $\genome$'s initialization of the $\adderArray$ and $\bracket$ would begin.  These blocker tiles prevent $\externalCommunication$ datapaths from initiating queries in a seed tile macrotile and also prevent the macrotile from growing the $\adderArray$ and $\bracket$.
\item A single-tile-wide path of tiles which connects the row of the $\genome$ with each of the blocking tiles, in a similar fashion and for the same reason as the path to the $\bracket$ output tiles.
\end{enumerate}

If $\sigma$ contains more than one tile, then the final addition to $\sigma_{\calT}$ is that the macrotiles are given an ordering and a connecting path of tiles links their $\genome$ rows in that sequence, again to ensure that the seed assembly is stable.  Note that, for all of the connecting paths, it is simple to ensure that they don't occupy space that could be used by other components which may grow in a macrotile.

Seed macrotiles in $\calUT$ are different from other macrotiles in several ways.  First, the six input query regions on the critical orientation of the $\genome$ are blocked to ensure that no later input can ever initiate a query.  Second, the $\adderArray$ and the $\bracket$ are prevented from growing.  Only the $\externalCommunication$ datapaths will be initialized. Once they grow to the necessary location to receive output from the $\bracket$, the tiles that were placed there as part of the seed to represent the tile type of the macrotile then cooperate to begin growth of (1) the signal that travels back along those datapaths to the $\genome$ and initiates the growth of the $\genome$ into neighboring macrotiles and (2) the $\externalCommunication$ datapaths to all neighboring macrotiles.

\subsection{Independence of timing of growth}

\begin{figure}[htb]
\centering
\includegraphics[width=\textwidth]{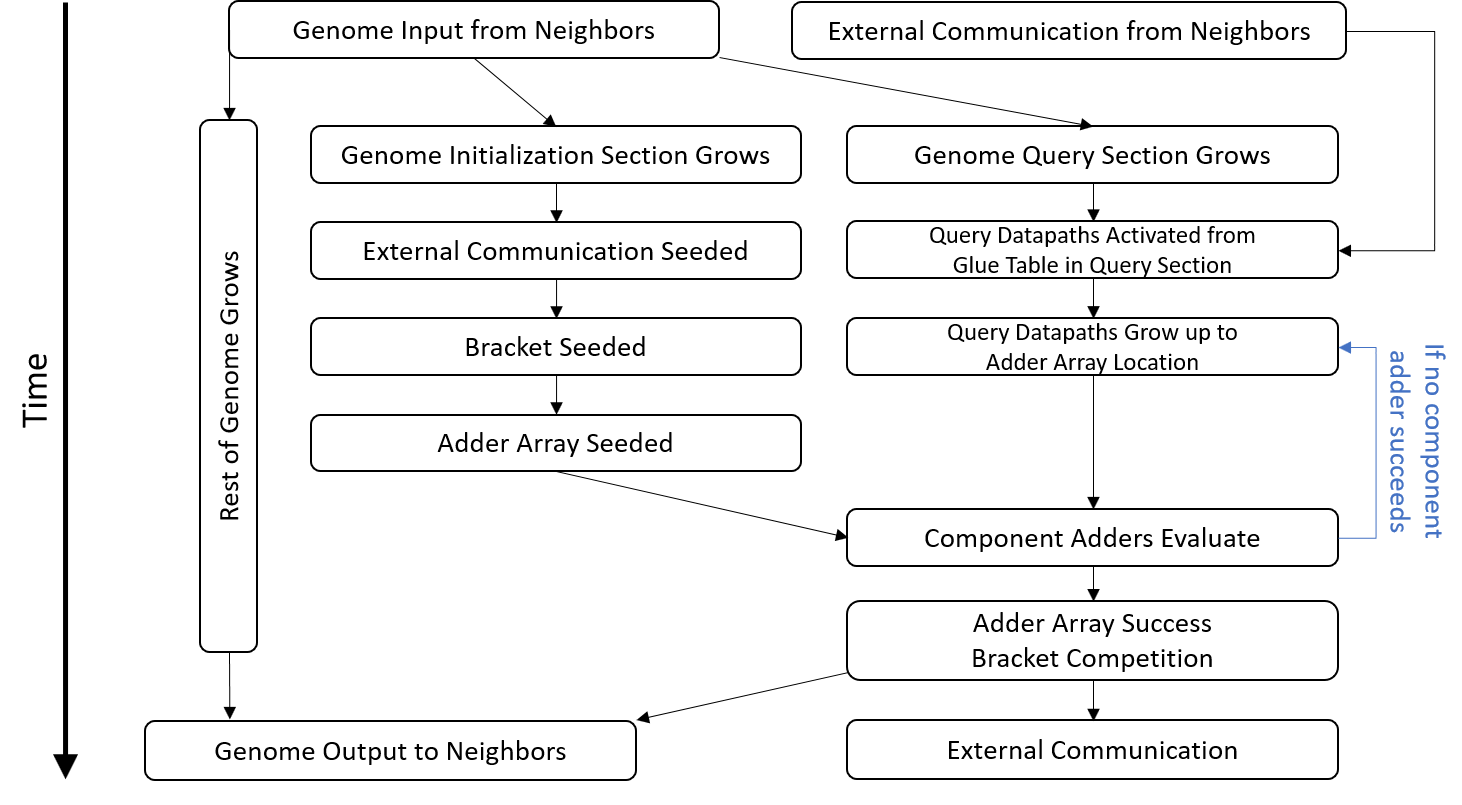}
\caption{A diagram of the time dependencies in the full growth process of a macrotile. Events at the beginning of an arrow must happen before an event at the end of an arrow can happen.}\label{fig:time_dependencies}
\end{figure}

\begin{lemma}\label{lem:timing}
Let $L$ be a macrotile in $\calUT$. Whether the modules within $L$ grow in the order specified in Section~\ref{sec:full-growth-process} or a different order, the set of all possible tile placements within macrotile $L$ remains the same.
\end{lemma}

\begin{proof}
As illustrated in Figure~\ref{fig:time_dependencies}, there are four sequential sub-processes of growth within the full growth process of a macrotile. The first sub-process is the growth of the $\genome$ outside of the critical orientation. The second sub-process is the initialization of the other modules from the $\genome$. The third sub-process is the growth of the query section of the $\genome$. The fourth sub-process is the growth of external communication datapaths from neighbors. There are three points at which two of these paths must merge (i.e. the components must interact). Because the sub-processes leading up to these points are all sequential, these ``merge points'' are the only events within the full growth process in which we must show that the arrival of one path before another cannot cause different tile placements to occur.

The first merge point is the activation of query datapaths from the glue table of the $\genome$. This requires cooperation from a query row in the $\genome$ and an $\externalCommunication$ datapath from a neighboring macrotile. Regardless of which component grows to the designated point of cooperation first, they can only interact using cooperation, meaning no additional growth can happen until both components are present, and therefore both timings result in the same tile placement.

The second merge point is the evaluation of a $\componentAdder$. This requires cooperation between the ``seeded'' part of the $\adderArray$ (i.e. the portion whose growth is initiated by the $\genome$'s initialization section and which can grow before any input arrives) and the query datapaths from the $\genome$. Again, because this requires cooperation, the evaluation of a specific $\componentAdder$ cannot happen until both the seeded parts of the $\adderArray$ and the query datapaths are present, after which the same tiles will be placed, regardless of which component arrived first.

The final merge point is the propagation of the $\genome$ to neighbors. This requires cooperation between the $\genome$ and a signal that is activated once a guide rail wins the $\bracket$ and grows to the $\externalCommunication$. The $\genome$ presents strength-1 glues that will begin propagation to neighbors, but the signal provides a strength-2 glue that is required to initiate the growth. Without the completed $\genome$, the signal will stall on the partial $\genome$ until more fills out. Without the signal, the output regions of the $\genome$ will stall until the signal is provided. Therefore, regardless of which is present first, the same tile placement ensures.
\end{proof}

A consequence of Lemma~\ref{lem:timing} is that, in later proofs, since different ordering always yield the same tile placements for the majority of macrotile growth (and in fact all macrotile growth if $\calT$ is directed), we can always assume that growth happens in the preferred ordering given earlier in this section.

\subsection{Genome}\label{sec:genome}
The $\genome$ module is present in partial and completed macrotiles and contains the entirety of information about both the simulated tile set $T$, the temperature $\tau$ of $\calT$, and the information (dimensions and directions) needed to construct the macrotiles of $\calUT$. It consists of three sections: (1) a series of instructions used to grow the bands of the $\genome$ itself, referred to as $G_1$, (2) information related to the definitions of the glues and tiles in $T$, referred to as $G_2$, and (3) the information necessary to construct the $\adderArray$, $\bracket$, and $\externalCommunication$ modules with dimensions and locations dictated by the scale factor of the simulation, referred to as $G_3$.

\subsubsection{Structure}

The smallest unit of the $\genome$ is a \emph{row}, a linear set of tiles. All of the information that the $\genome$ represents is contained within each row. However, the overall structure of the $\genome$ is three connected, concentric bands each made up of multiple rows, as illustrated on the left in Figure \ref{fig:intercube_genome}. The three bands are connected where they overlap. While the system definition information encoded in the rows of the $\genome$ is only utilized by other modules on the ``up'' face of the innermost band (which we also refer to as the \emph{critical orientation} of the $\genome$), the full band of the $\genome$ is necessary to transport the information in the $\genome$ from a neighboring macrotile in any direction to the place where queries which utilize the information will eventually happen. (Note that during the proof of Theorem~\ref{thm:directed3DIU} this will be important because this also works to preserve directedness of $\calUT$, since the completed structure doesn't indicate which neighbor initiated the growth of the $\genome$ structure.)

We refer to the corners of each band as \emph{intersections}. These are the points that connect the $\genome$ of one macrotile to another, as illustrated on the right in Figure \ref{fig:intercube_genome}. Whenever the growth of a $\genome$ is initiated by a neighboring macrotile, it grows to a designated ``input'' intersection that begins the propagation of the full $\genome$ structure. In the event that the macrotile eventually differentiates, the $\externalCommunication$ module sends a callback signal, which will be further described in Section~\ref{sec:callback}, that circulates over the $\genome$ and initiates the propagation of the $\genome$ to the six neighbors through six designated ``output'' intersections. Note that, if another neighbor tries to propagate the $\genome$ and the current macrotile already has a full $\genome$ structure, the incoming path of tiles will seamlessly merge with the existing $\genome$ bands via an ``input'' intersection.

\subsubsection{$G_1$ - $\genome$ movement}

Whenever a macrotile location finishes the process of differentiation (i.e. determining the tile type from $T$ that it simulates), it initiates the growth of the $\genome$ into all of the neighboring macrotile locations. As the $\genome$ begins to grow into an empty macrotile location, it uses the information encoded in the first section, $G_1$, to propagate around the entire macrotile in the concentric bands structure. (While only the information on the critical orientation of the $\genome$ is used to control growth of the macrotile's additional modules, the entire structure is propagated in order to preserve directedness when necessary, since the $\genome$ growth can be initiated by any of the six neighboring locations, as will be discussed for the proof of Theorem~\ref{thm:directed3DIU}. Details of the mechanisms which ensure this can be found in Section~\ref{sec:collision-tolerance}.) The information is encoded as a series of instructions which are executed one at a time as the $\genome$ grows, and specific instructions are executed for each band.  A schematic representation of the bands of the $\genome$ can be seen in Figure~\ref{fig:intercube_genome}.

\subsubsection{$G_2$ - Glue table}\label{sec:glue_table}
The second section of the $\genome$ includes information about the tiles in $T$. The section has up to $6|T|^2$ entries, one for each direction cross each pair of tiles if that pair of tiles can bind in that direction, all separated by delimiters. (If a pair of tiles cannot bind in a given direction, there is no entry.) Entries are organized into the following hierarchy.

\begin{enumerate}
\item \emph{Highest Level} - a group that represents the tile type $t \in T$ represented by the neighboring macrotile
\item \emph{Middle Level} - a subgroup that represents the direction of the neighboring macrotile from the currently growing macrotile
\item \emph{Lowest Level} - individual entries that represent the tile types of $T$ which could bind to $t$ in the specific direction 
\end{enumerate}

This hierarchy can be seen in Figure \ref{fig:GlueTableDetailed}. Each entry in the glue table includes a row of tiles encoding instructions that, once ``activated'' (i.e. when tiles grow into a specific adjacent location, allowing cooperative tile attachments to initiate further growth), will initiate growth of new datapaths to the $\adderArray$. We call the growth of a datapath across the $\genome$ for this purpose, and then to the $\adderArray$ (when necessary) a \emph{query}.  There are $6$ rows in the critical orientation of the $\genome$ that are specified for queries, one for each direction.  A depiction of how the section of the critical orientation grows so that there is a row specific to each direction can be seen in Figure~\ref{fig:InstructionsByBand}.  Whenever an $\externalCommunication$ datapath grows from a neighboring macrotile to the designated row for its direction, it grows along this section, counting along the delimiters until it finds both (a) the group that represents the tile type represented by the neighboring macrotile that it grew from and (b) the subgroup that represents the direction that it grew from. Here, it activates the datapaths for each tile type that could attach to the adjacent glue represented by that neighboring macrotile, each of which grows to the correct portion of the $\adderArray$ for that tile type. Embedded within these datapaths is the information about (a) the strength of the bond that can form between the neighbor and the potential tile type that the datapath represents and (b) the relative location of the specific component within the $\adderArray$ that corresponds to the same tile type as the datapath.

\begin{figure}[htb]
\centering
\includegraphics[width=6.5in]{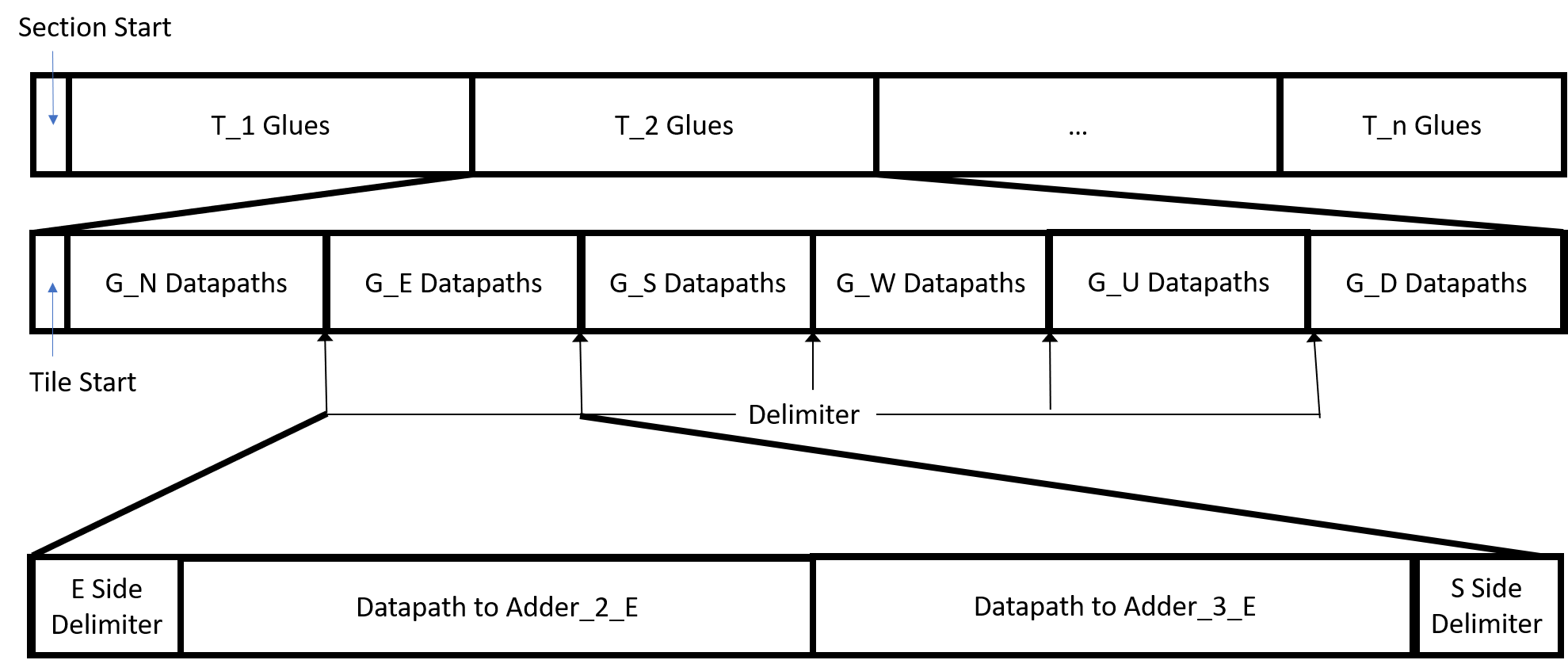}
\caption{Layout of the glue table portion of the $\genome$.  The top row depicts how the entire glue table is divided into sections with one for each tile of $T_i \in T$.  The middle row shows how each tile entry is actually composed of one section for each of the six directions, corresponding to the glues on those sides of $T_i$.  The bottom row shows how each entry for a direction is actually composed of a series of datapaths, where each encodes the directions needed for the datapath to grow to the section of the $\adderArray$ which corresponds to a tile which could bind to the specific glue.}\label{fig:GlueTableDetailed}
\end{figure}

\subsubsection{$G_3$ - Initialization}\label{sec:initialization}
The third and final section of the $\genome$ contains the instructions for how to ``seed'' the other major components of the macrotile, i.e. the $\adderArray$, $\bracket$, and $\externalCommunication$ modules. As soon as the $\genome$ has grown into the critical orientation, this section initiates the growth of datapaths that cause the growth of the general structure of these other modules so that they can perform the correct functions during later stages of macrotile growth. The type of growth through the internal components of a macrotile during differentiation makes the ordering in which these modules are created important (i.e. sometimes the growth relies on blocking tiles to be present in advance). Therefore, the initialization datapaths employ a protocol that we refer to as a \emph{callback} in which the datapath that seeds the $\externalCommunication$ module, once finished, will send a signal that initiates the datapath that seeds the $\bracket$, which, in turn, will send a signal that initiates the datapath that seeds the $\adderArray$. Therefore, the ordering of completion of these components always begins with the $\externalCommunication$ module, then the $\bracket$, and finally the $\adderArray$. (More information about callback functionality can be found in Section~\ref{sec:callback}.)  The seeds for the $\externalCommunication$ include $6$ datapaths which are positioned so that the first is adjacent to the final output location of the $\bracket$. This allows the output of the $\bracket$ to spur cooperative growth which initiates the growth of $\externalCommunication$ datapaths which navigate to the $6$ neighbors. Additionally, the first seeded datapath also provides the glues for a path whose growth is initiated by the $\bracket$ output which travels back to the $\genome$ then around the bands, causing them to output the $\genome$ into all $6$ neighboring macrotiles.

\subsubsection{Correctness of $\genome$}
In this section we show that the $\genome$ correctly performs the necessary functions specific to it.

\paragraph{Correctness of $\genome$ encoding}

\begin{lemma}[Genome correctness]\label{lem:genome-data}
Given an input system $\calT$, the $\genome$ is correctly generated so that it encodes all necessary data for $G_1$, $G_2$, and $G_3$.
\end{lemma}

\begin{proof}
Here we provide a high-level sketch of the correctness of Lemma~\ref{lem:genome-data}.  Details of the scale factor calculation as well as the algorithms for creating the $\genome$ and seed can be found in Section~\ref{sec:scale-details}. The online software implementation of these algorithms can be found at the address provided in Section~\ref{sec:implementation}.

The $\genome$ for simulating $\calT$ is essentially just a series of datapaths which encode directions and distances of growth and payloads to be delivered, separated by delimiters which dictate what each datapath corresponds to. The calculations of values for each specific datapath depend upon the scale factor of the simulation by $\calUT$, which in turn depends upon the size of the simulated tile set $T$.  However, the scale factor of the simulation also depends (to some extent) upon the widths of the datapaths needed to encode the distances to travel. Fortunately, this circular dependency can be easily resolved by attempting to overestimate the number of bits that will be needed to encode distances, using that estimate to determine datapath widths, computing the scale factor based on that estimate, then checking to ensure that the datapath widths are sufficient to encode distances for that scale factor.  If not, the estimate is increased until a sufficient scale factor is found.  Since the number of bits only grows logarithmically as the scale factor grows, it is easy to limit the number of iterations necessary.  The scale factor $m$ of the simulation is $O(|T|^2\log(|T|\tau))$. (Details of this analysis can be found in Section~\ref{sec:scale-details}.) Once the scale factor has been computed, and using the size of $T$ and value of $\tau$, it is relatively straightforward to determine the necessary spacing and location of macrotile components, and from that to compute the datapaths necessary to (1) initialize and grow macrotile components, and (2) resolve $\genome$ queries by growing to necessary $\adderUnit$s within the $\adderArray$, as this simply requires knowledge of the location of the individual $\adderUnit$s as well as an inspection of which pairs of tiles of $T$ have matching glues in each direction and the strengths of those glues.

Given the specifications of the necessary datapaths, the creation of the $\genome$ is a simple matter of placing the datapaths of $G_1$, $G_2$, and $G_3$ in a specified order separated by the correct delimiters. The time complexity of the entire seed generation is $O(|\sigma||T|^2\log(|T||\sigma|\tau))$.  (Details of this analysis can also be found in Section~\ref{sec:scale-details}.)
\end{proof}

\paragraph{Correctness of $\genome$ propagation}

\begin{lemma}[Genome growth]\label{lem:genome-growth}
Let $\beta \in \prodasm{\calUT}$ and $L \sqsubseteq \beta$ be a subassembly of $\beta$ such that $L$ is a single macrotile. If $L$ contains any complete row of an intersection point of the $\genome$, then the complete and correct bands of the $\genome$ will grow.
\end{lemma}

\begin{proof}
The proof of Lemma~\ref{lem:genome-growth} is based on the low-level design of the instructions which make up the $G_1$ portion of the $\genome$, which are responsible for the growth of the bands of the $\genome$.  By Lemma~\ref{lem:genome-data}, we know that the $\genome$ contains the set of instructions necessary to direct the growth of the bands with correct dimensions and locations.  The design of the instructions is such that, as growth is occurring which spreads into neighboring macrotiles, the first row which enters a neighboring macrotile is one of the specially designated intersection rows.  At this point, the set of instructions which become ``active'' are specific to the formation of the bands of the $\genome$ in a macrotile when started from the specific entry point of that intersection row.  While this is sufficient to ensure correct growth of the $\genome$ bands when initiated by any single neighbor, it is also necessary for the growth of those bands to be correct when initiated, possibly even simultaneously, by multiple neighboring macrotiles.  In order to ensure correctness when this occurs, the $\genome$ is designed so that (1) the intersection points for the $\genome$ between neighboring macrotiles are well defined, and (2) the design of the bands allows for bi-directional growth from those intersection points in such a way that, after the bands of the $\genome$s of a neighboring pair of macrotiles have fully formed, it is impossible to tell in which order they grew. This is accomplished by careful use of circular latches (see Section~\ref{sec:collision-tolerance}).  In order to guarantee that the growth of the bands of the $\genome$ result in the same tile placements irrespective of the particular intersection point from which they grew, they are completely ``collision tolerant'', meaning that even if growth is initiated via multiple different directions and partial bands grow toward each other, they ultimately merge and form bands which are identical to those which would have formed via growth from only a single source. Thus, given the existence of any single row of the $\genome$ at an intersection point within $L$, whether it is the only that ever grows or an arbitrary subset of the others grow in an arbitrary ordering, the complete and correct bands of the $\genome$ will grow.
\end{proof}

\begin{lemma}[Genome propagation]\label{lem:genome-prop}
Other than the $\genome$, assume that all the modules in the simulator $\calUT$ work correctly. Let $\alpha \in \prodasm{\mathcal{T}}$ such that $\sigma \sqsubset \alpha$ (i.e. $\alpha$ is larger than the seed) and $\beta \in \prodasm{\calUT}$ such that $R^*(\beta) = \alpha$. Let $l' \in (\dom{\alpha} \setminus \dom{\sigma})$ be a location in $\alpha$ outside of the seed, $l \not \in \dom{\alpha}$ be any location adjacent to $l'$ but outside of the assembly $\alpha$, and $L$ be the macrotile location in $\beta$ which maps to $l$. Then macrotile $L$ either already has completely grown the bands of the $\genome$ or it will.
\end{lemma}

\begin{proof}
Recall that $l'$ be the location of the non-seed tile in $\alpha$ which is adjacent to $l$, and let $L'$ be the macrotile of $\beta$ which represents $l'$.  Since $R^*(\beta) = \alpha$, then $L'$ maps, via $R$, to some tile type $t \in T$.  By definition of $R$, this means that the binary encoding of $t$ is contained in tiles at the end of the output of the $\bracket$ of $L'$. Since $L'$ does not represent a seed tile in $\sigma$ and thus was initially empty, the only way that it could have tiles in the output location of the $\bracket$ is for $L'$ to have correctly grown its $\genome$ into the critical orientation, grown its $\externalCommunication$ seeds, $\adderArray$, and $\bracket$, and to have received at least one $\externalCommunication$ query from a neighboring macrotile. Because the $\externalCommunication$ seeds must be complete, they will receive the output from the $\bracket$ and the path which grows back to the $\genome$ to initiate the propagation of the $\genome$ bands into all neighbors will grow.  The fact that this will result in the first row of the $\genome$ growing in an intersection location in $L$ along with Lemma~\ref{lem:genome-growth} guarantees that the complete bands of the $\genome$ will grow in $L$.
\end{proof}

\paragraph{Correctness of $\genome$ querying}

\begin{lemma}[Genome querying]\label{lem:genome-query}
Other than the $\genome$, assume that all the modules in the simulator $\calUT$ work correctly. Let $\alpha \in \prodasm{\mathcal{T}}$ and $\beta \in \prodasm{\calUT}$ such that $R^*(\beta) = \alpha$, $l \not \in \dom{\sigma}$ be a location outside of the seed of $\calT$ but which is adjacent to a tile in $\alpha$ of type $t_d$ in direction $d$, and $L$ be the macrotile location in $\beta$ which maps to $l$. A query datapath that encodes $s$ will grow in $L$ from the $d$ query row of its $\genome$ to an $\adderUnit$ that corresponds to a tile type $t \in T$ if and only if a bond of strength $s \in \mathbb{N}$ could form between a tile of type $t$ in location $l$ and a tile of type $t_d$ in the $d$ direction.
\end{lemma}

\begin{proof}
By Lemmas~\ref{lem:genome-prop} and \ref{lem:seed-growth}, since $L$ has a neighboring macrotile which represents a tile (which is outside or inside the seed, respectively), it will correctly and fully grow the bands of its $\genome$. (Note that Lemma~\ref{lem:seed-growth} only relies on Lemmas~\ref{lem:genome-data} and \ref{lem:genome-growth} but is located later in the proof as it deals solely with the structure of seed macrotiles.) By the design of the $\genome$ and its correctness by Lemma~\ref{lem:genome-data}, it is guaranteed to grow the $\externalCommunication$ seeds, $\bracket$, and $\adderArray$ during its initialization phase of growth.  This ensures that if an $\externalCommunication$ datapath grows to a query row, the $\genome$ will be able to cooperate to grow the query datapath.  Also, by the assumption of the correctness of $R^*(\beta) = \alpha$ and of the $\externalCommunication$ module, if the tile in direction $d$ is of type $t_d$, an $\externalCommunication$ datapath will grow into $L$ to the location of the query row of its $\genome$ for direction $d$ and encode tile type $t_d$.

Given those facts, we first prove the forward direction.  When $L$ receives an \emph{external communication} datapath which arrives at the location of the query row of its $\genome$ for direction $d$ and encodes tile type $t_d$, the $\externalCommunication$ datapath will cooperate with the $\genome$ to begin growing the query datapath which will count the delimiters in the glue table of the $\genome$ until it finds the entries that correspond to the tile type $t_d$ and the direction $d$. It will then activate the datapaths for each of these entries. Given correct design of the $\genome$ (by Lemma~\ref{lem:genome-data}) and the correctness of the $\externalCommunication$ (by assumption), only if tiles of type $t$ and $t_d$ can form a bond in the $d$ direction will there be a glue table entry which corresponds to tile type $t$. This entry, when activated, has a datapath that encodes the strength $s$ of this potential bond and the instructions for navigating to the $\adderUnit$ that corresponds to the tile type $t$, and thus that path will grow.

Now, we prove the opposite direction. If no bond can form in $\alpha$ between the tile of type $t_d$ in direction $d$ and a tile of type $t$ in $l$, the query datapath which is initiated by the $\texttt{external}$ $\texttt{communication}$ datapath arriving at the query location of the $\genome$ in $L$ will not find an entry in the glue table for the pair of types $t$ and $t_d$ binding in direction $d$ (again by the correctness of the $\genome$ data, Lemma~\ref{lem:genome-data}).  Therefore, no query datapath for direction $d$ will grow to the $\adderUnit$ which represents tile type $t$.

\end{proof}

\subsection{$\adderArray$}\label{sec:adderArray}

As the simulation progresses, the growth of a new macrotile which resolves into a tile of $T$ should be able to cause each neighboring macrotile to potentially form into the representation of a tile of $T$ itself. The $\adderArray$ is a module inside each macrotile that keeps track of the tiles from $T$ into which a specific macrotile location may eventually differentiate. It works by having a component for every tile type in $T$, and summing up the matching glues between that tile type and the surrounding tiles as they form in the simulation, and comparing that sum with $\tau$. The $\adderArray$ is broken down into the following series of components:

\begin{enumerate}
\item $\adderUnit$ - One $\adderUnit$ exists for each tile type in $T$. The $\adderUnit$ for $t_i \in T$ sums up the glue strengths of all the surrounding glues represented by adjacent macrotiles matching with the corresponding glues of $t_i$. 
\item $\componentAdder$ - A component of an $\adderUnit$ that sums up the glue strengths from a specific subset of the 6 sides of the tile type that the $\adderUnit$ corresponds to. There are 63 $\componentAdder$ components in an $\adderUnit$ (one for each possible subset of sides except for the empty subset).
\item $\partialAdder$ - A component of the $\componentAdder$ that performs the addition of the strength of a glue on a specific side with the strength of a glue on another side or with $\tau$. There are 7 $\partialAdder$ components in a $\componentAdder$, one for each direction and one for the value of negative $\tau$ given in two's compliment binary form. Note that in any given $\componentAdder$ only a specific subset of $\partialAdder$s are used for possible input values which may arrive later, and the others are initialized to the value $0$.
\end{enumerate}

Whenever an $\externalCommunication$ datapath carrying the encoding of tile type $t$ grows from the $d$ direction to the query section of the $\genome$ of the currently growing macrotile, it will spawn new datapaths for each tile type with a matching glue into which the currently growing macrotile may eventually differentiate. These new datapaths grow into each $\partialAdder$ that corresponds to that direction within the $\adderUnit$ that corresponds to the same tile type as the datapath. At this point, if every $\partialAdder$ in a $\componentAdder$ has received input (i.e. input has arrived from the full subset of sides which it is designed to sum), the $\componentAdder$ will add them and subtract $\tau$. If this result is a negative integer, the $\componentAdder$ ``fails'', and the $\adderArray$ waits for more datapaths. However, if the result is greater than or equal to $0$, then the $\componentAdder$ ``succeeds''.

Success of a $\componentAdder$ causes the $\adderUnit$ to send a guide rail that encodes the id of the successful tile type into the $\bracket$. This signifies that there are enough matching glues for a tile of that type to bind in $\calT$. The growth is initiated by the $\componentAdder$ which cooperates with a ``backbone'' of tiles for that $\adderUnit$ so that a row indicating that success grows toward a location which causes growth into the $\bracket$.  This ``success'' row can be grown by any $\componentAdder$ which succeeds within an $\adderUnit$, and consists of repeated attachments of a single tile type, allowing it to grow in both directions and preventing the possible success of multiple $\componentAdder$s from interfering with each other. Furthermore, this results in growth such that it is impossible to tell in which order multiple succeeding $\componentAdder$s may have completed. (An illustration of the design if this portion can be seen in Figure~\ref{fig:Component_Adder_Success_Signal}.)

The reason for having the $63$ separate $\componentAdder$ components in an $\adderUnit$ is so that every incoming datapath will cause at least one $\componentAdder$ to be fully activated rather than having to wait for additional datapaths. This means that every combination of surrounding glues will be evaluated as they appear in the simulation.

\begin{figure}[htb]
\centering
\includegraphics[width=5.5in]{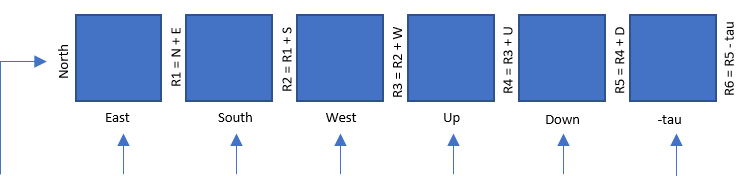}
\caption{$\partialAdder$ layout within a $\componentAdder$}
\label{fig:Component_Adder}
\end{figure}

\begin{figure}
\centering
\begin{subfigure}[b]{2.5in}
\centering
\includegraphics[width=2.5in]{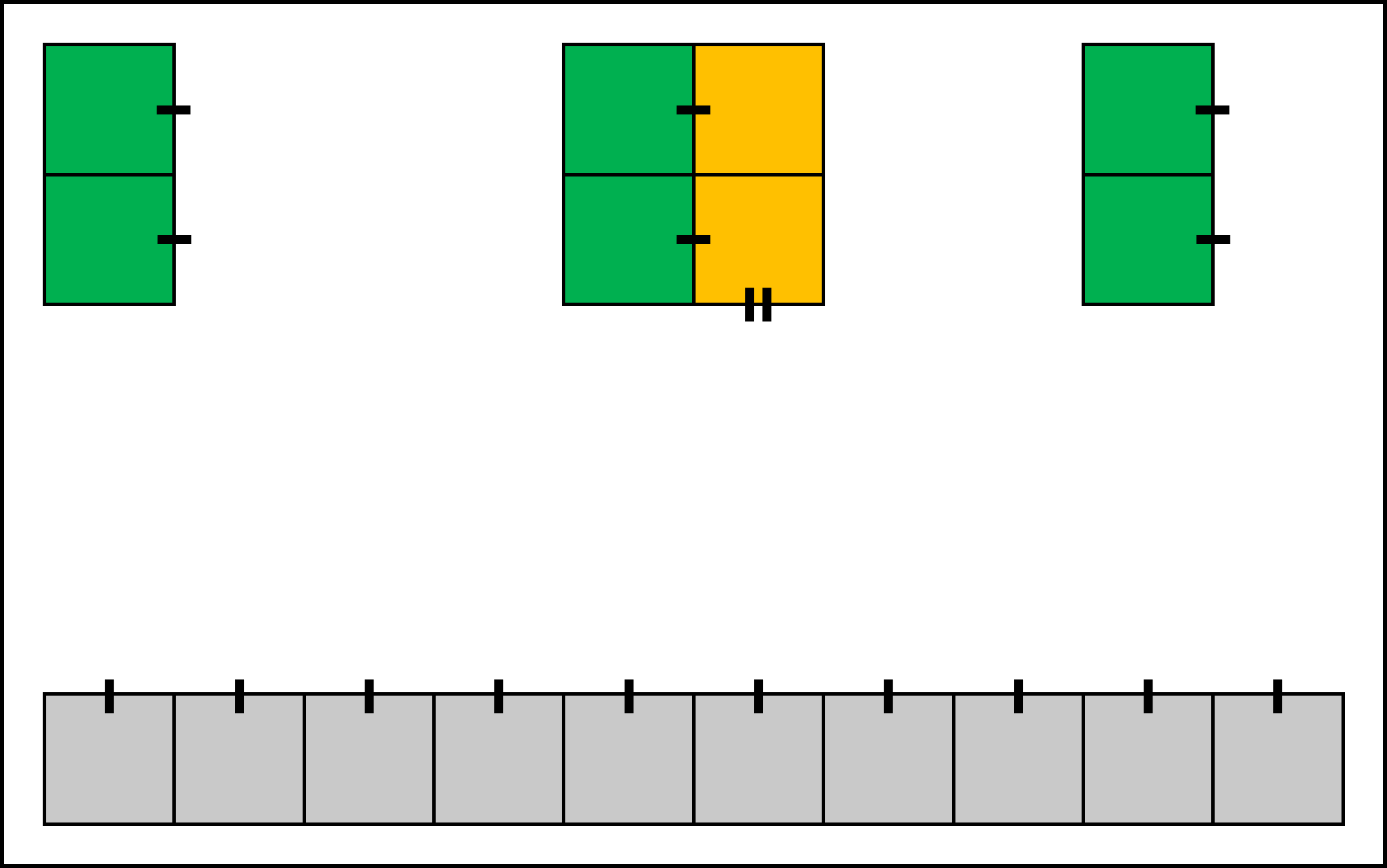}
\end{subfigure}

\begin{subfigure}[b]{2.5in}
\centering
\includegraphics[width=2.5in]{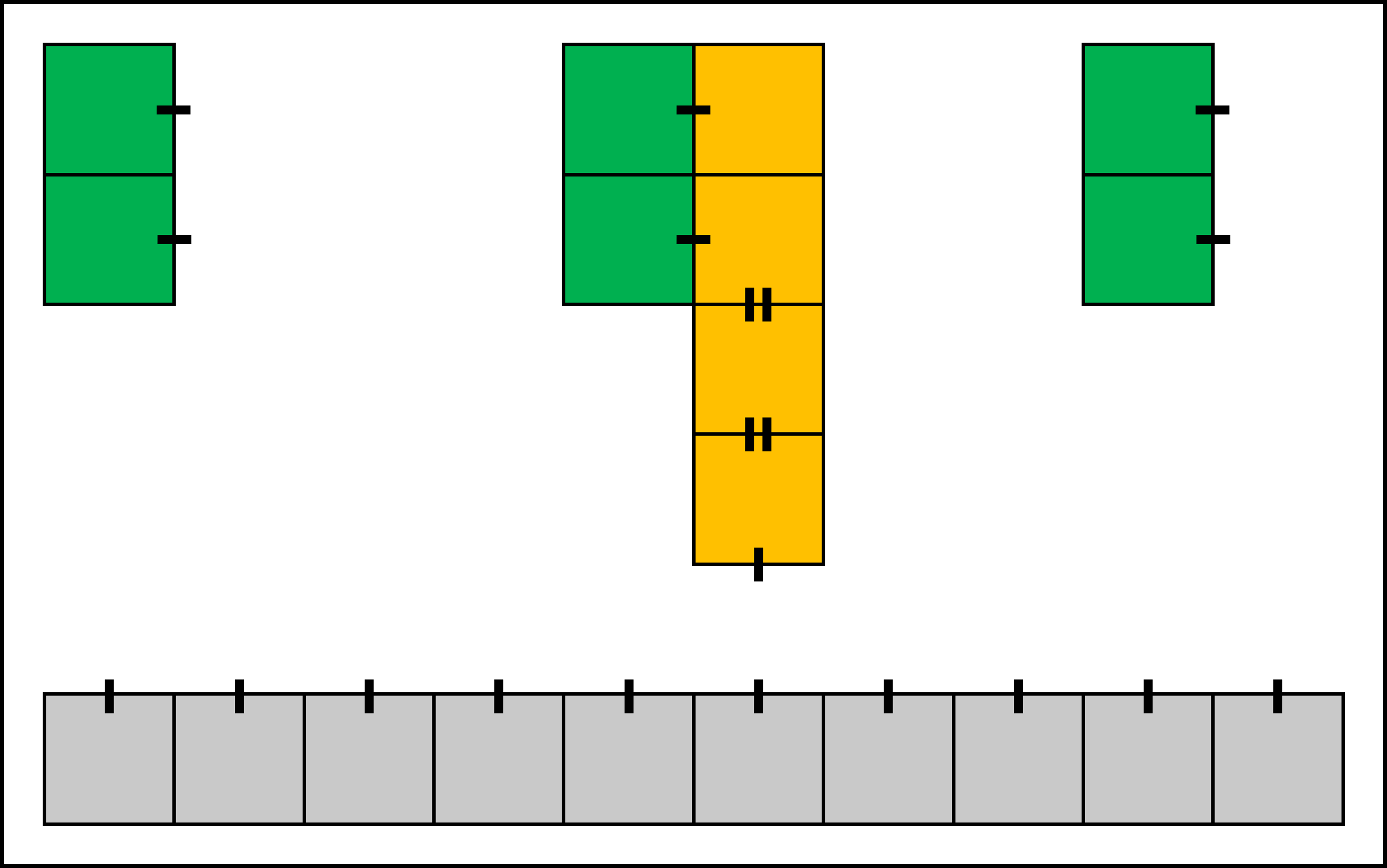}
\end{subfigure}

\begin{subfigure}[b]{2.5in}
\centering
\includegraphics[width=2.5in]{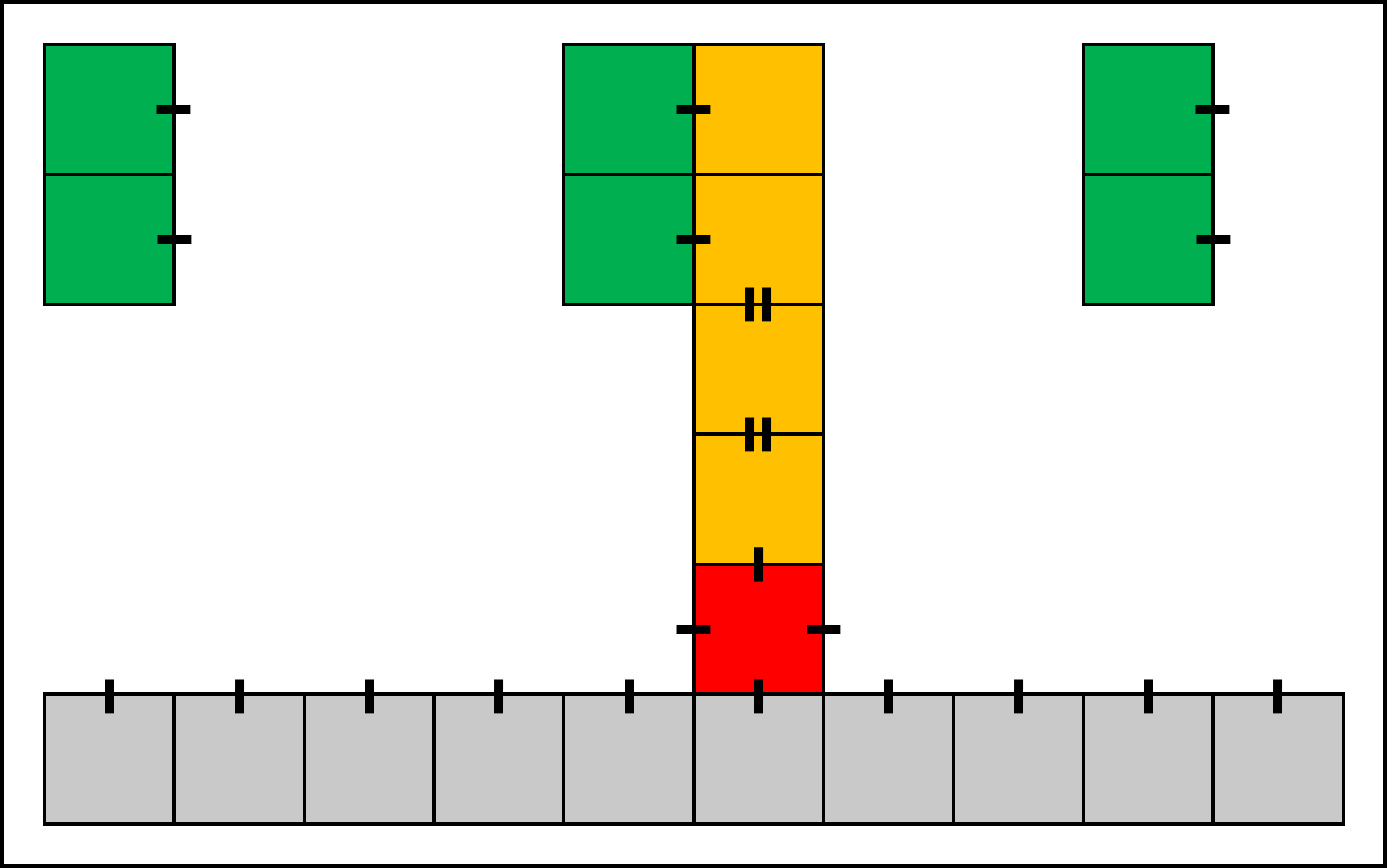}
\end{subfigure}

\begin{subfigure}[b]{2.5in}
\centering
\includegraphics[width=2.5in]{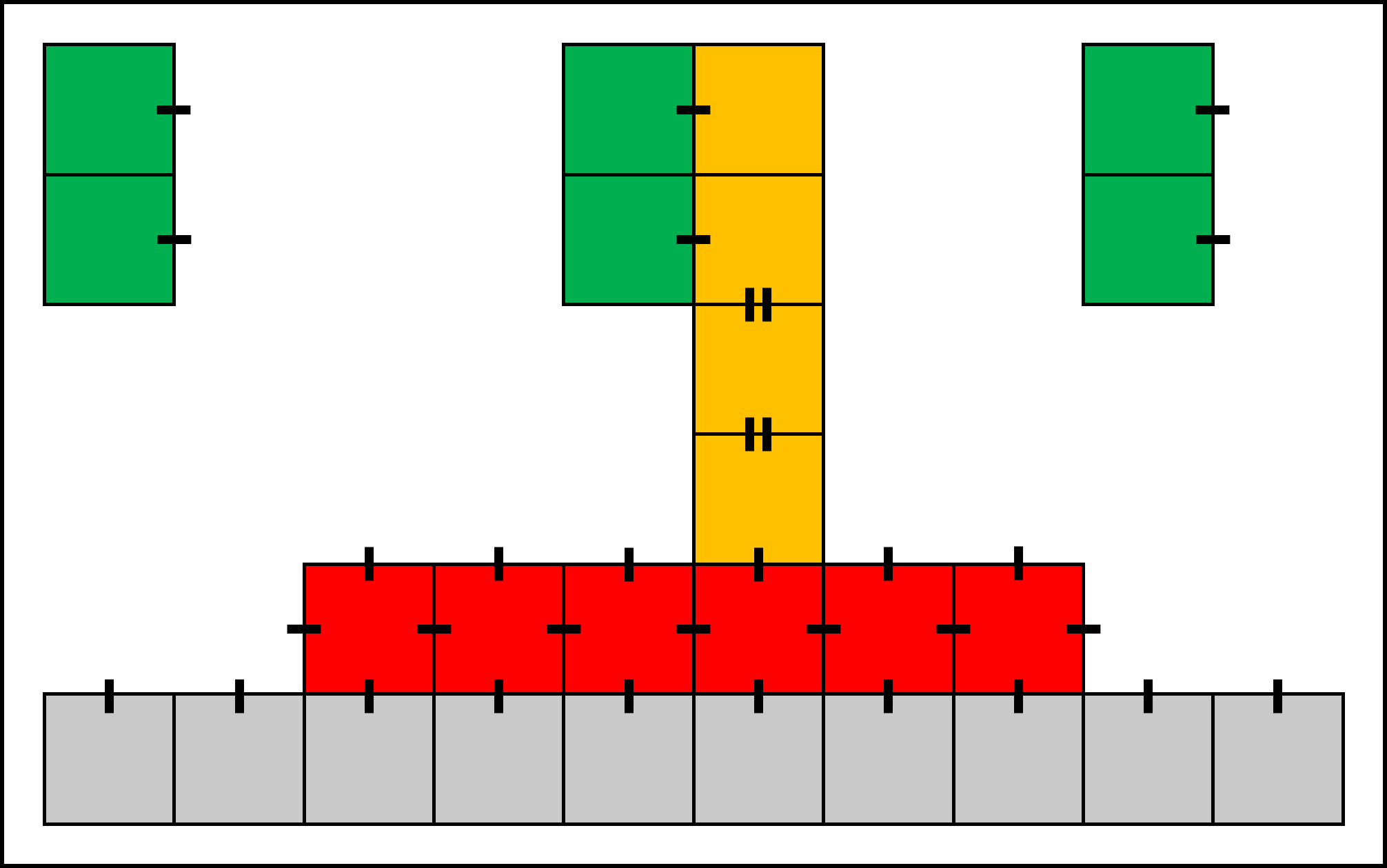}
\end{subfigure}

\caption{The initialization of the ``success signal'' within an $\adderUnit$. The green tiles are each a $\componentAdder$, while the red tiles make up the success signal. The cooperation over the gap in the third image allows for the initiation of the success signal, but the generic strength-1 glue on the top of the success signal tiles doesn't allow for back growth.}
\label{fig:Component_Adder_Success_Signal}
\end{figure}

\subsubsection{Correctness of $\adderArray$}

\begin{lemma}[Adder summation]\label{lem:component_evaluate}

Other than the $\adderArray$, assume that all the modules in the simulator $\calUT$ work correctly. Let $\alpha \in \prodasm{\mathcal{T}}$ and $\beta \in \prodasm{\calUT}$ such that $R^*(\beta) = \alpha $, $l \not \in \dom{\sigma}$ be a location outside of the seed of $\calT$ but which is adjacent to a tile in $\alpha$ of type $t_d$ in direction $d$, and $L$ be the macrotile location in $\beta$ which maps to $l$. Then, within macrotile $L$, for every tile type $t \in T$ which can form a bond with a tile of type $t_d$ in direction $d$, in the $\adderUnit$ for tile type $t$, every $\componentAdder$ which takes input from direction $d$ will receive an input for that direction equal to the strength of the possible bond.
\end{lemma}

\begin{proof}
By Lemmas~\ref{lem:genome-prop} and \ref{lem:seed-growth}, since $L$ has a neighboring macrotile which represents a tile (which is outside or inside the seed, respectively), it will correctly and fully grow the bands of its $\genome$. (Note that Lemma~\ref{lem:seed-growth} only relies on Lemmas~\ref{lem:genome-data} and \ref{lem:genome-growth} but is located later in the proof as it deals solely with the structure of seed macrotiles.) By the design of the $\genome$ and its correctness by Lemma~\ref{lem:genome-data}, it is guaranteed to grow the $\externalCommunication$ seeds, $\bracket$, and $\adderArray$ during its initialization phase of growth.
Given that all of those components of the macrotile $L$ are correctly constructed, Lemma~\ref{lem:component_evaluate} follows directly from Lemma~\ref{lem:genome-query} and the design of the $\adderArray$.
\end{proof}

\begin{lemma}[Adder output]\label{lem:correct_adder}
Other than the $\adderArray$, assume that all the modules in the simulator $\calUT$ work correctly. Let $\alpha \in \prodasm{\calT}$ and $\beta \in \prodasm{\calUT}$ such that $R^*(\beta) = \alpha$, $l \not \in \dom{\sigma}$ be a location outside of the seed of $\calT$, and $L$ be the macrotile location in $\beta$ which maps to $l$. The $\adderArray$ in macrotile $L$ will output the encoding of tile type $t \in T$ to the $\bracket$ if and only if a tile of type $t$ could attach to $\alpha$ in location $l$ (regardless of whether $l$ already contains a tile or not) with bonds which sum to $\ge \tau$.
\end{lemma}

\begin{proof}
By Lemmas~\ref{lem:genome-prop} and \ref{lem:seed-growth}, since $L$ has a neighboring macrotile which represents a tile (which is outside or inside the seed, respectively), it will correctly and fully grow the bands of its $\genome$. (Note that Lemma~\ref{lem:seed-growth} only relies on Lemmas~\ref{lem:genome-data} and \ref{lem:genome-growth} but is located later in the proof as it deals solely with the structure of seed macrotiles.) By the design of the $\genome$ and its correctness by Lemma~\ref{lem:genome-data}, it is guaranteed to grow the $\externalCommunication$ seeds, $\bracket$, and $\adderArray$ during its initialization phase of growth.

First, we show that if a tile of type $t$ is able to attach to $\alpha$ in location $l$ with bonds which sum to $\ge \tau$, then the $\adderArray$ in macrotile $l$ will propagate the encoding of a tile type $t$ to the $\bracket$.  If a tile of type $t$ is able to attach to $\alpha$ in location $l$ by forming such bonds, then there is some subset of directions, $B \subseteq \{N,S,E,W,U,D\}$, such that the tiles which neighbor location $l$ in those directions have glues adjacent to $l$ that match those of tile type $t$ and sum to $\ge \tau$. By Lemma~\ref{lem:component_evaluate}, we know that for each $d \in B$, for the $\adderUnit$ which corresponds to tile type $t$, each $\componentAdder$ which takes $d$ as input will receive as input the value of the corresponding glue. Recall that within each $\adderUnit$, there is a $\componentAdder$ for every subset of directions in $\{N,S,E,W,U,D\}$ other than the empty set, resulting in 63 instances of the $\componentAdder$ for each $t$. That means that there is one that exactly matches $D$, and that $\componentAdder$ will receive all inputs, which, since they sum to $\ge \tau$, will cause that $\componentAdder$ to evaluate to a non-negative integer, causing the $\adderUnit$ for $t$ to succeed and propagate a guide rail that encodes the tile type $t$ to the $\bracket$. Furthermore, by the design of the $\adderUnit$ and the guide rail it uses for output, even if multiple of its $\componentAdder$s ``succeed'' (meaning that multiple subsets of sides could match with total strength $\ge \tau$), with any ordering to their relative timings, they are guaranteed not to interfere with each other and the output will be unchanged and correct.

Now, we will prove the opposite direction. Assuming that a tile of type $t$ could not attach to $\alpha$ in location $l$ with bonds which sum to $\ge \tau$, we show that the $\adderArray$ in $L$ cannot propagate the encoding of tile type $t$ to the $\bracket$.  If a tile of type $t$ cannot attach in location $l$ of $\alpha$ by forming such bonds, that means that there is no subset of its glues which match those of the surrounding adjacent glues such that the sum of the strengths of those glues is $\ge \tau$.  In the case where $t$ has zero matching glues, there will be no datapaths in $L$ which grow from a query of the $\genome$ into the $\adderUnit$ corresponding to $t$, and therefore clearly that $\adderUnit$, which is the only which could propagate an encoding of tile type $t$ to the $\bracket$, will not do so.  In the case where one or more glues of $t$ match adjacent glues, for each such match a datapath resulting from a query of the $\genome$ will grow to the $\adderUnit$ corresponding to $t$.  For every subset of matching sides, there will be a $\componentAdder$ which corresponds to exactly that subset of sides and which will therefore receive all of its required inputs to proceed in its computation.  However, since no subset sums to $\ge \tau$, the $\componentAdder$ computation will result in a negative integer and it will therefore not initiate growth of a guide rail to the $\bracket$.  All other $\componentAdder$s of this $\adderUnit$ will fail to receive their full set of inputs and will thus not even perform their full computations and will not initiate growth of a guide rail to the $\bracket$.  Therefore, the $\bracket$ will never receive an encoding of tile type $t$.

\end{proof}

\subsection{$\bracket$}\label{sec:bracket}

Whenever an $\adderUnit$ $A_t$, for $t \in T$, succeeds, a guide rail that encodes the binary representation of the tile type $t$ (i.e. the number of the tile type in binary) is propagated to the next module, the $\bracket$. This propagation signifies that a tile of tile type $t$ could attach with at least strength $\tau$ in the simulated location. There are exactly $|T|$ inputs to the $\bracket$, one for each tile type in $T$, and each $\adderUnit$ of the $\adderArray$ is designed so that its output grows directly to the input location of the $\bracket$ which corresponds to the unique $t$ for which it computes. Because the simulated system $\mathcal{T}$ might be undirected and may have multiple valid tiles types that can attach in a specific location, there may be more than one successful $\adderUnit$ that propagates a tile type number to the $\bracket$. The purpose of the $\bracket$ is to move these guide rails through a competition such that only one guide rail continues to the final module and all other guide rails are blocked. It does this my making use of \emph{turn barriers} and \emph{merge barriers} (details of which can be found in Section~\ref{sec:bracket-details}). Intuitively, these pieces simply move these guide rails in pairs through \emph{points of competition} which are single-tile locations where exactly one path can place a tile, thus ``winning'' the competition.  The winning guide rail continues growth through the $\bracket$, and the losing guide rail cannot continue growth into or beyond the point of competition. The $\bracket$ has multiple levels, with guide rails competing in pairs at any given level, making for a grand total of $\ceil{\log_2|T|}$ levels which allows for selection of a single output from a maximum of $|T|$ possible inputs. Once a path has completed the $\bracket$, it initiates the growth of the datapath which grows the blocking mechanism which blocks all currently unused inputs to the $\bracket$ before finally allowing the output of the $\bracket$ to initiate the $\externalCommunication$ datapaths.  (See Section~\ref{sec:bracket-blocker} for more details.)

\subsubsection{Correctness of $\bracket$}

\begin{lemma}\label{lem:bracket}
Other than the $\bracket$, assume that all the modules in the simulator $\calUT$ work correctly. Let $\alpha \in \prodasm{\calT}$ and $\beta \in \prodasm{\calUT}$ such that $R^*(\beta) = \alpha$, $l \not \in \dom{\sigma}$ be a location outside of the seed of $\calT$ but which is adjacent to a tile in $\alpha$, $L$ be the macrotile location in $\beta$ which maps to $l$, and $l_T \subseteq T$ be the set of all tile types that could attach to $\alpha$ in location $l$ with bonds summing to $\ge \tau$. Then, the $\bracket$ of macrotile $L$ will either:

\begin{enumerate}
\item not output anything if $|l_T| = 0$, or
\item output exactly one guide rail encoding a tile type $t \in l_T$, causing $L$ to represent a tile of type $t$ under representation function $R$, otherwise.
\end{enumerate}

\end{lemma}

\begin{proof}
By Lemmas~\ref{lem:genome-prop} and \ref{lem:seed-growth}, since $L$ has a neighboring macrotile which represents a tile (which is outside or inside the seed, respectively), it will correctly and fully grow the bands of its $\genome$. (Note that Lemma~\ref{lem:seed-growth} only relies on Lemmas~\ref{lem:genome-data} and \ref{lem:genome-growth} but is located later in the proof as it deals solely with the structure of seed macrotiles.) By the design of the $\genome$ and its correctness by Lemma~\ref{lem:genome-data}, it is guaranteed to grow the $\externalCommunication$ seeds, $\bracket$, and $\adderArray$ during its initialization phase of growth.

In the specific case that $l_T$ is empty, this means that no tiles can attach to $\alpha$ in location $l$ with sufficient bonds. Therefore, by Lemma \ref{lem:correct_adder}, no $\adderUnit$ can succeed, and no guide rail can be propagated to the $\bracket$. Since no guide rails are input to the $\bracket$, by the design of the $\bracket$ no guide rail can be output from it.

Whenever $|l_T| > 0$, since each tile type in $t' \in l_T$ can attach to $\alpha$ in location $l$, then by Lemma \ref{lem:correct_adder}, there must be a $\componentAdder$ of the $\adderUnit$ corresponding to $t'$ that succeeds, causing that $\adderUnit$ to propagate a guide rail encoding that specific tile type to the $\bracket$. If $|l_T| = 1$, only one guide rail corresponding to the lone element in $|l_T|$ will be propagated, meaning it alone can be output from the $\bracket$, and by design of the $\bracket$ is guaranteed to do so. If $|l_T| > 1$, an input for each $t' \in l_T$ will be provided to the $\bracket$. In this scenario, each input falls into one of three categories: (1) it propagates fully through the $\bracket$, which exactly one input will do, (2) it enters the $\bracket$ but loses a point of competition and thus is blocked from completing the $\bracket$, or (3) the winner of the $\bracket$ completes the $\bracket$ and initiates and completes the blocking datapath (see Section~\ref{sec:bracket-blocker} for details about how the blocking is performed) before this input can enter the $\bracket$. This prevents this input from ever entering the $\bracket$.

By the design of the $\bracket$, one or multiple guide rails, depending on which of the above scenarios occurs, will move through a series of points of competitions that block all but one (with some points of competition blocking one path of two which arrive, and some possibly only ever having a single path growing to them and thus just allowing such a path to continue and not needing to block another). The single final guide rail which emerges from the point of competition of the final stage of the $\bracket$ will thus encode a single tile type $t$ from the set $l_T$ which is the ``winner'', causing the macrotile $L$ to differentiate into a tile of type $t$.
\end{proof}

\subsection{$\externalCommunication$}\label{sec:external_communication}

The final step in the differentiation process is the activation of the $\externalCommunication$ module. This module accepts the encoding of the winning tile type from the $\bracket$ as input. Once that information arrives, it is input into six datapaths, one for each direction, that grow to the neighboring macrotiles of the currently growing macrotile. Each datapath contains the binary encoding of the tile type received from the $\bracket$ and instructions for navigating to the correct row in the critical orientation of the neighbor macrotile's $\genome$ that represents the direction the datapath is coming from. Once there, it initiates a query in the glue table of that neighbor's $\genome$ by growing along the $G_2$ section, as outlined in Section \ref{sec:glue_table}. If this neighboring macrotile has not yet differentiated, this furthers that process, potentially allowing one of the components in the $\adderArray$ to succeed and for the neighboring macrotile to continue growing.

\begin{figure}[htb]
\centering
\includegraphics[width=6in]{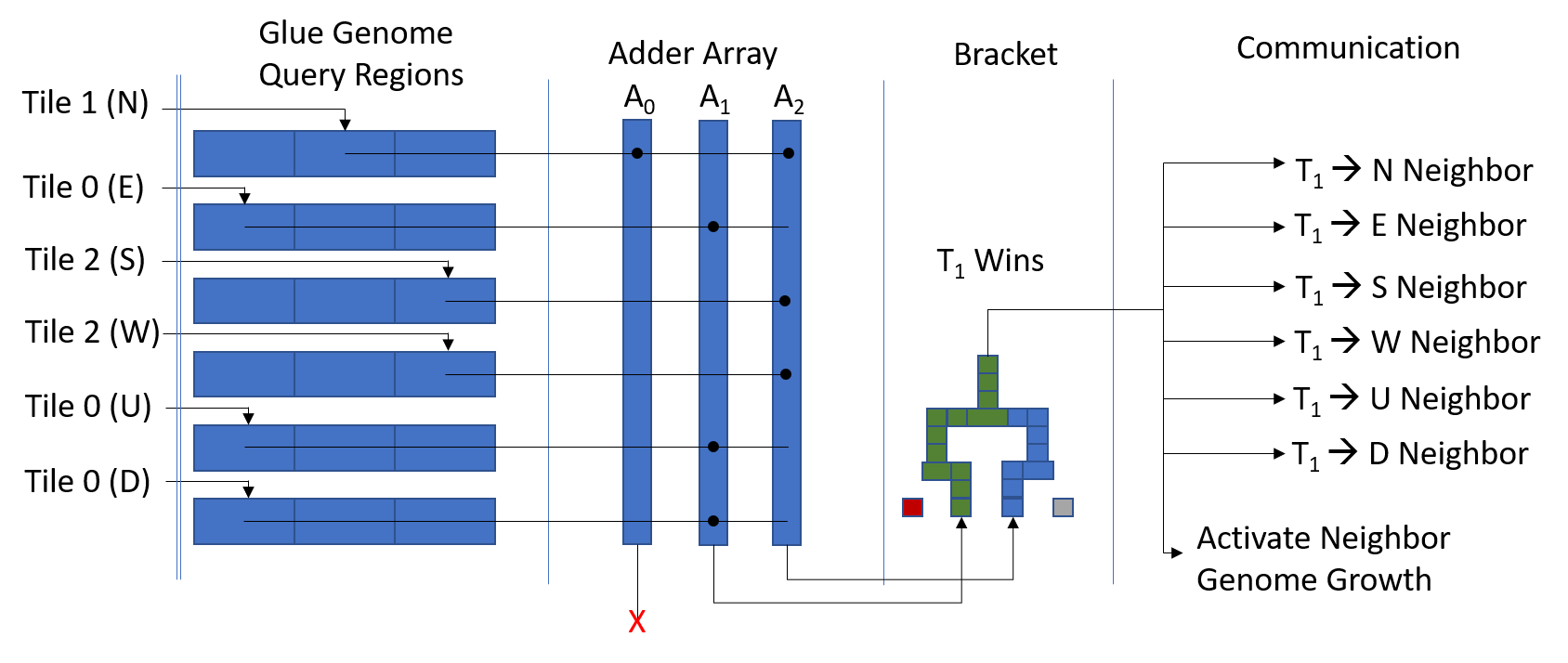}
\caption{This figure gives an example of six neighbors contributing to a central tile in a temperature 2 system, initialization is assumed to have completed for each neighbor tile.  All inputs in this figure are strength-1.  In this figure the adder representing tile 0 (A0) only receives one input and fails to output to the bracket, but A1 and A2 each receive 3 inputs and output to the bracket.  Tile 1 wins the bracket and outputs its data to each of its neighbors.  Note: This diagram is an abstract, out of context snapshot meant to show data flow through a single macrotile.  The tileset shown may or may not actually be a valid tileset. }\label{fig:DataFlow}
\end{figure}

\subsubsection{Correctness of $\externalCommunication$}

\begin{lemma} \label{lem:external-communication}
Other than the $\externalCommunication$ assume that all modules in the simulator $\calUT$ work correctly. Let $\alpha \in \prodasm{\mathcal{T}}$ and $\beta \in \prodasm{\calUT}$ such that $R^*(\beta) = \alpha$, $l \not \in \dom{\sigma}$ be a location outside of the seed of $\calT$, and $L$ be the macrotile location in $\beta$ which maps to $l$. An $\externalCommunication$ datapath encoding tile type $t$ will grow into the necessary location, for its relative direction, to perform a query in the $\genome$ of each of the $6$ neighboring macrotile locations of $L$ if and only if $L$ represents a tile of type $t$.
\end{lemma}

\begin{proof}
By Lemmas~\ref{lem:genome-prop} and \ref{lem:seed-growth}, since $L$ has a neighboring macrotile which represents a tile (which is outside or inside the seed, respectively), it will correctly and fully grow the bands of its $\genome$. (Note that Lemma~\ref{lem:seed-growth} only relies on Lemmas~\ref{lem:genome-data} and \ref{lem:genome-growth} but is located later in the proof as it deals solely with the structure of seed macrotiles.) By the design of the $\genome$ and its correctness by Lemma~\ref{lem:genome-data}, it is guaranteed to grow the $\externalCommunication$ seeds, $\bracket$, and $\adderArray$ during its initialization phase of growth.

If $L$ represents a tile of type $t$, then an encoding of $t$ exists at the output of the $\bracket$ (since that is the criteria for $L$ to represent $t$).  This output will grow to locations adjacent to the datapaths which were created during the initialization phase of macrotile formation and which contain the instructions needed to grow each of the $6$ datapaths into adjacent macrotiles.  The encoding of $t$ is incorporated as the payload of these datapaths, causing each to grow the $\externalCommunication$ datapath which will terminate in the direction-specific query location of the critical orientation of the $\genome$ of each macrotile neighboring $L$.  If $L$ does not represent any tile type, then no cooperative growth will be possible which would allow the $\externalCommunication$ datapaths to grow, and if a tile type $t' \ne t$ is represented by $L$, then it will be the encoding of $t'$ which is carried by the $\externalCommunication$ datapaths.
\end{proof}

\subsection{Correctness of seed structure}\label{sec:seed-correctness}

\begin{lemma}[Seed correctness]\label{lem:seed-correct}
Let $\sigma$ be the seed assembly of $\calT$. If $R$ is the representation function defined in Section~\ref{sec:rep-function} and $\sigma_{\calT}$ is the seed assembly for $\calUT$ created using the techniques of Section~\ref{sec:seed}, then $R^*(\sigma_{\calT}) = \sigma$.
\end{lemma}

\begin{proof}
Lemma~\ref{lem:seed-correct} follows directly from the definition of the macrotiles in $\sigma_{\calT}$ and the definition of $R$, since each macrotile has the encoding of the necessary tile type at the end of the $\bracket$'s output path location, and that is where $R$ checks to resolve each macrotile, each macrotile in $\sigma_{\calT}$ will resolve to the correct tile in $\sigma$.  Since no tiles are placed outside of the macrotiles which map to tiled location of $\sigma$, $R^*(\sigma_{\calT}) = \sigma$.  Furthermore, by design of the macrotiles of the seed, the points which could initiate query datapaths along the $\genome$ in response to any $\externalCommunication$ datapaths received from neighboring macrotiles have blocking tiles in place, which prevent queries from growing.  Also, by the design of the $\genome$, any $\genome$ bands growing into a seed macrotile location from a neighboring macrotile will simply merge with the existing $\genome$ and not cause any new growth. Since those are the only pathways for growth into a macrotile, the seed macrotiles cannot have their behavior changed by any neighbors, and thus are guaranteed to always correctly represent the seed tiles.
\end{proof}

\begin{lemma}[Seed expansion]\label{lem:seed-growth}
Let $\beta \in \prodasm{\calUT}$ such that $R^*(\beta) = \sigma$, $l \not \in \dom{\sigma}$ be an empty location adjacent to a tile in $\sigma$, and $L$ be the macrotile location in $\beta$ which maps to $l$. The following will grow into macrotile $L$: (1) the complete bands of the $\genome$, and (2) a valid $\externalCommunication$ datapath encoding tile type $t$ which is in the adjacent location of $\sigma$.
\end{lemma}

\begin{proof}
Let $l'$ be the location of a tile in $\sigma$ which is adjacent to $l$, and $L'$ be the macrotile which maps to $l'$.  By the definition of a seed tile macrotile, $L'$ will contain a row of the $\genome$ at an intersection, and by Lemma~\ref{lem:genome-growth}, that is capable of growing the complete set of bands of the $\genome$ in $L'$.  Further, by Lemma~\ref{lem:genome-data} we know that the $\genome$ contains the necessary instructions to seed the $\externalCommunication$ datapaths.  As the critical orientation of the $\genome$ grows and enters the initialization stage, the $\externalCommunication$ seed is the first to grow and will correctly grow the datapaths necessary to accept output of the $\bracket$.  One of these datapaths also contains a pathway for a callback instruction, also initiated by that output, which causes growth back along that datapath to the $\genome$ and initiates the $\genome$ output to neighbors, resulting in an intersection row of the $\genome$ to grow in each, and thus the full bands in each by Lemma~\ref{lem:genome-growth}.  (The subsequent initialization of the $\bracket$ and $\adderArray$ will be prevented by the blocker tiles unique to seed macrotiles.)  Because the $\externalCommunication$ datapaths are correctly seeded and will grow to the point of accepting output from the $\bracket$, the preexising encoding of the tile type represented by $L'$, at the precise location of the end of the $\bracket$'s output, will initiate the cooperative growth that both begins the output of the $\genome$ and the growth of the $\externalCommunication$ datapaths encoding tile type $t$ to all $6$ macrotiles neighboring $L'$, including $L$. Also, exactly as discussed in the proof of Lemma~\ref{lem:seed-correct}, by design of the seed macrotiles, no other macrotiles can cause incorrect growth within them, and therefore $L'$ will always output the correct growth into $L$.
\end{proof}

Now we also state an prove a Lemma about the bounded nature of the ``fuzz'' that grows in $\calUT$, which will be useful later for our proof of correctness.

\begin{lemma}[Bounded fuzz]\label{lem:bounded-fuzz}
Let $\alpha \in \prodasm{\mathcal{T}}$ and $\beta \in \prodasm{\calUT}$ such that $R^*(\beta) = \alpha$, $l \not \in \dom{\sigma}$ be a location in space which is not part of the seed and not adjacent to any tile in $\alpha$, and $L$ be the macrotile location in $\beta$ which maps to $l$. No individual tiles of the system $\calUT$ can be placed within macrotile $L$.
\end{lemma}

\begin{proof}
Lemma~\ref{lem:bounded-fuzz} follows immediately from the facts that (1) since $R^*(\beta) = \alpha$, $L$ is not adjacent to any macrotile which maps to a tile in $\alpha$, (2) since $l$ is not in $\sigma$, $L$ does not initially have any tiles within it in the seed of $\calUT$, and (3) the first tiles that can grow into an initially empty macrotile are those from the $\genome$ or $\externalCommunication$ datapaths of a neighboring macrotile, and (4) only a macrotile which represents a tile of $\alpha$ can output either of those to neighboring macrotiles. Therefore, since no macrotiles neighboring $L$ represent tiles of $\alpha$, none of them can output into $L$ and therefore $L$ can have no tiles within it.
\end{proof}

\section{Proof of Correctness of General 3D aTAM Construction} \label{sec:thm1-proof}

In this section, we piece together the proofs of correctness of individual modules from Section~\ref{sec:construction} to prove the correctness of the entire construction and ultimately Theorem~\ref{thm:3DaTAMIU}.

\subsection{Correctness of construction}

In this section, we prove Theorem~\ref{thm:3DaTAMIU} by proving the correctness of our construction.  To do this, we will show that $\calUT$ correctly simulates $\calT$ following Definition~\ref{def:s-simulates-t}, specifically showing how $\calT$ follows $\calUT$ (Definition~\ref{def-t-follows-s}) and $\calUT$ models $\calT$ (Definition~\ref{def-s-models-t}), which will also show equivalent productions (Definition~\ref{def-equiv-prod}). To do this, we will prove the correctness of the growth of $\calUT$ as it simulates the growth $\calT$ through the full set of possible tile addition scenarios.

\begin{lemma}\label{lem:models}
$\calUT$ models $\calT$.
\end{lemma}

Given an arbitrary 3D aTAM system $\calT$ and producible assembly $\alpha \in \prodasm{\calT}$, there may be an arbitrary number of new assemblies that $\alpha$ can grow into via a single tile addition, i.e. the number of frontier locations where new tiles could validly attach may be arbitrarily large, and each frontier location may potentially allow for multiple tile types to attach. (See Figure~\ref{fig:models-example} for a small example.) By Definition~\ref{def-s-models-t}, for $\calUT$ to model $\calT$ there must exist assembly sequences which allow for the assemblies of $\calUT$ to grow into representations of \emph{any} of the resulting assemblies, which essentially means that $\calUT$ cannot overaly restrict growth options, preventing representations of some of $\calT$'s producible assemblies.  However, an importance nuance of this type of universal simulation arises because, given any particular macrotile representing a frontier location in $\alpha$, there may be an arbitrary number of tile types which can attach in that location in $\alpha$.  Since the simulating tile set $U$ must consist of a number of tile types which is constant regardless of what system is being simulated, the number of options to choose from could be much larger than the number of tile types in $U$. If the number is large enough, standard information theory therefore dictates that more than one tile placement in $\calUT$ must be used to determine the selection, and that after one or more of those tile additions, the range of options must become limited (i.e. the set of options is reduced) before a tile placement which makes the final selection. What this means is that a universal simulator is therefore forced to grow assemblies in $\calUT$ which represent assemblies of $\calT$ but which, at some point, have only partially completed the selection of which tile to represent in a given location, meaning that the assembly in $\calUT$ will still treat that location as representing empty space, but will no longer be able to grow into representations of all of the full set of tile types which could appear there in $\alpha$. Due to this fundamental constraint, in Definition~\ref{def-s-models-t} the set $\Pi$ is defined to capture the requirement of there being some set of assemblies through which $\calUT$ can grow so that, even if options become restricted at some point, there at least was an assembly sequence through which every possible assembly of $\alpha$ could have been represented.

\begin{figure}[htb]
\centering
\includegraphics[width=2.0in]{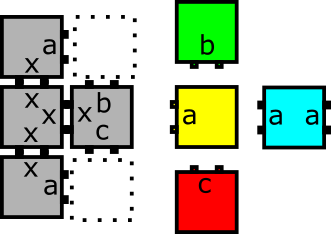}
\caption{(Left) An example assembly $\alpha \in \prodasm{\calT}$. Frontier locations are shown by dashed lines, (Right) Tiles that can attach in the frontier locations to the left.  The green, yellow, or blue could attach in the top location, and the red, yellow, or blue could attach in the bottom location.} \label{fig:models-example}
\end{figure}

A specific example of this as it relates to our construction is the following. Let $\beta \in \prodasm{\calUT}$ where $R^*(\beta) = \alpha$ is the assembly shown in Figure~\ref{fig:models-example}, and assume that the macrotile has begun growth in the location representing the top frontier location. If that macrotile has grown the bands of the $\genome$ and initialized all components, and then has received the $\externalCommunication$ datapaths from both the west and south sides and those have resulted in growth through the $\adderUnit$s for each of the green, yellow, and blue tile types, those three $\adderUnit$s will provide input to the $\bracket$.  Before growth of those inputs begins into the $\bracket$, the assembly still has the potential to grow into assemblies which would represent each of the three.  However, as soon as a point of competition is won by one of the $\bracket$ paths over another, but before a path has completed the entire $\bracket$, that macrotile will still represent empty space, but it will no longer have the potential to grow into all of the options that $\alpha$ can.  Furthermore, now assume the assembly sequence continues in such a way that no further tile attaches into the macrotile representing the top frontier location until the macrotile of the bottom location grows until it represents one of the possible tiles that could attach there in $\alpha$.  Now this new assembly in $\calUT$ maps to one which represents $5$ tiles in $\calT$ but at no time while it represents that assembly can it still grow into all of the options that match those still available to the assembly in $\calT$, since that can still receive any of the three tiles in the top frontier location, but one of those is no longer possible in the assembly of $\calUT$.  However, due to our construction, $\calUT$ still correctly models $\calT$, and this can be seen by noting that the growth of all macrotiles which have not yet differentiated (from representing empty space to representing a tile) is completely independent. This means that one macrotile's growth cannot affect another, and it also means that for any such macrotile which differentiates, there exists a valid assembly sequence in which all growth in all of the others which haven't yet differentiated is stalled before any inputs enter the $\bracket$. The assemblies of this form make up the set $\Pi$ of Definition~\ref{def-s-models-t}, and are the witnesses that it's possible through some assembly sequence(s) for assemblies in $\calUT$ to grow into representations of all valid assemblies in $\calT$.

Due to the independent growth of macrotiles, we can now finish our proof of Lemma~\ref{lem:models} by simply using induction which shows that $\calUT$ properly models $\calT$ throughout a (possibly infinite) assembly sequence by explaining how each of the different scenarios through which a single tile attachment can happen in $\calT$ is correctly modeled in $\calUT$  For our base case, we consider $\sigma$, the seed of $\calT$.  By Lemma~\ref{lem:seed-correct} we know that the seed $\sigma_{\calT}$ of $\calUT$ correctly represents $\sigma$. 
For our induction hypothesis, we assume that the assembly $\alpha \in \prodasm{\calT}$ is correctly represented in $\calUT$ and then show how $\calUT$ correctly models any possible tile addition to $\alpha$. The following sections will iterate through all of the possible scenarios in which new tiles in $\calT$ can attach to a $\alpha$ at a location $l$, and how those scenarios are modeled under Definition~\ref{def-s-models-t} in $\calUT$ with an assembly $\beta \in \prodasm{\calUT}$ and macrotile $L$ such that $R^*(\beta) = \alpha$ and macrotile $L$ in $\beta$ maps to location $l$ in $\alpha$.

\subsubsection{Single-sided binding}

First, we look at the variety of scenarios under which a tile of type $t$ can attach to $\alpha \in \prodasm{\calT}$ in location $l$ using a single $\tau$-strength bond to a neighboring tile $l_d$ (of tile type $t_d$) which is represented in $\beta \in \prodasm{\calUT}$ by macrotile $L_d$.

In each case of these scenarios, we can assume macrotile $L$ has a completed $\genome$ by Lemmas~\ref{lem:timing} and \ref{lem:genome-prop}. From there, we can assume that an $\externalCommunication$ datapath encoding tile type $t_d$ will grow in from macrotile $L_d$ to the query section of the $\genome$ and that it will activate the correct query datapaths corresponding to all the possible tile types in $T$ that could bond with a tile of type $t_d$ in the $d$ direction, using Lemmas~\ref{lem:genome-data} and \ref{lem:genome-query}. We also know that one of these query datapaths specifically corresponds to tile type $t$, since we know that a tile of type $t$ can bond to assembly $\alpha$ in location $l$.

\paragraph{No other neighbors} The trivial case of attaching a tile is when no other neighboring tiles are present and $t$ is the only tile type that can attach to assembly $\alpha$ at location $l$. In this case, since $t$ is the only tile type that can attach at location $l$, Lemma~\ref{lem:correct_adder} tells us that only a single guide rail that encodes $t$ will leave the $\adderArray$ and enter the $\bracket$. By Lemma~\ref{lem:bracket}, it is the only guide rail that can leave the $\bracket$, signifying that macrotile $L$ has differentiated to represent tile type $t$. Finally, this information is propagated to the neighbors, along with the $\genome$.

\paragraph{Multiple neighbors, single tile type} Still assuming that $t$ is the only tile type that can bind in location $l$, we now look at the scenario in which other neighbors are present but do not represent adjacent glues which are able to attach a tile type other than $t$. By Lemmas~\ref{lem:genome-data} and \ref{lem:genome-query}, these neighbors will initiate query datapaths from the $\genome$ to the $\adderArray$ that represent bonds they can form with other potential tile types. However, by Lemma~\ref{lem:correct_adder}, we know that none of these query datapaths can cause an $\adderUnit$ to succeed (other than the $\adderUnit$ corresponding to $t$), since no tile type other than $t$ can attach to assembly $\alpha$ at location $l$. Therefore, only a guide rail encoding $t$ will leave the $\adderArray$, go through the $\bracket$, and cause the macrotile $L$ to differentiate.

\paragraph{Multiple neighbors, single tile type, different attachments} Now, we look at the same case as before, but now we assume that tile type $t$ can attach at location $l$ through neighbors other than $l_d$. Similar to the last case, Lemma~\ref{lem:correct_adder} tells us that a guide rail encoding $t$ is the only possible output of the $\adderArray$. The difference from the last case is that query datapaths initiated by the neighboring macrotiles other than $L_d$ will cause other $\componentAdder$ pieces within the $\adderUnit$ that corresponds to tile type $t$ to also succeed. In fact, these query datapaths from the neighboring macrotiles may beat the query datapaths from macrotile $L_d$. However, the $\adderUnit$ is set up so that, regardless of which $\componentAdder$ succeeds first, the result in the same, making timing irrelevant (also by Lemma~\ref{lem:correct_adder}).

\paragraph{Multiple neighbors, multiple tile types} Now, we look at the case when there are multiple neighbors of tile location $l$, causing multiple tile types to be able to attach to $\alpha$, but tile type $t$ is chosen. Again, by Lemmas~\ref{lem:genome-data} and \ref{lem:genome-query}, each potential bond between a neighbor and location $l$ in $\alpha$ will have a corresponding query datapath in macrotile $L$ from the $\genome$ to the $\adderArray$. By Lemma~\ref{lem:correct_adder}, we know each tile type that can attach at location $l$ to assembly $\alpha$ will have a corresponding guide rail that leaves the $\adderArray$. Now, there are two scenarios we can consider, depending on whether a guide rail enters the $\bracket$ before or after the ``winning'' guide rail encoding tile type $t$ has caused the macrotile to differentiate. If the other guide rails enter before differentiation, they will in the traditional competition, eventually losing to the guide rail encoding tile type $t$ (assuming we are still modeling tile type $t$ attaching before another tile type in location $l$). However, if the other guide rails enter after differentiation, they may collide with the $\bracketBlocker$, whose construction was signaled by the ``winning'' guide rail, and stop growing (see Section~\ref{sec:bracket-blocker}). Either way, tile type $t$ is correctly represented and that information is propagated to the neighbors of $L$.

\subsubsection{Multi-sided binding}

Next, we look at the scenarios under which a tile of type $t$ can attach to $\alpha \in \prodasm{\calT}$ in location $l$ using cooperation between multiple bonds from neighbors $l_0$ through $l_i$ (of tile types $t_0$ and $t_i$) which are represented in $\beta \in \prodasm{\calUT}$ by macrotiles $L_0$ through $L_i$ such that each $L_n$ maps to $l_n$ for $0 \leq n \leq i$.

In each case of these scenarios, we can assume macrotile $L$ has a completed $\genome$ by Lemmas~\ref{lem:timing} and \ref{lem:genome-prop}. From there, we can assume that $\externalCommunication$ datapaths encoding tile types $t_0$ through $t_i$ will grow in from macrotiles $L_0$ through $L_i$ respectively to the query section of the $\genome$ and that they will initiate the correct query datapaths corresponding to all the possible tile types in $T$ that could bind through cooperation with tile types $t_0$ through $t_i$, using Lemmas~\ref{lem:genome-data} and \ref{lem:genome-query}. We also know that one of these query datapaths specifically corresponds to tile type $t$, since we know that a tile of type $t$ can bind to assembly $\alpha$ in location $l$.

\paragraph{Single tile type, single attachment} The trivial case for multi-sided binding is when just the neighbors $l_0$ through $l_i$ are present, allowing for only tile type $t$ to attach to assembly $\alpha$ in location $l$. In this case, Lemma~\ref{lem:correct_adder} tells us a guide rail encoding tile type $t$ is the only possible output of the $\adderArray$. Therefore, it must win the $\bracket$, causing macrotile $L$ to differentiate to a representation of $t$.

\paragraph{More neighbors, single tile type} Assuming $t$ is still the only tile type that can attach to assembly $\alpha$ at location $l$, we now consider when extra neighbors and exposed glues are present but unable to allow for the attachment of a tile of any other type. Again, Lemma~\ref{lem:correct_adder} says only a guide rail encoding tile type $t$ can leave the $\adderArray$. The only difference is that additional query datapaths will leave the $\genome$ to the $\adderArray$ representing the bonds that can form between the extra neighbors and specific tile types in $T$ (if any). However, these additional query datapaths will be unable to cause any $\componentAdder$ to succeed.

\paragraph{Single tile type, multiple attachments} Still assuming $t$ is the only tile type that can attach to assembly $\alpha$ at location $l$, we now assume that there are multiple different sets of neighboring glues which are sufficient to allow it to attach. Lemma~\ref{lem:correct_adder} says only the guide rail encoding tile type $t$ will leave the $\adderArray$. The difference is that additional query datapaths that correspond to additional bonds that can form will grow to the $\adderArray$ and cause additional $\componentAdder$ pieces within the $\adderUnit$ that corresponds to tile type $t$ to also succeed.

\paragraph{Multiple tile types} Now, we look at when additional tile types could also attach to assembly $\alpha$ at location $l$. Lemma~\ref{lem:correct_adder} says that a guide rail encoding each tile type (including $t$) will leave the $\adderArray$. Similar to the scenario from single-sided binding, there are two subcases we must look at for each additional guide rail. If a specific guide rail reaches the $\bracket$ before the macrotile has differentiated, it will go into the $\bracket$ and be blocked by another guide rail. However, if it reaches the $\bracket$ after the macrotile has differentiated, it may be blocked by the $\bracketBlocker$, whose growth was initiated by the guide rail encoding tile type $t$ winning the bracket. Regardless, $t$ will be correctly represented.

\begin{proof}[Proof of Lemma~\ref{lem:models}]
The cases discussed show that, starting from the seed assembly $\sigma$ and then for any assembly $\alpha \in \prodasm{\calT}$, $\calUT$ correctly models all possible tile attachments through the valid assembly sequences of $\calT$.  Additionally, by Lemma~\ref{lem:bounded-fuzz} we know that growth will not occur in $\calUT$ outside of the region which represents $\alpha$, and therefore $\calUT$ models $\calT$.

\end{proof}

\begin{lemma}\label{lem:follows}
$\calT$ follows $\calUT$.
\end{lemma}

\begin{proof}
Rather than show $\calT$ can make any tile attachment that can be represented in $\calUT$, we opt to show that, if a tile of type $t$ cannot attach in $\calT$, a macrotile in the representative location in $\calUT$ cannot differentiate to represent $t$. This proof follows directly from Lemma~\ref{lem:correct_adder} that says, if a tile of type $t$ cannot attach in $\calT$, then a guide rail that encodes $t$ cannot leave the $\adderArray$ in the representative macrotile. If a guide rail encoding tile type $t$ cannot leave the $\adderArray$, then it cannot win the $\bracket$, and the representative macrotile, therefore, cannot differentiate to represent $t$.
\end{proof}

\begin{lemma}\label{lem:equiv-prod}
$\calT$ and $\calUT$ have equivalent productions.
\end{lemma}

\begin{proof}
Because the seed representation is correct by Lemma~\ref{lem:seed-correct}, and $\calT$ and $\calUT$ have equivalent dynamics, $\calT$ and $\calUT$ therefore have equivalent productions.
\end{proof}

\begin{lemma}\label{lem:simulates}
$\calUT$ simulates $\calT$.
\end{lemma}

\begin{proof}
Since $\calT$ and $\calUT$ have equivalent productions by Lemma~\ref{lem:equiv-prod} and equivalent dynamics by Lemmas~\ref{lem:models} and \ref{lem:follows}, then by the definition of simulation, $\calUT$ simulates $\calT$.
\end{proof}

\begin{proof}[Proof of Theorem~\ref{thm:3DaTAMIU}]
The simulation $\calUT = (U, \sigma_\calT,\tau')$ uses $\tau' = 2$.

We use the computable algorithms outlined in Section~\ref{sec:scale} and provided in detail in Section~\ref{sec:scale-details} to calculate the scale factor $m \in \mathbb{N}$ for the simulation, the composition of the $\genome$, and the structure of the seed $\sigma_\calT$.

We use the algorithm (see Section~\ref{sec:rep-function}) which generates the representation function $R$ for $\calUT$ by simply assigning a binary number to each tile type of $T$ and then checking the output location of the $\bracket$ of a macrotile location to determine if a binary number is fully represented to map the macrotile to a tile type in $T$, else it is mapped to empty space.

Since $\calT$ can be an arbitrary 3D aTAM system, and $\calUT$ simulates $\calT$ at scale factor $m$ using representation function $R$, the tile set $U$ is intrinsically universal for the class of all 3D aTAM systems at temperature $2$.  Since $\tau'$ is always $2$, $U$ is intrinsically universal for the class of all 3D aTAM systems, and thus the class of all 3D aTAM systems is intrinsically universal.
\end{proof}

\section{Technical Details for General 3D aTAM Construction} \label{sec:low-level}

In this section, we provide technical details to support the design and implementation of our construction in Section~\ref{sec:construction-short}, as well as the claims we make in Section~\ref{sec:construction} and Section~\ref{sec:thm1-proof}.

\subsection{Growth patterns and paradigms}\label{sec:growth-patterns}

First, we discuss the basic patterns of growth throughout our construction.

\paragraph{Diagonal Advance}
The diagonal advance growth pattern uses a single strength-2 glue per row to advance into the next row and cooperation to fill in the rest of the row.  The first tile into the row allows the tile immediately to its right or left to be the next tile to advance into a new row (this is done via a glue label which is unique for the tiles along the diagonal), thus the next row's advancing tile can always be found one tile forward and one tile right (or left) of the previous advancing tile, forming a diagonal line of tiles.  This growth pattern is bi-directional but it is not collision tolerant.

\paragraph{Limited Strength-2}
In the limited strength-2 growth pattern, each tile in a row uses a strength-2 glue on its forward face and a different strength-2 glue on its backward face (i.e. each row uses glues hard-coded to be specific for that row).  Using different glues means that the limited strength-2 growth strategy can be used to advance a constant distance.  The limited strength-2 growth pattern can be collision tolerant.  If growth occurs from row 0 to row 1 and row 2 to row 1 at the same time, row 1 will always consist of the same tiles. This pattern is desirable when a constant distance must be covered, but is not effective if the required distance may vary.  This growth pattern is bi-directional and collision tolerant and is mainly used for growth where growth may occur from the start and end at the same time.

\paragraph{Unlimited Strength-2}
The unlimited strength-2 growth pattern, also referred to as \emph{pumping}, uses the same strength-2 glue on the forward and backward face. Once started this growth pattern will move in a straight line unless it collides with a previously placed ``stopper'' tile.  This is desirable when there is a variable distance between two ends of the collision zone. However, it introduces a dependency that the stopper tile(s) be placed before the unlimited strength-2 growth starts, otherwise it is impossible to guarantee the growth won't result in infinite growth.  Given correct usage of stopper tiles, this growth pattern can be reversible and collision tolerant, although in our construction it is generally only used when growth can occur from a single direction.

\begin{figure}[htb]
\centering
\includegraphics[width=5.0in]{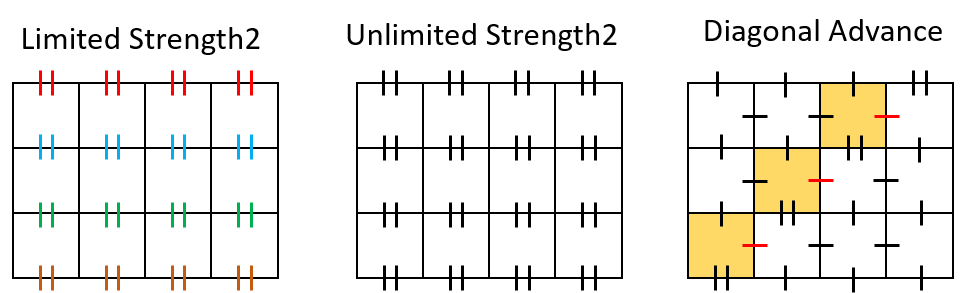}
\caption{Three general growth patterns}\label{fig:Movement Patterns}
\end{figure}

\paragraph{Delayed Activation}
The \emph{delayed activation} growth pattern, which we also refer to as \emph{priming}, is a technique that prepares glues to send new datapaths but waits for a signal from another component in the construction to do so.  An example can be seen in Figure~\ref{fig:DelayedActivation}.  There are two stages to the priming growth pattern.  The first stage, called the priming stage, happens when a row that is growing in a given plane presents strength-1 glues out of that plane.  The second stage is called activation.  Activation occurs when some event (like a query to the $\genome$ or a tile winning the $\bracket$) triggers attachment of a tile that cooperates with the primed strength-1 glues.  The activated row has strength-2 glues that initiate new growth (e.g. start a datapath or trigger the growth of the $\genome$ into a neighboring cube). The chain of cooperative growth in the activation stage will terminate if it reaches a ``dead zone'' in which the primed component purposely has a null glue included in the sequence of otherwise strength-1 glues.  Dead zones are important for callbacks in the initialization phase and for activating only the desired sections of the glue $\genome$.

\begin{figure}[htb]
\centering
\includegraphics[width=5.5in]{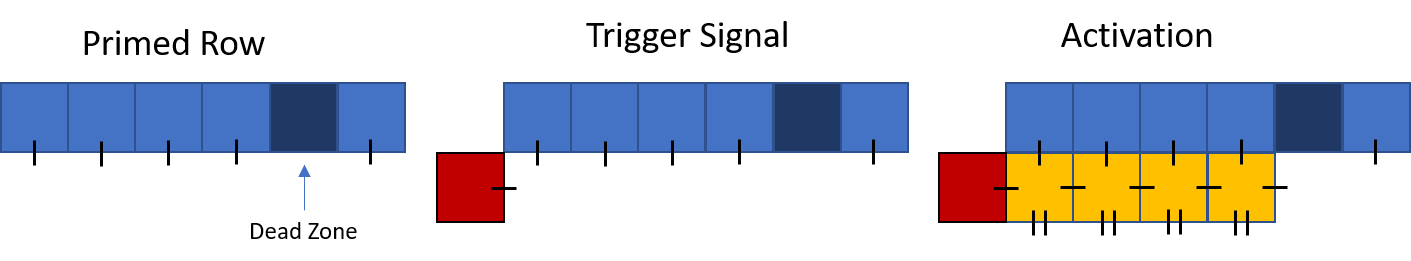}
\caption{The Delayed Activation movement pattern}\label{fig:DelayedActivation}
\end{figure}

\paragraph{Guide Rail}
The \emph{guide rail} growth pattern consists of two portions: a single-tile-wide ``guide rail'' and a ``payload'' which can carry an arbitrary amount of data (since it may have arbitrary width).  The guide rail section may use either a limited strength-2 or unlimited strength-2 growth pattern to move linearly in a single axis.  The payload section presents a strength-1 glue representing the data being carried in the forward and backward directions.  As the guide rail moves, the guide rail extends a strength-1 glue orthogonal to its direction of growth which allows the payload to use cooperation to fill in each new row that the guide rail has grown into. The guide rail is used when the relative distance between two points is fixed or only varies in one axis, which mainly occurs in the interiors of components in our construction, such as the $\adderArray$ and the $\bracket$.

\paragraph{Datapath}
The \emph{datapath} growth pattern is very similar to the diagonal advance pattern in that there is only ever a single strength-2 glue advancing the path at any given time which progresses over the width of the datapath as it grows forward. The key difference between the datapath and diagonal advance growth patterns is that, in the datapath paradigm, each of the tiles that place strength-2 glues which advance the path can also cause the datapath to perform an instruction as well. These instructions can cause the path to turn, move forward a defined amount using a binary counter, and place tiles to the right of the datapath among other instructions which will be explained shortly. It's important to keep in mind that, in addition to having the ability to perform instructions, the a datapath can also carry data along its width. Once the datapath finishes its instructions, this data is presented as strength-1 glues along the end of the datapath with which other structures can interact.

\paragraph{Latch}
The \emph{latch} (See Figure \ref{fig:SimpleLatch}) growth pattern is important to providing collision tolerance and preserving directedness.  The latch allows growth in one direction and prevents it in the reverse direction. From the prevented direction, once the latch tile grows it cannot grow the pre-latch tile because it only presents strength one glues.  From the allowed direction the pre-latch tile presents the same strength one glue toward the latch row, but it has a strength two glue that allows placement of the "Gen" and "Key" tiles, which allows cooperation into the latch row.  One important property of this latch mechanism is that once the latch is complete it is impossible to tell whether the pre-latch tile or the latch tile was placed first.  This property is vital to maintaining directedness in many situations.

\begin{figure}[htb]
\centering
\includegraphics[width=3.5in]{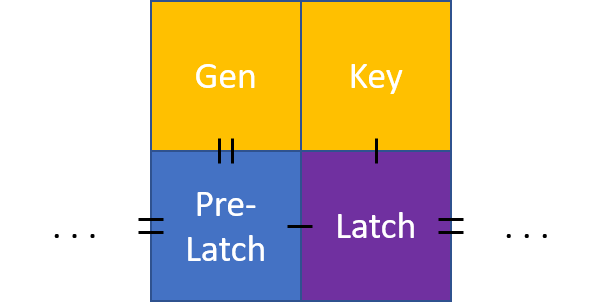}
\caption{The latch tile cannot generate the pre-latch tile, but the pre-latch tile can generate the latch tile}\label{fig:SimpleLatch}
\end{figure}

\subsection{General instructions}

In our construction, we use \emph{instructions} to refer to glues that attach new tiles in a certain way such that it accomplishes a specified movement or placement of data. Instructions provide an easy way of understanding and organizing the massive number of different glues and tiles that are used in moving data around. Instructions are mostly used in the $\genome$ and different datapaths. Multiple instructions are encoded into these gadgets by having each tile in a single row of the gadget encode a different instruction. Then, by growing in a diagonal advance pattern, each tile that is attached using a strength-2 bond is the ``active'' tile for that row, executing its instruction only. The instructions used by datapaths include:

\paragraph*{Buffer}
The buffer instruction advances a datapath one tile forward. Typically, this instruction is used to pad datapaths to a fixed length, but it can also be used to move small distances forward.

\paragraph*{Rise/Fall}
The rise and fall instructions encode an upward or downward turn respectively. Once interpreted, these instructions place a row of tiles forward with glues facing either up or down corresponding to the forward facing glues encoding the data in the datapath.

\paragraph*{Stop}
The stop instruction causes the datapath to stop growing forward and ignore all of the remaining instructions. The data encoded in these instructions, however, is still available to be read.

\paragraph*{Variable}
The variable instruction stalls the growth of the datapath until it receives some input from another set of tiles. This works by having the input data, growing orthogonal to the datapath using the unlimited strength-2 growth pattern, collide with the last row of the stalled datapath. The two gadgets can then cooperate to resume the execution of the remaining instructions in the datapath, while the datapath now encodes the information from the input data gadget.

Now, we describe a few additional instructions that work slightly different from the previous instructions in that they do not utilize the diagonal advance movement pattern.

\paragraph*{Forward}
The forward instruction is always followed by a series of tiles representing a number $c$ in binary. Using a standard, fixed width, binary decrementer, the forward instruction causes the datapath to move forward $c$ tiles as the number encoded is decremented until it equals 0. During the execution of a forward instruction, the forward propagation of the datapath is done by the strength-2 glues in the decrementer. Strength 1 glues along the left and right sides of the fixed width decrementer allow the tiles beyond the decrementer to fill in using cooperation. Once the decrementer is done, the instruction to the right of the tiles that used to represent $c$ becomes the active instruction and propagation returns to the diagonal advance pattern.

One thing to note is that, in our implementation of the datapaths, the decrementer used in the forward instruction decrements the encoded count on every other row of tiles. Furthermore, after the last row of the decrementer, one additional row of tiles is placed, so that the next instruction can be activated. Thus, for a number $c$ represented in the tiles after the forward instruction, the datapath, in our implementation, propagates $2c+1$ tiles forward. To design datapaths that grow an even number of tiles, the buffer instruction can also be used.

\paragraph*{Left/Right Turn}
The left turn and right turn instructions tell the datapath to turn in the respective direction relative to forward using a standard data rotation tile set.

\paragraph*{Place}
The place instruction moves the datapath forward 2 tiles and places a pair of tile below the rightmost tile of the new rows. Ideally a place instruction would be able to place a single tile rather than a pair; however, it is often the case that any placement tiles must be placed before the datapath can be allowed to continue growth. Because of this, the row in which the first tile is placed does not use a strength-2 glue to propagate and instead the datapath is advanced by a strength-2 glue between the placement tiles which guarantees their placement before allowing the datapath to continue.

In actuality, there are multiple tile types that encode place instructions, one for each pair of tile types that might need to be placed at some point in the construction.

\subsection{Callbacks}\label{sec:callback}

During the initialization phase of the construction, some datapaths need to grow fully before it's safe for others to grow, due to the use of unlimited strength-2 growth. In order to do this, we utilize a technique which we refer to as a \emph{callback} in which, once a certain datapath has either fully grown or has reached a special instruction, a single tile wide path grows along its right or left boundary. This growth requires strength-1 glues available along the far end of the boundary and requires a single tile wide open space in which the callback can grow. Since the boundary tiles don't require any tiles past the boundaries to operate, these conditions are easily met. It should also be noted that it's not necessary for a datapath to grow callbacks. In order to allow for datapaths with callbacks and datapaths without callbacks, different left and right boundaries can be given to the datapath. Some boundaries contain the necessary glues for the callback to grow and the others do not.

Furthermore, there are three different types of callbacks. The first two types are the right and left end callbacks. These begin once the datapath has finished growing and grow along the right and left boundaries respectively. The final kind of callback is a right variable callback. This grows once a variable instruction has caused the datapath to stop in order to wait for input. This variable callback is needed because, during initialization, some of the datapaths will have variable instructions that will not have receive input until later phases of the simulation. It's important that these datapaths are present before their input data begins growing since, otherwise, the paths might not collide properly. Therefore, the right variable callback allows a module that is setup during initialization to signal that it has been properly setup (thereby starting the datapath for the setup of the next module in the initialization process) despite it not having received input and growing to completion.

There are two distinct use cases for callbacks. The first and most common use case is to allow the next datapath in the initialization sequence to begin growth. In this use case, the datapath reaches either the last instruction in the datapath or a variable instruction and then sends a callback along its right side to activate the next primed datapath in the initialization. This assures that the initialization of certain components occurs in the correct order. The second use case is to signal to the genome that a tile has won the bracket and that differentiation has occurred. This callback grows after the datapath responsible for placing the bracket blocker gadget finishes growing. It grows along the datapath and then along the bottom of the genome to eventually activate the primed genome propagation.

Because callbacks have to grow backwards along the edges of their datapaths, it's important that they can perform all of the same turns that the datapath can. It's not difficult to see how a callback can propagate along the edge of a straight section of datapath, however the turn, rise, and fall instructions are a bit more complicated. Since, in the row before any turn instruction, the information that a turn will occur is passed along the width of a datapath, the boundary before a turn can present a glue which signals for the callback to place a tile with a glue in a direction orthogonal to the direction in which it was previously growing.

\subsection{Technical details for the $\genome$}

\subsubsection{Collision Tolerance and Circular Latches}\label{sec:collision-tolerance}

Collisions are a problem that can occur in the propagation of the $\genome$ whenever a path of tiles can grow in from two differing directions at the same time. If this happens, once the paths grow up to a single tile wide gap between them, unintentional cooperation can occur over the gap, causing non-determinism and potential errors. An example of this is depicted in Figure~\ref{fig:nondet-collision}.

\begin{figure}[htb]
\centering
\includegraphics[width=5.0in]{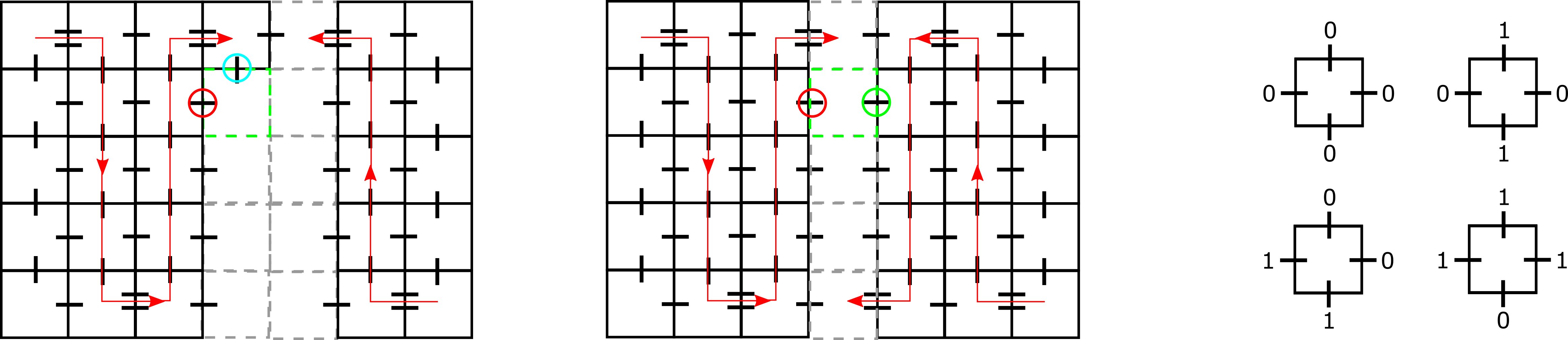}
\caption{An example of zig-zag growth which is initiated from both left and right sides, and grow toward a collision in the middle.  (Left) The fourth column on the left side happens to begin growth before third column on the right side begins growth.  This allows the tile which will bind in the location outlined in green to attach via the two adjacent glues (circled).  (Middle)  The third columns of both the left and right sides instead happen to complete before the fourth column begins. This allows the tile which will bind in the location outlined in green to attach via the two circled glues ``across the gap''.  (Right) Assuming that the glue label for the red circled glue is $0$, the blue circled glue is $1$, and the green circled glue is $0$, the ``correct'' tile for the location outlined in green is the top right tile.  Only this tile can attach in the scenario on the left, but in the scenario on the right, either of the top two tiles can attach.  This leads to nondeterminism and possible incorrect growth.}\label{fig:nondet-collision}
\end{figure}

The portions of the $\genome$ that utilize the diagonal advance growth pattern are susceptible to the errors caused by collisions. However, portions that utilize the limited strength-2 growth pattern are not susceptible to these errors, since these portions don't require tiles to cooperate with their neighbors within the same row to place the correct tile. Therefore, we have carefully designed the $\genome$ so that the potential collision locations are reduced to only portions that utilize the limited strength-2 growth pattern.

To prevent the diagonal advance portions of the $\genome$ from also being potential collision locations, these portions are grown within the confines of a gadget that we refer to as a \emph{circular latch}. This creates a combination of collision tolerant ``two-way'' regions and pairs of collision intolerant ``one-way'' regions, as shown in Figure \ref{fig:CircularLatch}.  The two-way regions operate using strength-2 glues to move a constant distance.  The one-way regions, or circular latches, use a single strength-2 glue per row to move forward and then cooperation to fill the row.  Each circular latch has a \emph{preferred} and a \emph{non-preferred} direction.  If data is to move in the non-preferred direction, it must rise out of its current plane and move forward until it reaches the next two-way region, at which point it can drop back down into the original plane.  Dropping back into the regular plane will then trigger the preferred direction to grow until it collides with the original path.  At the end of the process, it is impossible to tell whether the data grew from the preferred or non-preferred direction.

\begin{figure}[htb]
\centering
\includegraphics[width=6.5in]{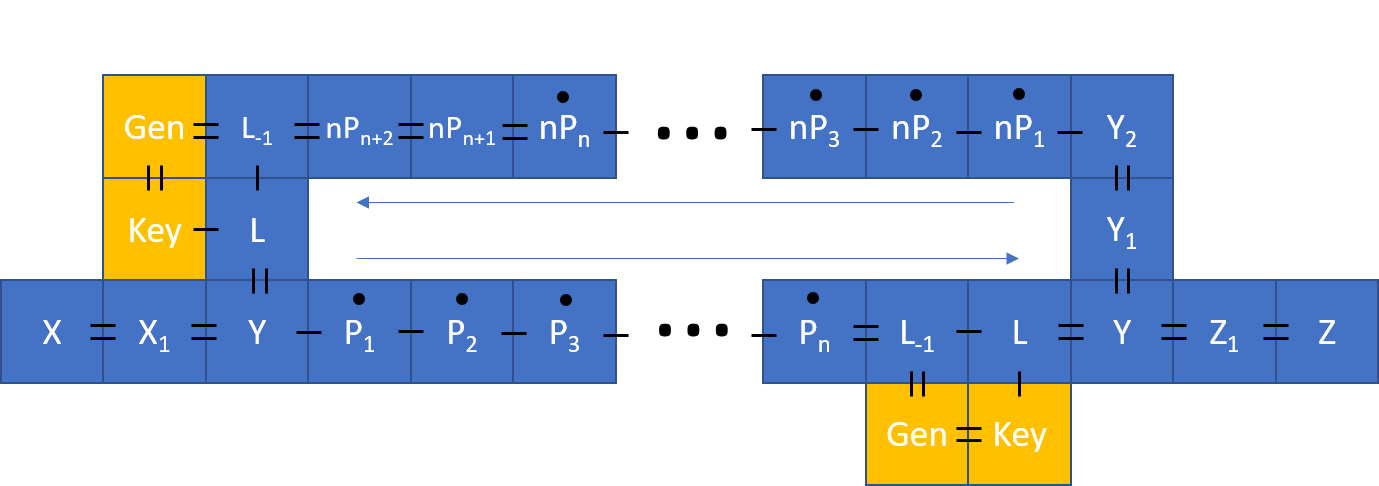}
\caption{Side view of circular latch with glues. Each square represents a full row of data.  Rows with dots use a diagonal advance movement pattern.}\label{fig:LatchMechanism}
\end{figure}

\begin{figure}[htb]
\centering
\includegraphics[width=6.5in]{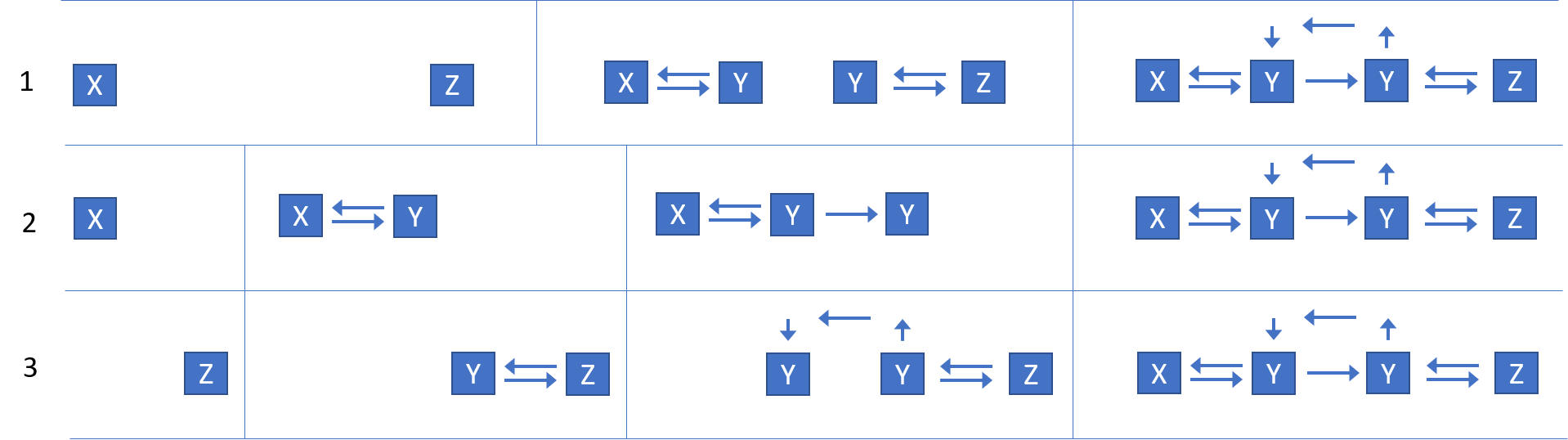}
\caption{Three possible growth sequences between point X and point Z}\label{fig:CircularLatch}
\end{figure}

\paragraph{Circular Latch Growth Sequences}
There are three types of growth patterns that can occur during a circular latch growth sequence, shown in Figure \ref{fig:CircularLatch}.  In Figure \ref{fig:CircularLatch} Points X and Z represent some checkpoint in the $\genome$ movement between which there exists a circular latch; points labeled Y represent the interface between the two-way and one-way sections of the circular latch.

\subsubsection{$\genome$ Specific Instructions}

These are additional instruction used specifically by the $\genome$ to propagate around macrotiles and to setup the modules in preparation for queries from other completed neighboring macrotiles.  A detailed layout of these instructions for each band can be found in Figure \ref{fig:InstructionsByBand}.

\paragraph*{Intersection}
The intersection instruction signifies the end of a circular latch region and triggers a limited strength-2 interface to the corner of one of the bands of the $\genome$. This instruction occurs once per orientation, or 24 times per macrotile.

\paragraph*{Turn}
The turn instruction signifies the end of a circular latch region and triggers the limited strength-2 interface to a cross band communication region.  This instruction occurs once per orientation, or 24 times per macrotile.

\paragraph*{Query}
This instruction signifies that a query may occur at this row of the $\genome$. This row of the $\genome$ is a priming row that primes all datapaths in $G_2$.  Although all datapaths are primed, the delimiters between each tile and each side within each tile are not primed which creates dead zones, ensuring that only the desired datapaths are activated.  This instruction occurs six times per macrotile only in the critical orientation.  Queries that occur in the non-preferred direction of a circular latch are ignored.

\paragraph*{Initialize}
This instruction primes $G_3$ and generates the activation signal for $G_3$, immediately beginning initialization once this instruction is reached.  This instruction occurs only once per macrotile in the critical orientation.

\paragraph*{Prop}
This instruction signifies an inert section of the $\genome$ which simply advances one tile location. Since the  A sequence of these forms a unary counter.  Differences in the amount of these instructions is responsible for the nesting of the bands.

\begin{figure}[htb]
\centering
\includegraphics[width=3.0in]{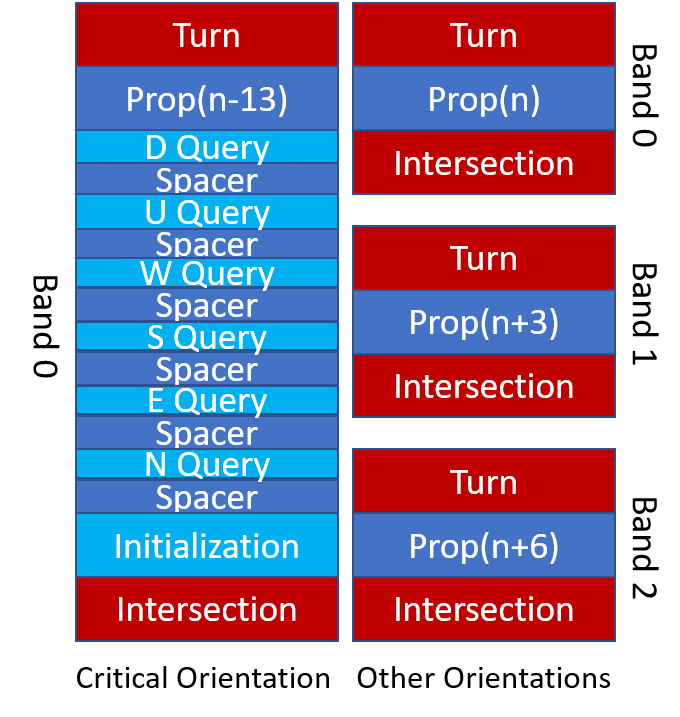}
\caption{Schematic depiction of sections of the bands of the $\genome$.  Each band is depicted as growing from bottom to top, with each horizontal row containing the entire contents of the $\genome$.  However, the $\genome$ contains instructions embedded within it, and as the bands grow each follows a specific set of those instructions. The labels in the figure show which types of instructions are being followed during that portion of forward growth.  Notice that in the critical orientation, the instructions cause the forward growth to result in specific sections which can be used for (1) initialization of macrotile structures, and (2) locations for $\externalCommunication$ datapaths from each direction to arrive to perform queries that may result in datapath growth of inputs to the $\adderArray$.}\label{fig:InstructionsByBand}
\end{figure}

\subsubsection{Input/Output}\label{sec:genomeIO}

\paragraph{Layout}
The the genome has one input region and one output region for each neighbor. Consider two neighboring fully differentiated macrotiles A and B.  If B is the Down neighbor of A, then A is the Up neighbor of B.  Both A and B receive the genome from and output the genome to the other.  A will output its genome to B from the intersection between its North and Down face, and B will receive its genome from A at the intersection between its North and Up face.  The opposite process will occur in the South of A and B, where B will output its genome from the intersection between its South and Up face, and A will receive its genome from B at the intersection between its South and Down face.  This example is illustrated in Figure \ref{fig:Genome_IO}.

\begin{figure}[htb]
\centering
\includegraphics[width=6.5in]{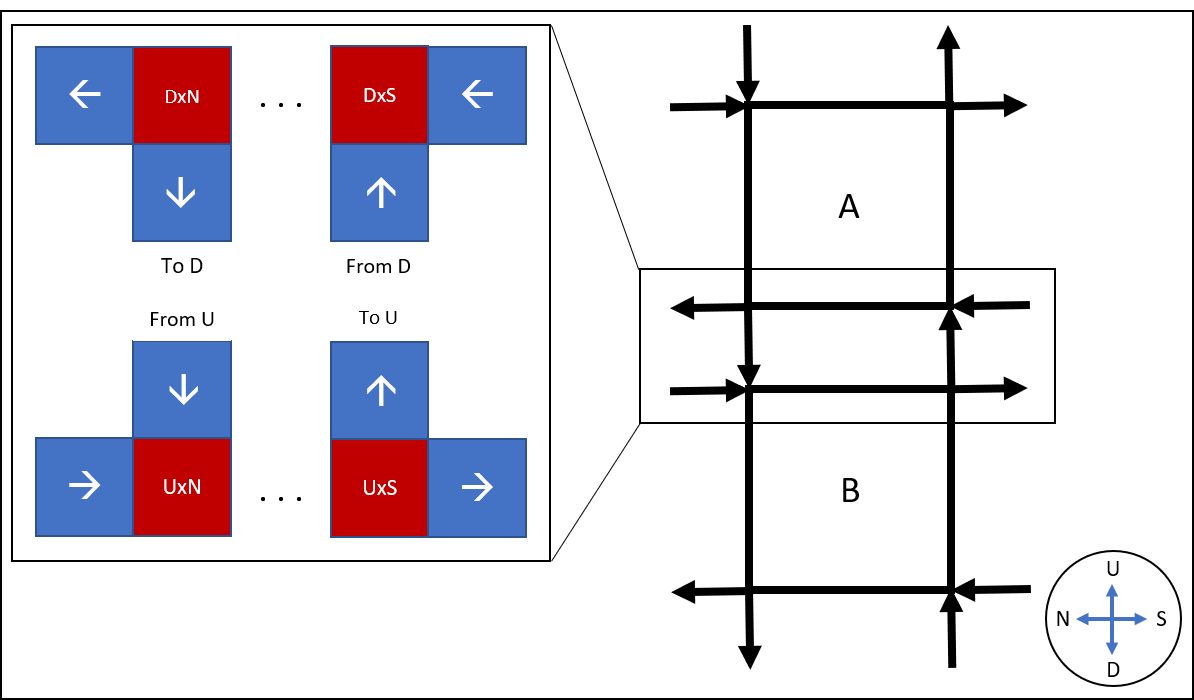}
\caption{Input and output arrangement between N/S/U/D neighbors.  Note that from the Down neighbor's perspective the current macrotile is the Up neighbor.}\label{fig:Genome_IO}
\end{figure}

\paragraph{Output process}
Genome output only needs to occur after the macrotile has differentiated, so the output regions are primed when the genome grows but only triggered after the differentiation callback (see \ref{sec:bracket-blocker}).  The differentiation callback causes a signal to propagate around the perimeter of the genome and places a trigger tile orthogonal to each of the output regions.  Once triggered, the genome begins growing toward the corresponding input region of the neighboring macrotile using limited strength-2 movement.

\paragraph{Input process}
The genome uses a latch to receive input. The latch serves two purposes. First, it prevents the genome from growing backward into the neighbor before the macrotile differentiates. Second, if the genome has already grown the latch makes it impossible to tell the order in which the inputs arrived.

\subsection{Technical details for the $\adderArray$}

\begin{figure}[htb]
\centering
\includegraphics[width=5.0in]{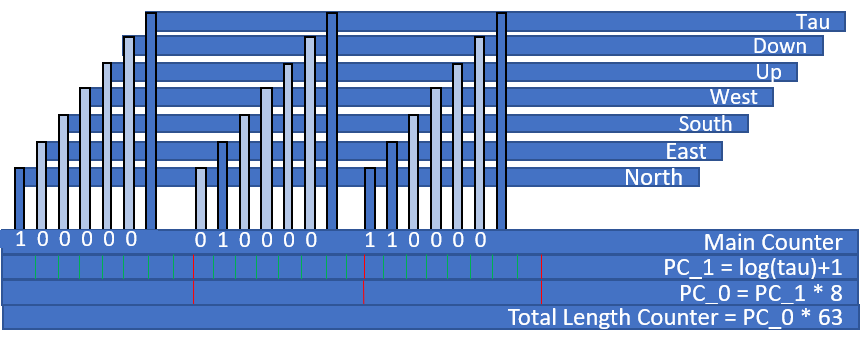}
\caption{Adder schematic (side)}\label{fig:AdderSchematicSide}
\end{figure}

\subsubsection{Inputs}
The top half of Figure~\ref{fig:AdderSchematicSide} shows that each input is sent into each component adder.  Each input is presented to each component adder, either to be rejected or accepted depending on the address of the component adder.  Not shown in Figure~\ref{fig:AdderSchematicSide} is that each of the inputs is guided by a copy of the Total Length Counter and $PC_0$.  Whenever the $PC_0$ zero signal is sent, the input data forks itself and drops down into the corresponding partial adder.  The guide rails are offset in the the forward direction by a multiple of the period of $PC_1$ to ensure that the input drops into the correct location.

\subsubsection{Periodic Counters}
The adder unit has several repeating structures which must appear multiple times, at regular intervals. Component adders must be regularly spaced along the adder, and partial adders must be regularly spaced along component adders.  We use a variant of binary decrementer called a \emph{periodic counter} to efficiently provide the required periodic structure.  A periodic counter functions like a regular binary decrementer, except that it preserves its starting value throughout its count and resets to that value once it hits zero.  If unrestricted, a generic periodic counter will continue repeating infinitely many times, which is not a desired behaviour in this construction.   We use a simple (non-periodic) decrementer initialized with the total desired length of the periodic counters to restrict the periodic counters.  A periodic counter receives signals from the layer below it and send signals to the the layer above it.  Periodic counters have two states: zero and non-zero.  When non-zero, the periodic counter passes a "continue" signal to the layer above.  When zero, the periodic counter sends a "zero" signal unique to its layer to the the layer above.  When multiple periodic counters are stacked, as in Figure~\ref{fig:AdderSchematicSide}, multiple zero signals may occur in several layers at the same point.  In this case, the bottom-most periodic counter's zero signal takes precedence and is displayed to the uppermost layer (See Figure \ref{fig:PeriodicCounterExample}).

\begin{figure}[htb]
\centering
\includegraphics[width=6.0in]{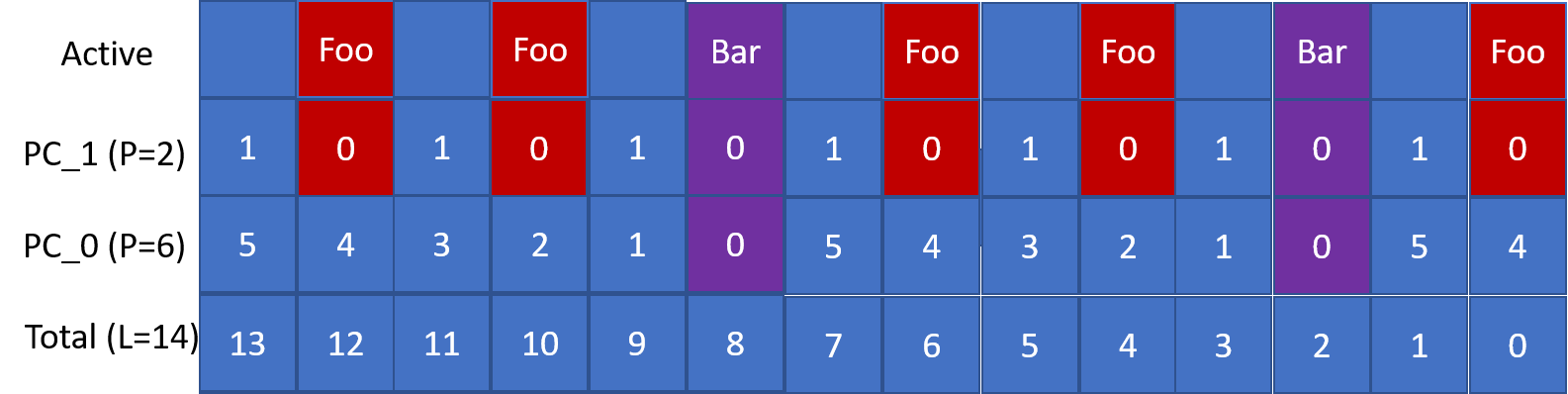}
\caption{Example of a 2-layer periodic counter}\label{fig:PeriodicCounterExample}
\end{figure}

\subsubsection{Main Counter}
The main counter layer contains the binary representation of the component adder.  Recall that a component adder is addressed by a six bit binary number where each bit represents whether to check or ignore input from a particular neighbor.  The main counter begins with the component adder address 000001 and increments each time the PC\_0 sends a zero signal.  Within each component adder, partial adders are constructed to be $log_2(\tau) + 1$ bits wide. $PC_1$ is set to repeat this distance plus one to allow a gap for the result of the addition to advance.  The component adder always considers the least significant bit (LSB) of the main counter first.  If the LSB is 1, then the partial adder is constructed to wait for input from the northern neighbor. If the LSB is 0, then the partial adder is constructed to ignore the input from the northern neighbor and present zero as the input.  When the $PC_1$ sends a zero signal, the main counter leaves a gap row and then constructs the partial adder of the next highest bit.

\subsubsection{Detailed layout}\label{sec:AdderLayout}
At initialization there are 19 layers of the adder aligned vertically which are each seeded by a different initialization datapath.  Starting from the lowest layer there is the non-periodic binary decrementer which is seeded with the full length of the adder.  Next there are two layers of periodic binary decrementers which are seeded with the size of the component adders and the subcomponent adders. The fourth layer up is the main counter layer.  The next 12 layers are made up a periodic counter seeded with the component adder size stacked on top of a non-periodic decrementer seeded with the full width for each of the six directional inputs.  The final three layers are similar to the previous six input groups but with the seventh input, a two's complement encoding of the temperature, stacked on top of the two decrementers.

\subsection{Technical details for the $\bracket$}\label{sec:bracket-details}

Whereas data being transported during a query uses a datapath, data being transported from a successful $\adderUnit$ through the $\bracket$ to the $\externalCommunication$ module uses a guide rail. This is because the $\bracket$ is initialized to facilitate the movement of guide rails as they grow through a series of points of competition. Recall that guide rails grow forward using the unlimited strength-2 growth pattern and therefore rely on other structures to change direction and stop growth.

\subsubsection{Turn barriers}
The first structure that a guide rail from the $\adderArray$ will interact with in the $\bracket$ is a \emph{turn barrier}. Turn barriers work by first blocking the guide rail and then placing a special tile using cooperation between the side of the guide rail and the second tile in the barrier that turns the guide rail and begins unlimited strength-2 growth to the left or right. The operation of the turn barrier can be seen in (1)-(3) of Figure~\ref{fig:bracketBarrier}.

If the barrier is a turn barrier, the backbone, after cooperating, will begin to propagate using strength-2 glues in the East or West direction depending on which side it cooperated. This propagation is temporary though and will only last until the backbone collides with a merge barrier. During this sideways propagation, the backbone will have a strength-1 glue in the upward direction to allow it to cooperate with the merge barrier once it collides.

\subsubsection{Merge barriers}
After colliding with a turn barrier, a guide rail in the $\bracket$ will grow indefinitely to the left or right of its original direction. It will then collide with a \emph{merge barrier}. Working similarly to a turn barrier, a merge barrier blocks a ``sideways'' growing guide rail and then cooperates with the side of the guide rail to place the first tile in a series of special tiles. The first special tile allows the attachment of four additional special tiles using limited strength-2 growth that move the guide rail to the point of competition. The last special tile is placed in a location where the two competing guide rails would overlap and has a generic strength two glue on its side that re-initiates unlimited strength-2 growth. This causes the guide rail to grow towards either the next turn barrier or the external communication module. The operation of the merge barrier can be seen in (3)-(5) of Figure~\ref{fig:bracketBarrier}.

\begin{figure}[htb]
\centering
\includegraphics[width=6in]{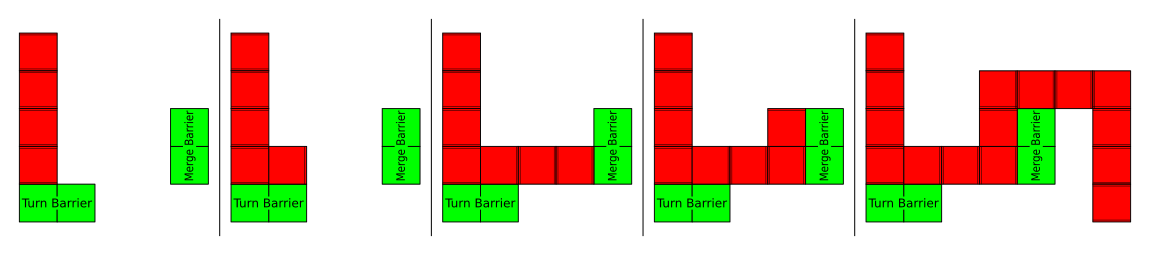}
\caption{(1) Backbone propagates until it collides with turn barrier (2) Cooperation between backbone and barrier (3) Eastward propagation until collision with merge barrier (4) Cooperation with merge barrier (5) Strength 2 path up and around merge barrier after which data propagates downward again}\label{fig:bracketBarrier}
\end{figure}

\subsubsection{Blocker mechanism}\label{sec:bracket-blocker}
After a tile wins the $\bracket$ but before the external communication is initiated, all unused inputs to the $\bracket$ are blocked off.  This prevents inputs received after the $\bracket$ completes from influencing the state of the $\bracket$ and is used to maintain directedness in specific circumstances which are discussed in the proof of Theorem~\ref{thm:directed3DIU}.
The datapath that blocks the $\bracket$ inputs (called the \emph{blocker datapath}) is the first datapath to receive data from the $\bracket$'s output.  Once the $\bracket$ winner data is received, it will not be allowed to continue to the external communication datapaths until the blocker datapath has sent a callback confirming that every unused $\bracket$ input has been blocked.  When the $\bracket$ winner data is received, the datapath navigates to the first layer of the $\bracket$ and places a tile which begins a zig-zag growth pattern of groups of four tiles, called blocker groups, shown in Figure \ref{fig:Bracket_Blocker_Mechanism}.  The blocker datapath grows along the length of the highest $\bracket$ layer with a 1 tile gap between its leftmost boundary and the $\bracket$ input locations (See Figure \ref{fig:Bracket_Blocker_Datapath}).  Each $\bracket$ input can be guaranteed to be an even number distance apart so that blocker group tile 3 will always be the tile that blocks the $\bracket$'s potential input locations.  If a $\bracket$ input region has already received input, then that input will have an exposed glue which allows for cooperation with tile 3 in order to continue blocking.  Once all inputs have been blocked and the blocker groups have reached the end of the blocker datapath, a callback is sent back to the variable input region of the blocker datapath to allow the winning tile data to propagate out of the $\bracket$ into the external communication datapaths.
Immediately upon receiving input from the bracket the bracket blocker sends a special callback called the \emph{differentiation callback} which allows the genome to propagate to its neighbors.  The bracket blocker is the first component initialized, which allows the differentiation callback upstream access to the start of the genome and thus the 12 intersection regions for genome output (See Section \ref{sec:genomeIO}).
Note there is a second blocker datapath which physically blocks the output guide rail from the final external communication datapath.  This datapath is trivial and does not have a specific name, but is also referred to as a blocker datapath in Section \ref{sec:datapath-widths}.

\begin{figure}[htb]
\centering
\includegraphics[width=5in]{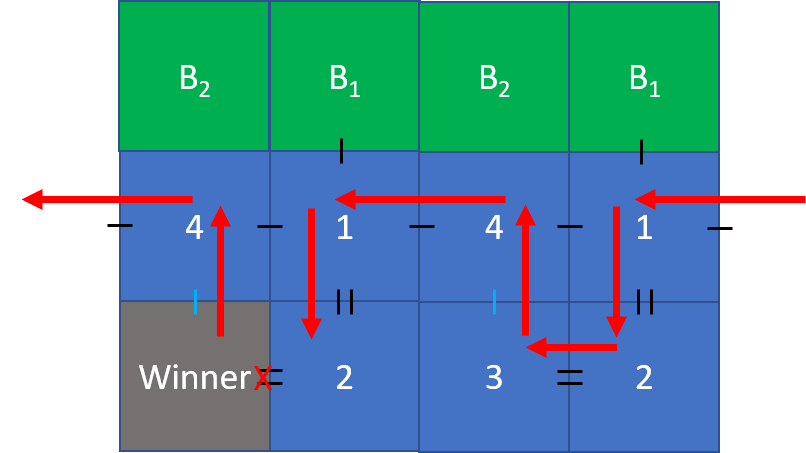}
\caption{Blocker mechanism is formed from repeating groups of tiles 1,2,3,4. 1 - Start tile, places 2. 2 - Places 3 if input does not exist.  3 - Not placed if bracket input already exists, cooperates with 1 to place 4.  4, cooperates with backbone(green) to place tile 1. } \label{fig:Bracket_Blocker_Mechanism}
\end{figure}

\begin{figure}[htb]
\centering
\includegraphics[width=3in]{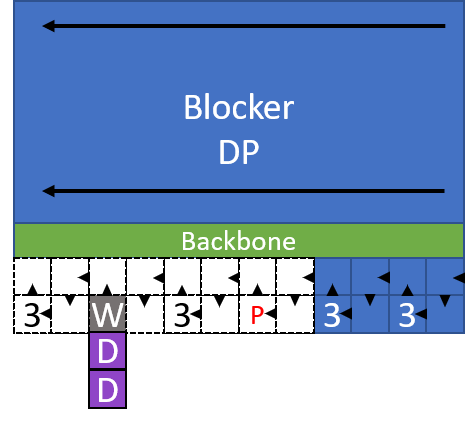}
\caption{P-Potential input, guide rail not arrived.  W-Winning input, guide rail arrived.  The blocker datapath grows along the length of the bracket input layer (2-input bracket shown) and places blocker groups to block all potential inputs.  Each bracket input will be blocked by a blocker tile 3.}
\label{fig:Bracket_Blocker_Datapath}
\end{figure}

\subsection{Technical details for the $\externalCommunication$}

While we talk about the $\externalCommunication$ module, the ``module'' really just consists of six variable datapaths (one for each direction) and a blocker. The initialization of the $\texttt{external}$ $\texttt{communication}$ is the growth of these datapaths to be directly under the $\bracket$ module. Each datapath grows to the necessary location and then initiates a callback to start the growth of the next datapath. The callback of the last variable datapath initiates the growth of the blocker.

Whenever a guide rail grows down from the $\bracket$, it triggers the \emph{bracket blocker} to begin blocking the bracket.  Once the bracket blocker is finished, it sends a callback to allow the winning data from the bracket to grow a guide rail to the external communication datapaths.  The guide rail collides with the second to last row of the top datapath. The collision allows the guide rail to cooperate with the last row of the top datapath and grow around using limited strength-2 growth to spawn a new row of the datapath. This new row has forward-facing glues that allow the datapath to continue to the neighboring $\genome$ and also downward-facing glues that drop a new guide rail (albeit, identical to the previous one) to the next stalled variable datapath. This process repeats for all six variable datapaths. The guide rail that drops out of the final variable datapath then grows into the blocker, where its growth is ultimately stopped.

\subsection{Details of Scale Factor and Seed Generation} \label{sec:scale-details}

The scale factor $m$ of our simulation of $\mathcal{T}$ is a positive integer that depends on the number of tiles $|T|$ and the temperature $\tau$ of $\mathcal{T}$. The size of $m$ is primarily determined by the width of a $\genome$ band, which in turn, is determined by the widths of the components, such as a number of datapaths, that make up the $\genome$ band. Therefore, in order to define an asymptotic bound on $m$, we will define a few useful values that describe how these components grow with respect to $|T|$, $\tau$ and $m$. It's important to note that, because certain datapaths will need to grow further as $m$ increases and because the size of $m$ is proportional to the size of the datapaths in our construction, there is a circular dependency between the widths of these datapaths and the scale factor. Fortunately, the width of a datapath grows logarithmically with respect to the forward distance it must travel, so this is not an issue.

\subsubsection{Datapath widths}\label{sec:datapath-widths}
Each datapath that will be described below is responsible for either placing tiles that will grow into structures such as the $\adderArray$ or for transporting data such as the strength of a certain glue between two tiles in $\mathcal{T}$. These datapaths all initially navigate to fixed locations within the macrotile which allows for the correct relative placement of structures such as the $\adderArray$ and $\bracket$. The datapaths navigate to the initial location by moving Down, North, and either West or East. Because the datapaths will never move forward more than $m$ tiles, this requires at most $8 + 3\log m$ tiles in the datapath ($2+\log m$ for each of the $3$ forward instructions and $2$ for the turn instructions in between). Also, each datapath requires an additional $3$ tiles to represent the left and right boundary tiles and the tile that represents the start instruction. Because each datapath will contain these tiles, for convenience we will define
$$|\text{Nav}| = 11 + 3\log m.$$

The first type of datapath is the type that moves from the glue table portion of the $\genome$ to the input of an $\adderUnit$. These datapaths each store a binary number representing the strength of the glue between the corresponding macrotiles. This strength is bounded by $\tau$ and thus the datapath needs to be at least $\log\tau$ tiles wide to be able to store the number. Since the only thing this datapath needs to do is navigate to its respective location, the number of tiles necessary to represent these datapaths is
$$|\text{DP}_{\text{glue}}|=\max(|\text{Nav}|,\log\tau)$$
$$\le |\text{Nav}| + \log\tau$$
\begin{equation} \label{eq:glue-dp-bound}
|\text{DP}_{\text{glue}}| \le 11 + 3\log(m\cdot\tau)
\end{equation}

Next are the datapaths responsible for initialization of macrotile structures. First, we will consider the datapaths responsible for placing the various rows of the $\adderArray$ ``seeds''. Each of these will, after navigating to the correct location, place $|T|$ instances of a single row of the adder seed separated by enough space to fit one of the glue genome datapaths between each instance. Each row of the adder seed is at most $6\log\tau$ tiles long and the space necessary to fit a glue genome datapath is $11 + 3\log(m\cdot\tau))$ tiles. For each of the $|T|$ adders, we will place the adder row tiles and move forward to account for the necessary spacing. The forward instruction width grows logarithmically with the distance travelled, so for each \adderUnit we will need at most $2 + 6\log\tau + \log(11 + 3\log(m\cdot\tau)))$ tiles (the additional $2$ comes from the fact that a forward instruction requires a tile before and after the binary representation of the forward amount). Therefore, the width of the adder placement datapaths is bounded by
$$|\text{DP}_{\text{adder}}|=|\text{Nav}| + |T|\big[2 + 6\log\tau + \log(11 + 3\log(m\cdot\tau))\big]$$
$$\le |\text{Nav}| + |T|\big[2 + 6\log\tau + \log(14\log(m\cdot\tau))\big]$$
$$\le |\text{Nav}| + |T|\big[2 + 6\log\tau + \log(14) + \log(\log(m\cdot\tau))\big]$$
$$\le 11 + 3\log m + |T|\big[6 + 7\log(m\cdot\tau)\big]$$
\begin{equation} \label{eq:adder-dp-bound}
|\text{DP}_{\text{adder}}| \le 11 + 3\log m + 6|T| + 7|T|\log(m\cdot\tau)
\end{equation}

We also have datapaths responsible for placing the $\bracket$. The $\bracket$ will have $\ceil{\log|T|}$ levels. Also, the number of inputs to the datapath will be the smallest power of $2$ greater than $|T|$; this is equal to $2^{\ceil{\log|T|}}$, which is bounded by $2|T|$. The spacing between the inputs will be equal to the spacing for the adders which was bounded by $11 + 3\log(m\cdot\tau))$ tiles. Each input of the $\bracket$ requires a turn barrier to be placed and between each pair of inputs needs to be placed two merge barriers for the $\bracket$ to work correctly. Placing the turn barriers requires $1$ instruction per input and placing the merge barriers requires $16$ instruction tiles because rising and falling is necessary to place the barriers perpendicular to the forward direction of the datapath. For each pair of inputs to the bracket, the datapath responsible will need to place a pair of turn barriers and a pair of merge barriers along with 3 forward instructions that move at most $11 + 3\log(m\cdot\tau))$ tiles each. Since their are $2|T|$ inputs or $|T|$ pairs of inputs, the width of a bracket datapath will be at most
$$|\text{DP}_{\text{bracket}}|=|\text{Nav}| + |T|\big[24 + 3\log(11 + 3\log(m\cdot\tau))\big]$$
$$\le |\text{Nav}| + |T|\big[24 + 3\log(14\log(m\cdot\tau))\big]$$
$$\le |\text{Nav}| + |T|\big[24 + 3\log(14) + 3\log(\log(m\cdot\tau))\big]$$
$$\le 11 + 3\log m + |T|\big[36 + 3\log(m\cdot\tau)\big]$$
\begin{equation} \label{eq:bracket-dp-bound}
|\text{DP}_{\text{bracket}}|\le 11 + 3\log m + 36|T| + 3|T|\log(m\cdot\tau)
\end{equation}
The final datapaths are the ones responsible for $\externalCommunication$. These are placed during initialization and contain two different navigational parts. The first set of navigation instructions bring these datapaths to just under the $\bracket$ so that they can intercept the winning tile information. The next set of navigation instructions bring the datapath to an adjacent macrotile to use the winning tile information to query. While the instructions necessary to navigate to an adjacent macrotile vary depending on the direction in which the $\externalCommunication$ is occurring, it's safe to say that the datapaths will contain at most 10 turns, rises, or falls and 5 forward instructions advancing no more than $2m$ tiles each. Also keep in mind that the $\externalCommunication$ datapaths are responsible for storing the winning tile ID and must thus be $\log|T|$ tiles long. Since each turn, rise or fall requires a single tile and each forward instruction requires at most $2 + \log(2m)$ or $3 + \log(m)$ tiles, the total number of tiles necessary for an $\externalCommunication$ datapath is
$$|\text{DP}_{\text{ext}}|=\max(|\text{Nav}| + 25 + 5\log m, \log|T|)$$
$$\le |\text{Nav}| + 25 + 5\log m + \log|T|$$
\begin{equation} \label{eq:ext-dp-bound}
|\text{DP}_{\text{ext}}| \le 36 + 8\log m + \log|T|
\end{equation}

There are also two blocker datapaths, one for blocking the $\bracket$ winning tile information once it has been read by all of the external verification datapaths and one for blocking the inputs to the $\bracket$ once a winner has been decided. Because both of these datapaths, after they navigate to their initial locations, only move within the macrotile to block some data, their widths are bounded by $|\text{DP}_{\text{ext}}|$ as well.

\subsubsection{Genome section widths}\label{sec:genome-section-widths}
Now that we have determined bounds for each of the datapaths that will occur in sections of the $\genome$, we can consider the bounds for the number of tiles in each of these $\genome$ sections themselves. The first section is the movement section. In this section, we encode for the width of a gap, through which the $\externalCommunication$ datapaths can move between adjacent datapaths. The width of this gap is represented in unary so there will be $|\text{DP}_{\text{ext}}|)$ tiles for each gap. The only other part of the movement $\genome$ section is a fixed number of single tile instructions. The number of these instructions is bounded by 40 and the unary gap is represented 4 times, so the width of this $\genome$ section is
$$|G_\text{move}|=4|\text{DP}_{\text{ext}}| + 40$$
$$= 4\big[36 + 8\log m + \log|T|\big] + 40$$
$$= 184 + 32\log m + 4\log|T|.$$

The next section is the glue $\genome$ section. This section encodes datapaths, in the form of a table, that courier information to the adders. The table has $6|T|^2$ entries, one for each cardinal direction per pair of tiles in the simulated system. Each entry in the table stores at most a single datapath of width $|\text{DP}_{\text{glue}}|$. There are at most 4 delimiters per entry in the table, so the number of tiles needed for this section of the genome is
$$|G_\text{glue}|=6|T|^2|\text{DP}_{\text{glue}}| + 24|T|^2$$
$$=6|T|^2\big[11 + 3\log(m\cdot\tau)\big] + 24|T|^2$$
$$=90|T|^2 + 18|T|^2\log(m\cdot\tau).$$

Finally the initialization section of the genome simply consists of several datapaths. There are 19 datapaths necessary to place the adder seed rows (Section \ref{sec:AdderLayout}), $\ceil{\log|T|}$ datapaths necessary to place the rows of the bracket (one for each level in the bracket), 6 datapaths for $\externalCommunication$ (one for each neighbor location), and the 2 blocking datapaths (Section \ref{sec:bracket-blocker}). Therefore the number of tiles necessary for this section of the genome is
$$|G_\text{init}| = 19|\text{DP}_\text{adder}| + \ceil{\log|T|}|\text{DP}_\text{bracket}| + 8|\text{DP}_\text{ext}|.$$
$$=19\big[11 + 3\log m + 6|T| + 7|T|\log(m\cdot\tau)\big] + \ceil{\log|T|}\big[11 + 3\log m + 36|T| + $$ $$3|T|\log(m\cdot\tau)\big] + 8\big[36 + 8\log m + \log|T|\big].$$

Notice that in all of these expressions, the largest term is of the order $|T|^2\log(m\cdot\tau)$. Since all we want is an upper bound, we can suppose that all of the terms in these expressions are terms of the order $|T|^2\log(m\cdot\tau)$ and simply add the coefficients to get a simpler upper bound.

$$|G_\text{move}|\le 220|T|^2\log(m\cdot\tau)$$
$$|G_\text{glue}|\le 108|T|^2\log(m\cdot\tau)$$
$$|G_\text{init}|\le 926|T|^2\log(m\cdot\tau)$$

\subsubsection{Scale factor bound}\label{sec:scale-factor-bound}

The scale factor the width of a $\genome$ band plus the width of two of the gaps encoded in the movement section so we can write
$$m = |G_\text{move}| + |G_\text{glue}| + |G_\text{init}| + 2|\text{DP}_{\text{ext}}|$$
$$= 1254|T|^2\log(m\cdot\tau) + 2(36 + 8\log m + \log|T|)$$
$$\le 1254|T|^2\log(m\cdot\tau) + 90|T|^2\log(m\cdot\tau)$$

Therefore we have
$$m \le 1344|T|^2\log(m\cdot\tau) = 1344|T|^2\log m + 1344|T|^2\log\tau$$

For convenience, let $c=1344$, $A=c|T|^2$ and $B=c|T|^2\log\tau$. Thus we have
$$m < A\log{m} + B$$
$$\frac{m}{A} < \log{m} + \frac{B}{A}$$
$$\frac{m}{A} < \log{m} + \frac{B}{A} - \log{A} + \log{A}$$
$$\frac{m}{A} < \log\frac{m}{A} + \frac{B}{A} + \log{A}$$
$$\frac{m}{A} - \log{\frac{m}{A}} < \frac{B}{A} + \log{A}$$
If we let $x = m/A$ we get
$$x - \log{x} < \frac{B}{A} + \log{A}$$
Notice that $x/3 < x - \log{x}$ for all positive $x$, therefore we have
$$\frac{x}{3} < \frac{B}{A} + \log{A}$$
$$\frac{m}{3A} < \frac{B}{A} + \log{A}$$
$$m < 3B + 3A\log{A}$$
$$m < 3c|T|^2\log{\tau} + 6c|T|^2\log{|T|}$$
Notice that $3c|T|^2\log{\tau} < 6c|T|^2\log{\tau}$, so
$$m < 6c|T|^2\log{\tau} + 6c|T|^2\log{|T|}$$
\begin{equation} \label{eq:sf-upper-bound}
m < 6c|T|^2\log(|T|\tau)
\end{equation}

Because $m < 6c|T|^2\log(|T|\tau)$ when $|T| > 1$ and $\tau > 1$, we have
$$m=O(|T|^2\log(|T|\tau)).$$

\subsubsection{Generation of the $\genome$ and seed}

In this section, we give pseudocode for the algorithm that generate the $\genome$ and a seed for a given $\calT$, and we analyze its run time.

\begin{algorithm}[H]

\SetKwFunction{S}{S}
\SetKwFunction{EmptyAssembly}{EmptyAssembly}
\SetKwFunction{CalculateScaleFactor}{CalculateScaleFactor}
\SetKwFunction{MovementInstructions}{MovementInstructions}
\SetKwFunction{GetLocation}{GetLocation}
\SetKwFunction{GenomeStartLocation}{GenomeStartLocation}
\SetKwFunction{AdderStartLocation}{AdderStartLocation}
\SetKwFunction{BracketStartLocation}{BracketStartLocation}
\SetKwFunction{BracketEndLocation}{BracketEndLocation}
\SetKwFunction{GlueDatapath}{GlueDatapath}
\SetKwFunction{CommDatapath}{CommDatapath}
\SetKwFunction{BlockerDatapath}{BlockerDatapath}
\SetKwFunction{BracketDatapath}{BracketDatapath}
\SetKwFunction{AdderDatapath}{AdderDatapath}
\SetKwFunction{SharesGlue}{SharesGlue}
\SetKwFunction{PlaceTiles}{PlaceTiles}
\SetKwFunction{GlueDelimeter}{GlueDelimeter}
\SetKwFunction{AdderSeedRows}{AdderSeedRows}
\SetKwFunction{AdderSpacing}{AdderSpacing}
\SetKwFunction{TileID}{TileID}

\SetKwProg{Func}{Function}{:}{}

\Func{\S{$\mathcal{T} = (T,\sigma,\tau) \in \frakC$}}{
\DontPrintSemicolon

$m$ := \CalculateScaleFactor{$\mathcal{T}$}
\BlankLine
\tcc{Create an empty assembly with tile set $U$ that will become our seed}

$\sigmaT$ := \EmptyAssembly{$U$}\;
\BlankLine

\For{$t_s\in \sigma$}
{
\tcc{Choose Starting Locations for the Genome, Adder, and Bracket}
$p_\text{genome}$ := \GenomeStartLocation{$\mathcal{T}$, $m$, $t_s$}\;
$p_\text{adder}$ := \AdderStartLocation{$\mathcal{T}$, $m$, $t_s$}\;
$p_\text{bracket}$ := \BracketStartLocation{$\mathcal{T}$, $m$, $t_s$}\;
$p_\text{winner}$ := \BracketEndLocation{$\mathcal{T}$, $m$, $t_s$}\;
\BlankLine

$p_\text{current}$ := $p_\text{genome}$\;
\BlankLine

\tcc{Generate tiles for the movement genome}
\PlaceTiles{$\sigmaT$, $p_\text{current}$, \MovementInstructions{$\mathcal{T}$, $m$}}\;
\BlankLine

\tcc{Generate tiles for the glue genome}
\PlaceTiles{$\sigmaT$, $p_\text{current}$, \GlueDelimeter{``Section Start''}}\;
\BlankLine
\For{$t_a\in T$}
{
\BlankLine
\PlaceTiles{$\sigmaT$, $p_\text{current}$, \GlueDelimeter{``Tile Start''}}\;
\For{$d \in \{N,E,S,W,U,D\}$}
{
\BlankLine
\PlaceTiles{$\sigmaT$, $p_\text{current}$, \GlueDelimeter{``Side Start '' + $d$}}\;
\For{$t_b\in T$}
{
\BlankLine
\tcc{Check if $t_a$ has the same glue on its $d$ face as $t_b$ has on the opposite face}
\If{\SharesGlue{$t_a$, $t_b$, $d$}}
{
\PlaceTiles{$\sigmaT$, $p_\text{current}$, \GlueDatapath{$p_\text{current}$, $p_\text{adder}$, $t_a$, $t_b$, $d$}}\;
}
}
}
}
\BlankLine
\tcc{Generate tiles for the initialization genome}
\PlaceTiles{$\sigma^U$, $p_\text{current}$, \BlockerDatapath{$p_\text{current}$, $p_\text{winner}$}}\;
\For{$d\in\{N,E,S,W,U,D\}$}
{
\PlaceTiles{$\sigmaT$, $p_\text{current}$, \CommDatapath{$p_\text{current}$, $p_\text{winner}$, $|T|$, $d$}}\;
}

\BlankLine
\For{$i\in\{1,2,\ldots,\ceil{\log|T|}\}$}
{
\PlaceTiles{$\sigmaT$, $p_\text{current}$, \BracketDatapath{$p_\text{current}$, $p_\text{bracket}$, $|T|$, $i$}}\;
}

\BlankLine
$adderSeed$ := \AdderSeedRows{$\mathcal{T}$}\;
\For{$r\in adderSeed$}
{
\PlaceTiles{$\sigmaT$, $p_\text{current}$, \AdderDatapath{$p_\text{current}$, $p_\text{adder}$, $|T|$, $r$}}\;
}

\BlankLine
\tcc{Place the seed tile ID at the end of the bracket}
\PlaceTiles{$\sigmaT$, $p_\text{winning}$, \TileID{$t_s$}}\;
}
}
\caption{\texttt{Seed Generation Function}}
\end{algorithm}

\paragraph{\texttt{CalculateScaleFactor}} This function calculates the scale factor for the assembly. It does so by first calculating an upper bound for the scale factor using inequality \ref{eq:sf-upper-bound} from Section~\ref{sec:scale-factor-bound}. It then uses this to compute the widths of the datapaths that will make up the genome using inequalities \ref{eq:glue-dp-bound}, \ref{eq:adder-dp-bound}, \ref{eq:bracket-dp-bound}, and \ref{eq:ext-dp-bound}. The actual scale factor is then computed using these datapath lengths which is guaranteed to be less than the previously computed upper bound. Because this computation is done using a fixed number of arithmetic operations on values no bigger than $O(|T|^2 \log(|T|\tau))$, the complexity of this subroutine is logarithmic in $O(\log(|T|) + \log(\log(|T|\tau)))$.

\paragraph{\texttt{EmptyAssembly}}
This function simply creates an empty assembly which will become our seed assembly. The assembly is stored as an associative array where the coordinates of the tile are used as a key for the type of tile being stored. The initialization of such a data structure can be done in constant time.

\paragraph{\texttt{StartLocation/EndLocation}} There are functions that generate a starting or ending location for a certain component of the construction. The starting location for the genome is simply a fixed number of tiles away from the north, west, up corner of the particular macrotile in which it's placed. And the starting locations for the components represent some locations within the macrotile at which the components will be placed. The absolute locations for each of these components can be computed as some fixed offset relative to the absolute position of a corner of the corresponding macrotile.

The location of a corner of a macrotile can be found by multiplying the coordinates of the corresponding seed tile in $\sigma$ by $m$. Since $\sigma$ must be connected, the largest coordinate of any tile in $\sigma$ can be bounded by $O(|\sigma|)$, where $|\sigma|$ is the number of tiles in $\sigma$. Thus the arithmetic for finding the corner of a super tile and any of the starting locations can be done in $O(\log(|\sigma|\cdot m)))$ time.

\paragraph{\texttt{PlaceTiles}} This function is responsible for placing all of the tiles into our previously initialized assembly. This function takes, as input, an assembly in which to place the tiles, a location at which to place the tiles, and a list of tiles to place. This function simply iterates over all of the tiles in the given list and inserts them into the assembly using the location as a key. After a tile is inserted, the location is incremented one tile in the East direction so that the next tile in the list will be placed adjacent to the last.

It's important to notice that the size of the associative array needed to store our assembly is $O(|\sigma|\cdot m)$. This is because, for each macrotile, we only seed a single row of tiles which will grow into the bands of the genome and everything else inside of the macrotile. Because this size is not fixed, the hash function with which we insert tiles into the assembly must have a range proportional to the size of the assembly. The hash function complexity can then be bounded by $O(\log(|\sigma|\cdot m))$. The computational cost of this function is therefore linear in the size of the given list of tiles times the complexity of performing a hash for each element.

\paragraph{\texttt{MovementInstructions}} This function generates a list of tiles representing the instructions for the movement part of the \texttt{genome}. The movement instructions consist of a number of \emph{Prop} instructions linear in the width of an external communication datapath plus a fixed amount of other instruction tiles. The size of the list of tiles returned by this function, and therefore its time complexity, is thus $O(|\text{DP}_{\text{ext}}|)$ or $O(\log m + \log|T|)$.

\paragraph{\texttt{GlueDelimiter}} This function simply returns a tile with the given name corresponding to some delimiter used in the glue table of the $\genome$. Since this function simply finds the tile with a given name from the fixed universal tile set, it can be done in constant time supposing tile sets are stored as hash tables.

\paragraph{\texttt{Datapaths}} There are five functions that return lists of the tiles necessary to place the datapaths which make up much of the genome. Because these functions simply return a list of tiles proportional to the width of the necessary datapaths, the time complexities of these functions are linear in the widths of their respective datapaths as discussed in Section~\ref{sec:datapath-widths}.

\paragraph{\texttt{AdderSeedRows}}
This function returns a list of rows of tiles, each containing the tiles necessary to seed a row of the adder construction as described in Section~\ref{sec:AdderLayout}. Instructions to place these tiles will make up part of the adder datapaths. There are 19 rows, each of which is no longer than $O(\log(\tau))$ tiles.

\paragraph{\texttt{SharesGlue}}
This function simply determines if two given tile types share a glue on the face in the specified direction. Because the maximum number of glues in a tile set is linearly proportional to the number of tiles in that set, the number of symbols necessary to represent each glue in $T$, is $O(\log|T|)$. Thus comparing two glues would take $O(\log|T|)$ time.

\subsubsection{Run Time Complexity}
In this algorithm, we first compute $m$, which takes $O(\log(|T|) + \log(\log(|T|\tau)))$ time. Then, For each tile in the given seed $\sigma$, we do the following. First we find the start and end locations of the various components which is done in time

$$O(\log(|\sigma|\cdot m)).$$

Next we place the tiles corresponding to the movement section of the $\genome$. This runs in time

$$O(t_\text{place}\cdot (\log m + \log|T|)),$$

where $O(t_\text{place})$ is the runtime complexity of the \texttt{PlaceTiles} function for convenience. Next we place the delimiters and glue datapaths necessary to make up the glue table part of the $\genome$. Remember from Section~\ref{sec:datapath-widths} that each of the datapaths in the glue table take $O(\log(m\cdot\tau))$ tiles. Since we are placing $O(|T|^2)$ datapaths, and since checking if two tiles in $T$ share a glue on a given face costs $O(\log|T|)$ time, the time complexity of this part of the algorithm is

$$O(|T|^2\log|T| + |T|^2\log(m\cdot\tau)\cdot t_\text{place}).$$

Notice that this bound takes into account the delimiter tiles placed in between the datapaths. Next we place the 2 \emph{bracket blocker} datapaths, the 6 external communication datpaths, the $\ceil{\log|T|}$ bracket datapaths, and the 19 adder datapaths. From Section~\ref{sec:datapath-widths}, the run time cost of this part of the algorithm is

$$O(t_\text{place}(|\text{DP}_\text{adder}| + |\text{DP}_\text{ext}| + \log|T||\text{DP}_\text{bracket}|)).$$
$$=O(t_\text{place}(|T|\log(m\cdot\tau) + \log|T| + \log m + |T|\log|T|\log(m\cdot\tau)))$$
$$=O(t_\text{place}(|T|\log|T|\log(m\cdot\tau)))$$

Finally we place the tiles that seed the given macrotile with the ID of tile $t_s$. The complexity of this is $O(t_\text{place}\cdot\log|T|)$ since there are at most $|T|$ different IDs that $t_s$ could have. This term fits into the complexity of the previous terms. The total run time complexity for placing the genome for a single tile in $\sigma$ is therefore

$$O(t_\text{place}|T|^2\log(m\cdot\tau) + |T|^2\log|T|)$$
$$=O(\log(|\sigma|\cdot m) |T|^2\log(m\cdot\tau) + |T|^2\log|T|)$$

Finally, by combining these terms, we find that the total run time complexity of the entire algorithm is

$$O(|\sigma| \cdot \Big[\log(|\sigma|\cdot m) |T|^2\log(m\cdot\tau) + |T|^2\log|T|\Big] + \log(|T|) + \log(\log(|T|\tau)))$$
$$=O(|\sigma||T|^2\log|T|\log(m\cdot\tau)\log(|\sigma|\cdot m))$$
$$=O(|\sigma||T|^2\log(|T||\sigma|\tau m))$$

Since $m=O(|T|^2\log(|T|\tau))$ this expands to

$$O(|\sigma||T|^2\log(|T|^3|\sigma|\cdot\tau\log(|T|\tau)))$$

Simplifying, we find that a bound for the run time of our algorithm is

$$O(|\sigma||T|^2\log(|T||\sigma|\tau))$$

\section{Full Proof for Directed 3D aTAM is Intrinsically Universal}
\label{sec:thm2-proof}

In this section, we provide technical details for Section~\ref{sec:directed3d-short} and Theorem~\ref{thm:directed3DIU}. To do this, we first prove that the construction from Section~\ref{sec:construction-short} is a directed simulator when simulating a directed system.

\begin{lemma}\label{lem:directed-simulation}
Let $\calT$ be an arbitrary directed 3D aTAM system. The system $\calUT = (U, \sigma_\calT, 2)$, which simulates $\calT$ using the construction of Theorem~\ref{thm:3DaTAMIU}, is also directed.
\end{lemma}

To prove Lemma~\ref{lem:directed-simulation} we show how the modules of our construction are designed to maintain directedness in spite of nondeterministic rates of growth of the components.
We proved in Lemma~\ref{lem:timing} that the order in which different modules grow does not have any effect on the overall tile placement within a macrotile. Recall that any timing issues in this case are mitigated by the need for cooperation between the modules that are racing. 
Therefore, we will analyze all of the possible tile locations within modules where there can be competition to place tiles, and we show that all of these tile locations end up with the same placement of tiles regardless of the timing of their placements.

To determine the race conditions that need to be discussed, we go through the whole growth process. The growth of two modules is initiated by neighboring macrotiles, the $\genome$ and $\externalCommunication$ datapaths. Different instances of the $\genome$ initiated by different neighboring macrotiles is the first race condition, since each individual neighboring macrotile can propagate the entire $\genome$, and yet the currently growing macrotile will only end up with exactly one instance of the full $\genome$. The $\externalCommunication$ datapaths, however, each have a designated path that they will grow on to a designated row in the query section of the $\genome$ of the currently growing tile (based on their relative direction from the macrotile), and therefore do not result in a race condition. These paths have a unique set of movement instructions for each direction they represent. Once at the $\genome$, they initiate the growth of query datapaths. Each of these datapaths also has a unique destination to which they will grow whenever initiated, unique to the tile type, direction, tile type triple that they represent, thereby never creating a race condition. These query datapaths then activate different $\componentAdder$ pieces, each of which may initiate a success signal that grows out of the $\adderUnit$. Since each $\componentAdder$ within a single $\adderUnit$ can individually cause the success signal tiles to start attaching, this becomes our second possible race condition. Next, the success signals initiate growth of guide rails that leave the $\adderArray$ for the $\bracket$ where they compete, creating our third and final race condition. Afterwards, the winning guide rail grows into the $\externalCommunication$ module, which initiates $6$ datapaths, each on their own unique path to a different neighboring macrotile, as well as the path back along the $\genome$ which initiates its growth into the $6$ neighbors. These three possible race conditions are now examined.

\paragraph{Genome} By Lemma~\ref{lem:genome-growth}, we know that the $\genome$ will place the same tiles regardless of which neighbor initiates the growth or whether multiple neighbors cause partial growth. Recall that this is because circular latches cause collision-prone portions of the $\genome$ to only grow in one direction at a time. Therefore, different neighboring macrotiles that simultaneously propagate the $\genome$ to the currently growing macrotile will have their growth intersect in collision-robust portions of the $\genome$ which use the limited strength-2 pattern of growth (see Section~\ref{sec:growth-patterns} for details). Thus, regardless of which neighboring macrotiles propagate the $\genome$ and in what order, the set of all tile placements within the $\genome$ bands will always be the same.

\paragraph{Adder Unit} The second race condition happens within particular $\adderUnit$ modules whenever multiple $\componentAdder$ pieces can all cause the $\adderUnit$ to succeed. Whenever any particular $\componentAdder$ does succeed, it sparks the primed growth of a ``success signal'' that moves forward and backwards along a backbone of generic primed strength-1 glues to place a series of generic success tiles along the whole $\adderUnit$. The success tiles cannot initiate growth into an inactivated $\componentAdder$, however, since they only present strength-1 glues and have nothing to cooperate with in that direction. Therefore, whenever another $\componentAdder$ succeeds, it grows down and merges with the already formed success signal directly below it, which consists of the same set of tiles that would have been placed by the new $\componentAdder$. See Figure~\ref{fig:Component_Adder_Success_Signal} for an illustration of this. By the end of the assembly process, the terminal assembly will then have the entire success signal filled out, with every $\componentAdder$ that succeeds over the course of the assembly process being connected to it, rendering an observer unable to distinguish which $\componentAdder$ succeeded first by solely looking at that part of the construction.

\paragraph{Bracket} The last race condition is the only one dependent upon the assumption that $\calT$ is directed. Since we assume this, and because we have proved that $\calUT$ is valid simulator for $\calT$ (without the constraint of directedness), we know that each specific macrotile in $\calUT$ will only have one guide rail that can grow into the $\bracket$ before it differentiates. This eliminates any race condition caused by the guide rails, since there will only ever be one ``winner'' guide rail. However, even when $\calT$ is a directed system, there is potential within our construction for multiple $\adderUnit$s within a single $\adderArray$ to receive inputs of enough strength from incident neighboring glues to output to the $\bracket$.  Such a situation is depicted (in 2D) by the blue tile in Figure~\ref{fig:bracket-blocker-scenario}.  It is important to note that only a blue tile type could ever be placed in that position, and the only potential nondeterminism could have come from the later inputs into that macrotile location (by neighbors whose attachments required the blue tile to be there first) as they possibly competed with each other within the $\bracket$ before ultimately losing to the blue path. For this reason, our construction includes the $\bracketBlocker$ mechanism (see Section~\ref{sec:bracket-blocker}) which is a module that grows after a path has won but before the macrotile initiates its output to neighboring macrotiles and blocks future inputs to the $\bracket$. Only after this blocking mechanism is in place will the macrotile output to neighbors.  Since we know that $\calT$ is directed, it must be the case that in any location where multiple tile types could match glues with enough strength to bind in that location, growth of at least one of the macrotiles which inputs those glues must require the differentiation and output of the macrotile in that location first, but since the $\bracket$ is blocked off for later inputs before the output of that macrotile, no undirected growth within the $\bracket$ can occur, and thus $\calUT$ remains directed.

\begin{figure}[htb]
\centering
\includegraphics[width=5.0in]{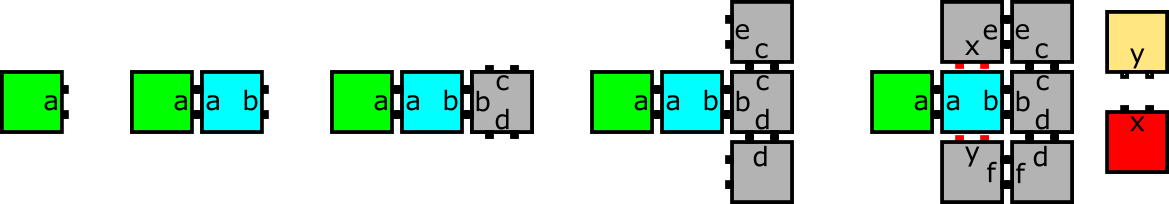}
\caption{(Left to right) A basic example, shown in 2D, in which an assembly can grow from a seed (the green tile).  When the final two tiles are placed (above and below the blue tile), they have strength-2 glues facing the location of the blue tile but which do not match glues on that tile (i.e. they are mismatches). If the blue tile were not present, the red tile could bind to the top tile or the gold tile could bind to the bottom tile.  However, this system is directed because the blue tile must always precede the placement of both the top and the bottom.} \label{fig:bracket-blocker-scenario}
\end{figure}

\begin{figure}[htb]
\centering
\includegraphics[width=4.5in]{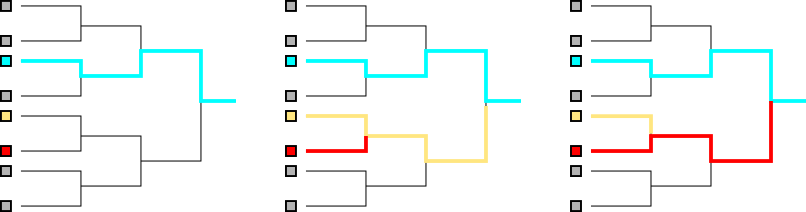}
\caption{A schematic depiction of the $\bracket$ in the macrotile representing the blue tile in Figure~\ref{fig:bracket-blocker-scenario} if the $\bracketBlocker$ mechanism did not exist.  (Left) The blue tile type must always win the $\bracket$ since the blue tile must differentiate and output to the neighbor to the right before any other inputs are possible. (Middle) If the growth from the southern neighbor happened more quickly than that of the northern, the path representing the gold tile type could win the point of competition with that of the red tile type (although it could never beat the blue, which must have won long before the others could ever grow). (Right) Conversely, if the growth of the northern neighbor was quicker, the path for the red tile type could win the $\bracket$. However, the $\bracketBlocker$ grows immediately after the blue path wins but before the macrotile outputs to neighbors, and thus prevents those later inputs from entering the $\bracket$, which preserves the directedness of the construction regardless of the timing of the growths from those neighbors.}\label{fig:nondet-brackets}
\end{figure}

\begin{proof}[Proof of Lemma~\ref{lem:directed-simulation}]
The proof of Lemma~\ref{lem:directed-simulation} follows from the fact that, given that $\calT$ is directed, any nondeterminism which can arise during the simulation by $\calUT$ cannot lead to instances of two assembly sequences which can place tiles of different types into the same location.  Therefore, $\calUT$ must be directed.
\end{proof}

\begin{proof}[Proof of Theorem~\ref{thm:directed3DIU}]
The proof of Theorem~\ref{thm:directed3DIU} follows directly from Theorem~\ref{thm:3DaTAMIU} and Lemma~\ref{lem:directed-simulation} since they prove that $U$ is intrinsically universal for the class of directed 3D aTAM systems, and thus the directed 3D aTAM is intrinsically universal.
\end{proof}

\section{Technical Details for Spatial aTAM IU Construction}
\label{sec:spatial_tech_details}

In this section, we provide technical details for Section~\ref{sec:spatial-atam}. To make the blocking protocol implemented by our construction for the proof of Theorem~\ref{thm:spatial-IU} more concrete, we will present an example of how it is implemented. Note that the construction could be more symmetric and timing efficient at the expense of using more components and having more timing dependencies. Our implementation relies on the passing of signals and information using datapaths and the use of cooperation to enforce sequential processes such that one component must wait on another to complete before processing itself in order to make the explanation easier to follow.

Presented here is a more in-depth description of the gadgets and protocols used to show that the Spatial aTAM is intrinsically universal. Note that this construction is very similar to the construction for showing that the 3D aTAM is intrinsically universal, with the addition of a blocking protocol which makes sure that each completed macrotile is completely encapsulated in a shell of tiles. This shell of tiles will serve to cause macrotiles to constrain a space exactly when tiles in the simulated space do. In this section we will simply describe the distinctions between the 3D aTAM IU construction and this construction and demonstrate that this shows that the Spatial aTAM is intrinsically universal. Also, it's important to notice that this construction is not concerned with preserving directedness.

For this section, let $\texttt{boundary}_{dir}$ be the set of tile locations that form the planar segment that is furthest in the $dir$ direction for a given macrotile where $dir \in \{U,N,E,S,W,D\}$. Let the tile location in the center of the pipe intersection be $\texttt{loc}_{center}$, since this is the centerpiece of the construction. Let $side$ designate the subset of directions $\{N,E,S,W\}$. Let $\texttt{pipe}_{dir}$ be the pipe subassembly that grows from the $\texttt{loc}_{center}$ to $\texttt{boundary}_{dir}$ for $dir \in \{N,E,S,W,D\}$.

\subsection{Incoming genome and external communication}
To have the correct set up for our blocking protocol to work, we first need to make a few adjustments to how the genome (and external communication) propagate to new macrotiles. The genome is a good module to make this adjustment in, since an incoming genome signifies that a neighbor has already differentiated. The adjustment is threefold. First we need the genome (and external communication) to come in at a new location. Figure~\ref{fig:side_face} shows an example of this for the $side$ faces. 
Next, we need the genome (and external communication) to exhibit a series of special glues at the boundary when it first comes into the current macrotile. These glues must be exposed in all directions perpendicular to the direction it is traveling. These glues will be used to direct the tiling of the boundary around the datapaths at a later stage.

Most importantly though, we need the incoming genome to generate what we refer to as an \emph{interception gadget} (shown in Figure~\ref{fig:intercept}). This gadget is used later to ``cut'' a hole in some of the piping to redirect tiles to the ``up'' direction, as previously mentioned. Once the pipes start forming, the gadget works by using special glues to cooperate with the attaching tiles such that the pipe can grow past it but will not fill in one location in the direction of the genome. Then, once the macrotile differentiates at a later stage, if there is a diffusion path from the ``up'' direction, tiles will be able to grow through this hole, along the genome, and signal that the macrotile is clear to propagate its information in the ``up'' direction. Note there will be 5 inception gadgets, one for each direction other than ``up'', which will refer to as $\texttt{intercept}_{dir}$ where $dir \in \{N,E,S,W,D\}$.

\begin{figure}
\begin{subfigure}[b]{0.45\linewidth}
\centering
\includegraphics[width=0.42\linewidth]{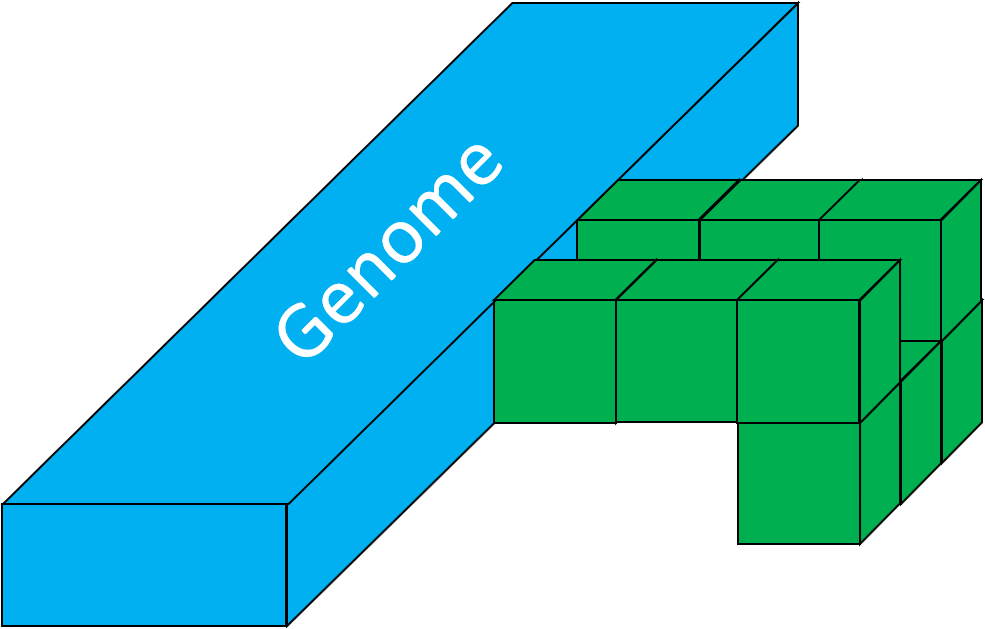}
\caption{Before piping grows}
\end{subfigure}
\begin{subfigure}[b]{0.52\linewidth}
\centering
\includegraphics[width=0.5\linewidth]{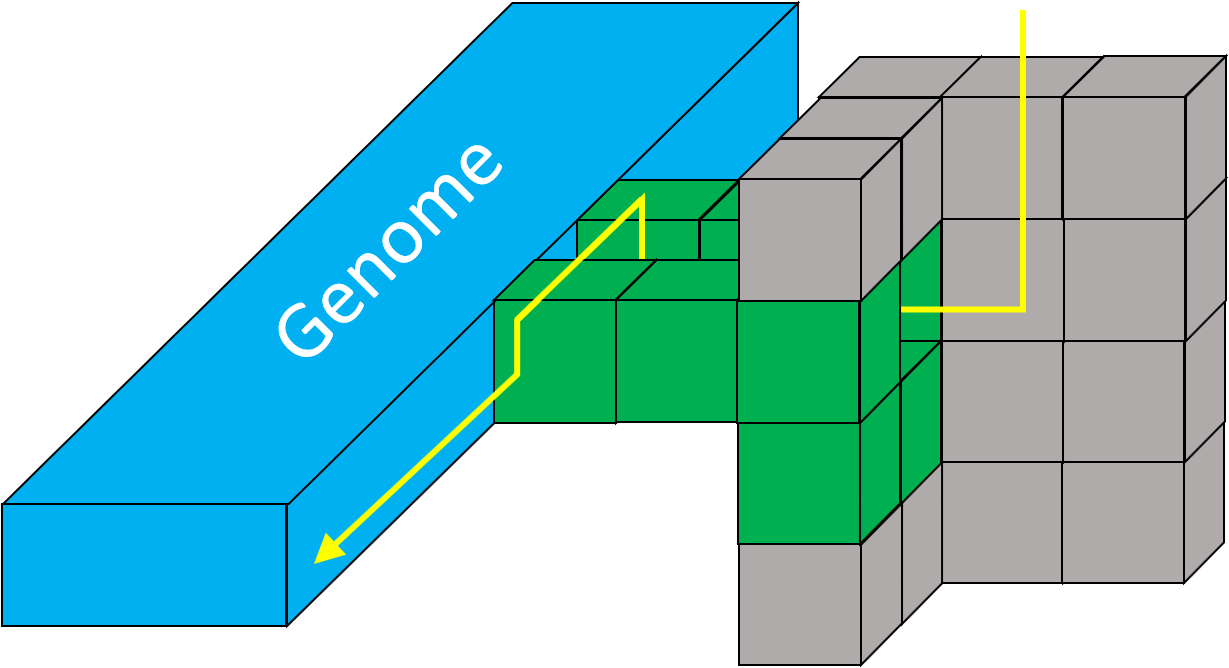}
\caption{After piping grows}
\end{subfigure}
\caption{The interception gadget works by growing out from an incoming genome path into the path of the future pipe that will extend in the same direction. Once the pipe grows, the original sequence of tile placements is blocked, and it instead cooperates with the interception gadget to continue growing but with a hole in the side toward the original genome path. Then, whenever the macrotile differentiates and tiles start growing through the piping, a signal can grow out of the hole and along the genome to activate its (and the corresponding external communication datapath's) propagation into the neighboring macrotile in the ``up'' direction. The path of this signal is shown in yellow.}
\label{fig:intercept}
\end{figure}

\subsection{Outgoing genome and external communication} The next step in our implementation is initiated whenever a guide rail encoding some tile type wins the bracket module. Once this happens, it grows into the external communication module, which grows datapaths in the direction of the neighboring macrotiles, and also tells the genome to do the same. Our first adjustment is encoding a ``variable'' instruction (from the original construction) into both datapaths such that they stop immediately before the boundary with the neighboring macrotile and wait for a signal to continue growing. Additionally, we augment the signal that tells the genome to start propagating so that it also initiates an added set of instructions in the initialization section of the genome that will begin the tiling process of the blocking protocol.

\subsection{Receiver and piping} Once activated, the tiling instructions in the genome will start by growing to the upper eastern corner of the north face and placing a gadget which we refer to as the $\texttt{receiver}$. This gadget just waits for the boundary tiling to complete, similar again to the ``variable'' instruction. Once placed the datapath will then grow back down until it is a constant number of units above the center of the bottom face, i.e. just above the center of $\texttt{boundary}_D$. Here, it will grow the pipe intersection. Next, a pipe of constant length $\texttt{pipe}_D$ will grow down to the bottom face, along with an encoding of the scale factor of the macrotile $m$. Along $\texttt{boundary}_D$, this encoding will be used to seed counters in all four directions that will grow a distance of $\frac{m}{2}$ to the edges of the macrotile, i.e. the edges created by $\texttt{boundary}_D \times \texttt{boundary}_{side}$. In the space between the counters, filler tiles will attach by cooperation. One quadrant, however, will use special filler tiles that grow only a constant hardcoded number of steps. This will allow the outgoing genome and external communication to grow to the boundary first, and then cooperation will allow the filler tiles to continue, thereby rectifying a potential timing dependency. In the opposite quadrant, the same type of cooperation is used to allow the filler tiles to grow around an incoming genome and external communication, although no special tile is used here. This way, if the incoming datapaths aren't present yet (as they possibly never will be) the filler tiles can fill out that space anyways, thereby blocking the datapaths if they ever do come in. Continuing on, the final row of each counter will initiate special signals in both directions perpendicular to the direction of growth, causing special tiles to be placed at the four bottom corners of the macrotile. The layout of the counters and special signals on the bottom face of the macrotile is illustrated in Figure~\ref{fig:bottom_face}.

Additionally, from the pipe intersection, a $\texttt{pipe}_{side}$ will grow out for each $side$ direction. It will grow independently for a small constant number of steps, and then will be tethered to the counter on $\texttt{boundary}_D$ by a small, constant-width additional strip of tiles so that it grows right up to $\texttt{boundary}_{side}$ without overgrowing it.

\begin{figure}
\centering
\includegraphics[width=0.5\linewidth]{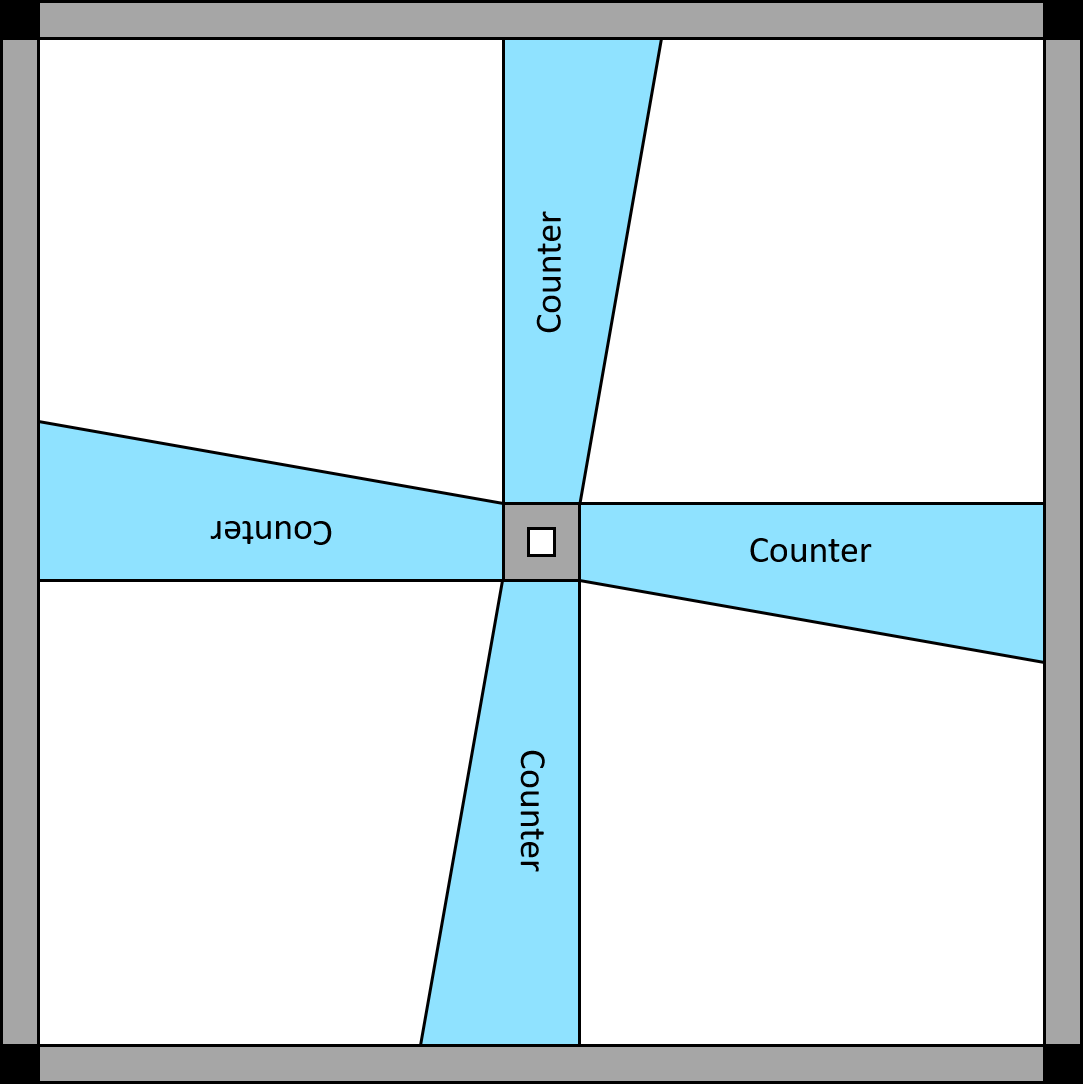}
\caption{For this example, we make it so that the scale factor only needs to be encoded into the bottom face so that it knows how far to grow outwards to get to each side face. Each pipe that grows outwards in the N,E,S,W directions is tethered to the bottom face so that it also grows exactly up to each side face. Each black tile in the corner then grows up a pole (as seen on the left in Figure~\ref{fig:bottom_strip}) that begins the growth of the bottom strip of each side face.}
\label{fig:bottom_face}
\end{figure}

\subsection{Growing the side faces} Each of the special tiles at the four bottom corners then grow up a constant height pole. Each pole initiates growth of a constant width strip along the bottom of each $\texttt{boundary}_{side}$, as seen in Figure~\ref{fig:bottom_strip}. These strips must be careful to keep growing whether the incoming external communications and genome have grown through or not. Again, this is done by allowing the filler tiles to grow independently, or by letting them cooperate with special glues on the datapaths if the datapaths are present and block the filler tiles. Once the filler tiles make it to the middle of the boundary, as designated by a special tile placed at the end of the counters, the tiles will wait (enforced by cooperation) for the corresponding $\texttt{pipe}_{side}$ to grow up to the boundary. Once present, the tiles can grow over the end of the piping and continue on the other side. However, they must wait again for the outgoing external communication and genome. This waiting is enforced by using a constant number of hardcoded tiles after passing the piping that must then cooperate with the tiles of the outgoing information in order to continue on. Once the two datapaths have also come in, the tiles can continue all the way to the other edge. Since the cooperation will happen left-to-right and bottom-to-top, we can enforce that the upper right most tile will be the last placed in this strip.

Of the poles that initiated these bottom strips, the north western pole will also start growing a path around the top of all four strips (in the opposite direction that the strips themselves are growing). This path will also use generic tiles that turn around the corners and run along the top of these strips. Once the path has grown all the way around one loop, it will lift up one unit in the ``up'' direction and continue around again and again in a spiral assembly pattern. Since the first loop relies on cooperation with the bottom strips, we know that all four bottom strips must have completed by the time the path makes its first full loop. The spiraling path should make the pattern seen in Figure~\ref{fig:side_face} on the north face. The final path around the macrotile in the uppermost ring should be calculated to begin right next to where the $\texttt{receiver}$ was placed in the previous steps. This will cause the path to loop around and hit the $\texttt{receiver}$, signaling that the bottom face and all four side faces are completely tiled, and the macrotile is ready for differentiation.

\begin{figure}
\centering
\includegraphics[width=\linewidth]{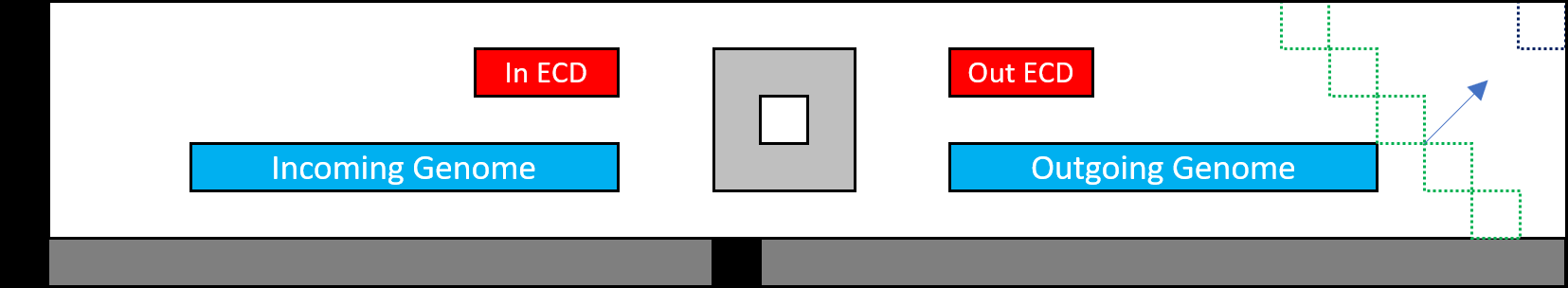}
\caption{The bottom strip of each side boundary $\texttt{boundary}_{side}$. Growth starts from the special tiles on the left. Generic filler tiles grow either (a) around the incoming genome and external communication or (b) over the slots to block the incoming datapaths from coming in later. The special tile in the middle of the bottom connects to the end of the piping (thereby preventing the filler tiles from growing over it). A constant number of hardcoded tiles then count over to the slots designated for the outgoing genome and external communication. The filler tiles must then wait for these datapaths to come in before they can continue on. The tile in the upper-rightmost corner is guaranteed to be the last placed.}
\label{fig:bottom_strip}
\end{figure}

\begin{figure}
\centering
\includegraphics[width=0.45\linewidth]{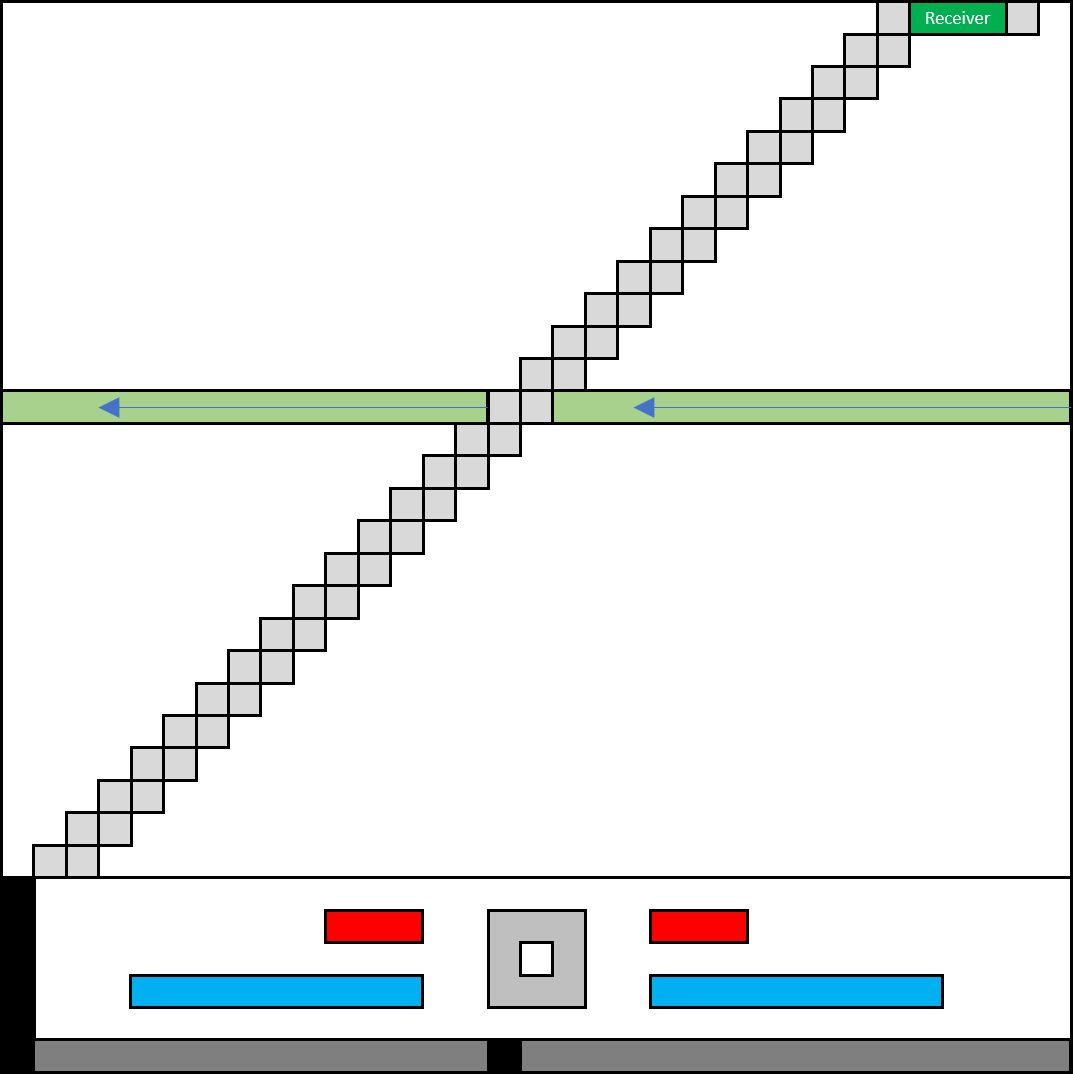}
\caption{Once the bottom strip has of each side face has completely filled out, one of the corner tiles starts the growth of the rest of the side faces by growing a one tile wide path around the top of each bottom strip. Each bottom strip can be certain to have completely by the time one loop has been made. Then, the path steps one in the ``up'' direction (through cooperation with the first tile in the path) and starts another path. This continue until the path reaches the $\texttt{receiver}$ in the upper right corner of the side face on which the first path started. Once this $\texttt{receiver}$ is hit, all the blocking protocol is certain to have completed, and the macrotile is clear to differentiate.}
\label{fig:side_face}
\end{figure}

\subsection{Differentiation and activation} Once the $\texttt{receiver}$ has been signaled, a path of tiles will grow back to the pipe intersection and into the $\texttt{loc}_{center}$ location from the ``up'' direction. Placing this tile signifies that the macrotile has officially differentiated under our representation function $R$. From here, for any direction $dir$ with non-intercepted pipes, if the neighboring macrotile hasn't differentiated, then tiles will attach within $\texttt{pipe}_{dir}$ until they come out in the neighboring macrotile. Then, these tiles will activate the halted outgoing external communication and genome datapaths to continue growing into the neighboring macrotile. As for any $\texttt{pipe}_{dir'}$ that was intercepted, this means that the neighboring macrotile has already differentiated, and tile don't need to / won't attach within $\texttt{pipe}_{dir'}$ past the $\texttt{intercept}_{dir'}$ gadget. However, the hole that $\texttt{intercept}_{dir'}$ leaves in $\texttt{pipe}_{dir'}$ allows for diffusion in the ``up'' direction. If the macrotile in the ``up'' direction has not already differentiated, tiles can therefore diffuse in from above and will grow through $\texttt{pipe}_{dir'}$, out of $\texttt{intercept}_{dir'}$, and along the genome to signal that it (and the ``up'' external communication datapath) can continue growing in the ``up'' direction. If the neighbor in the ``up'' direction has already differentiated, then there is no need for these to continue growing in that direction, and the currently growing macrotile will become a constrained subspace anyways, not allowing new tiles to attach regardless.

The tiled boundary of neighboring macrotiles can prevent external communication datapaths from correctly propagating into non-differentiated macrotiles. However, because the boundary tiling only occurs after a guide rail has left the bracket, the external communication datapath would be unable to affect the tile type that the neighboring macrotile would eventually differentiate into. In other words, the macrotile has already essentially chosen a representative tile type, meaning new incoming datapaths don't necessarily have to reach the macrotile's genome for the simulation to progress correctly.

\subsection{Seed and representation functions} In addition to the augmentations already discussed, the universal simulator also requires slightly tweaked seed and representation functions generator $S$ and $\mathcal{R}$. In our implementation, $S$ must generate the seed macrotiles such that the genome includes the new instructions previously mentioned. This covers the instructions to grow the interception gadgets when propagating into a new macrotile, the special glues on the perimeter when passing through the boundary of a new macrotile, the initialization instructions to place the $\texttt{receiver}$ once the bracket has finished, etc. The representation function $R$ has to be augmented (for every output of the $\mathcal{R}$ function) to ensure that macrotiles with an open pipe intersection map to empty space (regardless of bracket state) and macrotiles with a blocked pipe intersection (and processed bracket) map to a tile type in the simulated system.

\begin{figure}
\centering
\includegraphics[width=\linewidth]{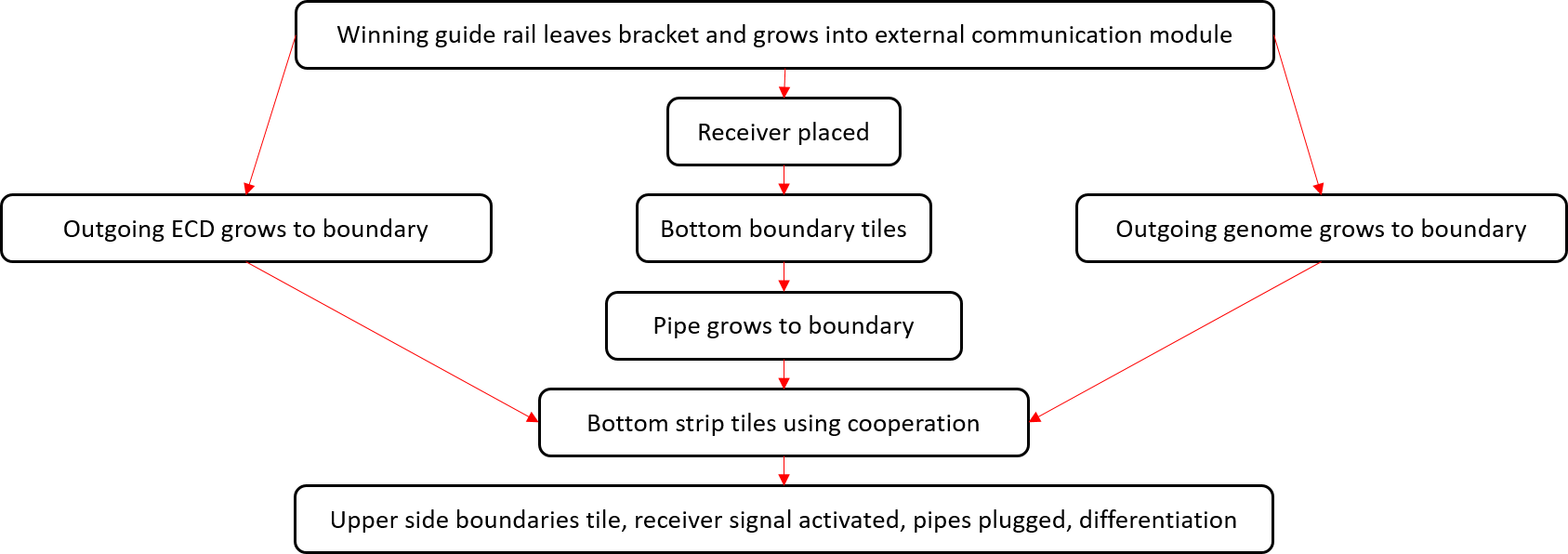}
\caption{A diagram depicting the timing dependencies present in our implementation. Once the bottom strips are tiled, all the timing dependencies are rectified and the rest of the process is sequential.}
\label{fig:timing}
\end{figure}

Overall, the main concern of our implementation is making sure the timing of the components is correct. This is important because, if certain components aren't completed when the macrotile differentiates, this may leave unintended diffusion paths open. A summary of the timing dependencies that are critical to the correctness of our augmented construction are shown in Figure~\ref{fig:timing}.

Now we give full proofs of the two lemmas mentioned in Section~\ref{sec:spatial-atam}.

\begin{lemma} \label{lem:diffusion_vs_differentiated}
Given a macrotile $L$ and two directions $dir_A,dir_B \in \{N,E,S,W,U,D\}$ such that $dir_A \neq dir_B$ and the neighboring macrotiles in those directions have not already differentiated, a diffusion path from a tile location in the neighboring macrotile $A$ in direction $dir_A$, through only macrotile $L$ (and no others), to a tile location in the neighboring macrotile $B$ in direction $dir_B$ will exist if and only if macrotile $L$ has not already differentiated.
\end{lemma}

\begin{proof}
For the proof, we will instead focus on two conditional statements that combined are logically equivalent to the Lemma~\ref{lem:diffusion_vs_differentiated}.

The first statement we will prove is, ``If macrotile $L$ has not already differentiated, then there is a diffusion path from $A$ to $B$ through $L$.'' The most constrained the problem can be while the premise is still true is if macrotile $L$ is one tile away from differentiating. By proving the diffusion paths still exist in this situation, we prove they exist if macrotile $L$ is earlier in the differentiation process. Now, we can break the problem down into two cases. Let's start with $dir_A, dir_B \in \{N,E,S,W,D\}$. In this simple case, the diffusion path through $L$ from $A$ to $B$ is just in the tube in the $dir_A$ direction, through $\texttt{loc}_{center}$, and out the tube in the $dir_B$ direction. The more complicated case is when, without loss of generality, $dir_A = U$ and $dir_B \in \{N,E,S,W,D\}$. In this case, we know that some neighbor in a direction $C \in \{N,E,S,W,D\} \setminus dir_B$ has already differentiated (in order for $L$ to have initiated the differentiation process). Therefore, $\texttt{intercept}_C$ was present to ``cut'' a hole in $\texttt{pipe}_C$. With this, our path is now moving from macrotile $A$, through the open top to macrotile $L$, into $\texttt{pipe}_C$ through $\texttt{intercept}_C$, through the $\texttt{loc}_{center}$, and out $\texttt{pipe}_B$.

The other statement we will prove is, ``If macrotile $L$ has already differentiated, then there is no diffusion path from $A$ to $B$ through $L$.'' Doing the opposite of the last claim, we will now look at the least constrained the problem can be while the premise is still true, right after the tile in $\texttt{loc}_{center}$ of macrotile $L$ has attached. Again, we can break the problem down into two cases. First, when $dir_A, dir_B \in \{N,E,S,W,D\}$, the pipes in these two directions are guaranteed to not have been intercepted, since neither $A$ nor $B$ has differentiated in the premise of the lemma. Since the pipe intersection is also blocked and there are no other breaks in the tiling of the side and bottom faces, there is no diffusion path. In the more complicated case, when $dir_A = U$ and $dir_B \in \{N,E,S,W,D\}$, the path can start through the open top of macrotile $L$. While some pipes must have necessarily been intercepted, we know by the premise of the lemma that the pipe in direction $dir_B$ was not. Since that pipe is no longer connected to the intercepted pipes due to the tile attachment in $\texttt{loc}_{center}$, there is no diffusion path.

Together, these two claims prove Lemma~\ref{lem:diffusion_vs_differentiated}.
\end{proof}

\begin{lemma} \label{lem:scaled_spatial}
Given an empty tile location $l$ in the simulated system $\calT$ and the corresponding macrotile $L$ in the simulator $\mathcal{U}$, tiles can attach within macrotile $L$ in the simulator if and only if tile location $l$ is not in a constrained subspace in the simulated system.
\end{lemma}

\begin{proof}
We again break the problem into two claims that together are logically equivalent to Lemma~\ref{lem:scaled_spatial}.

First, we will prove, ``If tile location $l$ is not constrained, then tiles can attach within macrotile $L$.'' By the premise, there must be a diffusion path in the simulated system $\calT$ from infinitely far away to the tile location $l$. Therefore, there must be a series of macrotiles in the simulator $\mathcal{U}$ from infinitely far away to the macrotile $L$. Since these macrotiles all must map to empty space under the representation function $R$, none of them must have already differentiated. By Lemma~\ref{lem:diffusion_vs_differentiated}, we know that each of these non-differentiated macrotiles has a path through it connecting all 15 pairs of directions. These mini-diffusion paths are guaranteed to lined up, since the pipes are designed to line up between adjacent macrotiles, and can therefore be linked together to comprise a longer path. Therefore, there must be a diffusion path in the simulator $\mathcal{U}$, comprised of these mini-diffusion paths through each macrotile concatenated together, from infinitely far away to any empty location within the macrotile $L$, thereby allowing tiles to attach in $L$.

Next, we will prove, ``If tile location $l$ is constrained, then tiles cannot attach within macrotile $L$.'' By the premise, the tile location $l$ is within a constrained subspace. By definition, there must be a constraining subassembly surrounding tile location $l$. In the simulator $\mathcal{U}$, all of the tiles from this constraining subassembly must be represented by already differentiated macrotiles. Since a diffusion path in the simulator $\mathcal{U}$ from infinity to the macrotile $L$ must pass through one of the macrotiles representing a tile in the constraining subassembly, and since we know that already differentiated macrotiles have no diffusion paths between neighbors in any two different directions by Lemma~\ref{lem:diffusion_vs_differentiated}, then it must necessarily be the case that no diffusion path from infinity to the macrotile $L$ exists, and tiles can therefore not attach within the macrotile $L$.

Together, these two claims prove Lemma~\ref{lem:scaled_spatial}.
\end{proof}

Given the ability to soundly simulate 3D aTAM dynamics as well as the spatial constraint, this construction is a valid universal simulator for the Spatial aTAM, thereby proving Theorem~\ref{thm:spatial-IU}.

\bibliographystyle{plainurl}
\bibliography{tam,experimental_refs,complexity}

\end{document}